\documentclass[
aps,pra,
reprint,
superscriptaddress,
twocolumn, 
nofootinbib,
floats,floatfix
]{revtex4-2}

\usepackage{amsmath}
\usepackage{pifont}
\usepackage{amssymb}
\usepackage{amsthm}
\usepackage{graphicx}
\usepackage{xcolor}
\usepackage[english]{babel}
\usepackage{csquotes}
\usepackage{braket}
\usepackage{mathtools}
\usepackage{enumitem}
\usepackage{IEEEtrantools}
\usepackage{relsize}
\usepackage{placeins}
\usepackage{comment}
\usepackage[normalem]{ulem}
\usepackage{systeme}
\usepackage{physics}
\usepackage{listings}
\usepackage{xcolor}
\usepackage{tcolorbox}
\usepackage[ruled,vlined]{algorithm2e}
\usepackage{lipsum}

\renewcommand{\vec}[1]{\boldsymbol{#1}}  
\newcommand{\vtheta}{\Vec{\theta}}

\newcommand{\nparams}{\ensuremath{L} }

\newcommand{\uni}{\boldsymbol{\mathcal{D}}_{\nparams}} 
\newcommand{\vol}{\boldsymbol{\mathcal{V}}_{\nparams}}

\newcommand{\OC}{\mathcal{O}}

\newcommand{\Var}{{\rm Var}}

\renewcommand{\geq}{\geqslant}
\renewcommand{\leq}{\leqslant}

\newcommand{\poly}{\operatorname{poly}}

\newtheorem{theorem}{Theorem}
\newtheorem{lemma}[theorem]{Lemma}
\newtheorem{proposition}[theorem]{Proposition}
\newtheorem{definition}[theorem]{Definition}
\newtheorem{remark}[theorem]{Remark}

\renewcommand{\phi}{\varphi}

\DeclareMathOperator{\prob}{\mathbb{P}}
\DeclareMathOperator{\exval}{\mathbb{E}}
\DeclareMathOperator{\Cov}{Cov}

\newcommand{\mathspace}{\quad}
\newcommand{\phaselessprod}{\odot}

\newcommand{\circuitDepth}{L}
\newcommand{\probContribute}[1]{P(#1)}
\newcommand{\nModes}{N}
\newcommand{\middlePoint}{m}
\newcommand{\majorana}{M}
\newcommand{\majoranaP}{M'}
\newcommand{\majoranaOther}{\mathcal{M}}
\newcommand{\majoranaOtherP}{\mathcal{M}'}

\newcommand{\majoranaGate}[1]{G_{#1}}
\newcommand{\majoranaGateNoParam}{G}
\newcommand{\majoranaGateNoParamP}{G'}
\newcommand{\partialNorm}[2]{\text{\rm¦¦}#1\text{\rm¦¦}_{(#2)}}
\newcommand{\fockState}{\vec{n}}

\renewcommand{\[}{\begin{equation}}
\renewcommand{\]}{\end{equation}}

\makeatletter 
\renewcommand\onecolumngrid{
\do@columngrid{one}{\@ne}
\def\set@footnotewidth{\onecolumngrid}
\def\footnoterule{\kern-6pt\hrule width 1.5in\kern6pt}
}
\renewcommand\twocolumngrid{
        \def\footnoterule{
        \dimen@\skip\footins\divide\dimen@\thr@@
        \kern-\dimen@\hrule width.5in\kern\dimen@}
        \do@columngrid{mlt}{\tw@}
}
\makeatother
\usepackage{minitoc}
\usepackage{tocloft}
\usepackage[colorlinks=true,linkcolor=blue,citecolor=blue,urlcolor=blue]{hyperref}

\begin{document}

\makeatletter
\makeatother
 
\doparttoc 
\faketableofcontents 

\title{Simulation of Fermionic circuits using Majorana Propagation}

\author{Aaron Miller}
\thanks{The first two authors contributed equally to this work.}
\author{Joachim Favre}
\thanks{The first two authors contributed equally to this work.}
\author{Zo\"e Holmes}
\author{\"Ozlem Salehi}
\author{Rahul Chakraborty}
\author{Anton Nyk\"anen}
\author{Zolt\'an Zimbor\'as}
\author{Adam Glos}
\author{Guillermo Garc\'ia-P\'erez}
\email{guille@algorithmiq.fi}
\affiliation{Algorithmiq Ltd, Kanavakatu 3 C, FI-00160 Helsinki, Finland}

\date{\today}

\begin{abstract}
    We introduce Majorana Propagation, an algorithmic framework for the classical simulation of Fermionic circuits.
    Inspired by Pauli Propagation, Majorana Propagation operates by applying successive truncations throughout the Heisenberg evolution of the observable.
    We identify \textit{monomial length} as an effective truncation strategy for typical, unstructured circuits by proving that high-length Majorana monomials are exponentially unlikely to contribute to expectation values and the backflow of high-length monomials to lower-length monomials is quadratically suppressed.
    We provide performance guarantees by proving analytically that approximation errors decrease exponentially with the truncation threshold and that only polynomial resources are required to compute the expectation value of observables up to a fixed error for an ensemble of circuits relevant to quantum chemistry.
    Majorana Propagation can be used either independently, or in conjunction with quantum hardware, to simulate Fermionic systems relevant to quantum chemistry and condensed matter. 
    We exemplify this by using Majorana Propagation to find circuits that approximate ground states for strongly correlated systems of up to 52 Fermionic modes.
    Our results indicate that Majorana Propagation is orders of magnitude faster and more accurate than state-of-the-art tensor-network-based circuit simulators.
\end{abstract}

\maketitle

Simulating many-body quantum systems is computationally challenging, despite its fundamental importance across various fields, including condensed-matter physics and quantum chemistry.
Quantum computing offers new perspectives, powered by processors that exhibit quantum phenomena and are therefore not limited by them.
Yet, operating quantum devices for practical purposes can benefit enormously from accurate classical circuit simulation algorithms~\cite{preskill2018quantum}.
These can be used, among other things, to produce good reference states for quantum simulation tasks, offloading costly circuit optimization routines to a classical machine.

In this context, Pauli Propagation (PP) has recently emerged as a new promising method for spin systems~\cite{rall2019simulation, aharonov2022polynomial, rakovszky2022dissipation, beguvsic2023simulating, fontana2023classical, shao2023simulating, rudolph2023classical, schuster2024polynomial, angrisani2024classically, gonzalez2024pauli, lerch2024efficient, cirstoiu2024fourier, angrisani2025simulating, fuller2025improved, rudolph2025pauli, angrisani2025simulating}. 
By leveraging typical features in the dynamics of Pauli strings evolving through a circuit, PP disregards terms that are expected to have a negligible contribution to the expectation value one is simulating, resulting in guaranteed accuracy for almost all circuits within a vast class.
PP thus presents itself as a very natural way to simulate expectation values of spin circuits.
Despite its recent conception, PP can already compete with well-established methods, such as tensor networks and neural network states.
Furthermore, it allows for fast circuit optimization in variational scenarios thanks to surrogate simulation~\cite{fontana2023classical}, and can be readily interfaced with quantum hardware~\cite{lerch2024efficient, fuller2025improved}, as it operates at the circuit level.

However, the applicability of PP to other many-body systems, such as Fermionic ones, remains an open question.
In particular, many quantum systems are mathematically described by algebras that do not have an equivalent to Pauli weight locality, a quantity central to guarantees for PP for typical circuits.
While these systems are mappable to qubit space, the resulting Pauli operators are typically highly non-local.

Here we present Majorana Propagation (MP), a PP-inspired classical simulation algorithm that can be naturally used for Fermionic systems.
Differently from PP, MP does not operate at the level of Pauli weight; in fact, it does not even assume the decomposability of the algebra describing the system into tensor products of sub-system operators.
Instead, MP leverages the typical behaviour of monomials spanning operator bases. 
We use this observation to prove that Majorana Propagation can efficiently simulate \textit{typical} Fermionic circuits.
Crucially, the Fermionic basis is effectively built `on the fly' and so is not restricted to simulating free Fermion models but rather can be used to simulate systems living in exponentially large Lie algebras. Moreover, unlike tensor network methods or near-Clifford methods, our approach is not intrinsically targeting low entanglement or magic growth scenarios.

A particularly promising application of MP is for optimizing circuits to prepare approximate ground states.
Established quantum algorithms for energetic structure calculations, including near-term~\cite{yu2025quantum,yeter2020practical,yeter2021benchmarking, larose2019variational,parrish2019quantum,stair2020multireference,kokail2019self,google2020hartree,xiaoyue2024strategies,cerezo2021variational} and fault-tolerant~\cite{nielsen2000quantum,temme2011quantum,kitaev1995quantum,brassard2002quantum,tan2020quantum,Low2019hamiltonian,gilyen2019quantum,ge2019faster,motta2020determining,motlagh2024ground,gluza2024double, suzuki2025double} approaches, need to be initialized with good reference states, as their performance depends on the overlap between the reference and the ground state.
How to find these remains an open question~\cite{lee2023evaluating, zimboras2025myths}, and heuristic variational ground state approximation methods are likely to play an important role.

Many of the most successful variational methods for ansatz generation construct Fermionic circuits that are then appropriately mapped to qubit space for hardware execution~\cite{miller2023bonsai,miller2024treespilation, neven_mapping, jw_mapping, bk_mapping, chiew2023discovering, chien2022optimizing}.
Operating in Fermion space has proved numerous advantages, including the facilitation of imposing known symmetries of the simulated system.
However, historically, many of these methods also rely on the quantum device itself to carry out the optimization, which results in costly and unreliable feedback loops between a classical and a quantum computer~\cite{mcclean2017hybrid, franca2020limitations, anschuetz2022beyond, larocca2024quantum, cerezo2023does, zimboras2025myths}.
Here we show that MP can be used to classically simulate and optimize adaptive Fermionic circuit ansatze to find circuits for preparing approximate ground states of strongly correlated molecules~\cite{grimsley2019adaptive}.

Beyond ground states, MP may find use both as a stand-alone classical algorithm, and as a tool to be used with quantum hardware, to study Fermionic systems. Potential applications range from simulating the dynamics of Fermionic systems~\cite{lerch2024efficient, fuller2025improved} to tackling classification problems with quantum data~\cite{cerezo2023does, bermejo2024quantum}.

\medskip

\paragraph*{Background.}
An $N$-mode Fermionic system in second quantization can be described in terms of $N$ creation operators $\{ a_i^\dagger \}_{i = 1}^{N}$ and annihilation operators $\{ a_i \}_{i = 1}^{N}$ that satisfy the canonical anticommutation relations $\{ a_i, a_j \} = \{ a_i^\dagger, a_j^\dagger \} = 0$, and $\quad \{ a_i^\dagger, a_j \} = \delta_{ij}$. These operators act on the
Fermionic Fock space $\mathcal{F}(\mathbb{C}^N)\cong \oplus^{N}_{k=0} \wedge^{k} \mathbb{C}^N $, which is a $2^N$-dimensional Hilbert space spanned by the so-called Fock basis.
In this space, the Fermionic vacuum $\ket{\text{vac}_\text{f}}$ is the unique vector such that $a_j \ket{\text{vac}_\text{f}} = 0$ for all $j = 1, \dots, N$.
The remaining Fock basis elements can be constructed by considering all possible combinations of occupation numbers $n_j \in \{0, 1\}$ as $\ket{n_1 n_2 \dots n_{N}} \coloneqq \prod_{j=1}^{N} (a_j^\dagger)^{n_j} \ket{\text{vac}_\text{f}}$.

It is also common to define the Fermionic operator space with so-called Majorana operators $\{ m_k \}_{k = 1}^{2N}$ as $m_{2 j - 1} \coloneqq a_{j}^\dagger + a_{j}$ and $m_{2 j} \coloneqq \mathrm i (a_{j}^\dagger - a_{j})$.
Such operators obey many convenient properties, such as being unitary and self-adjoint, and satisfy simpler anti-commutation relations:
\begin{equation}\label{eq:majorana_properties}
    m_i^\dagger=m_i,\quad\{m_i,m_j\}=2\delta_{ij}.
\end{equation}
These $2N$ Majorana operators are algebraically independent, that is, no two distinct products of them (ordered according to the indices)  result in the same operator.
Hence, up to a sign, each unique product of Majorana operators can be associated with a $2N$-dimensional binary vector $\vec{b} = (b_1, \ldots, b_{2N}), b_i \in \{0, 1\}$ through the expression
\begin{equation}
     M_{\vec{b}} = \mathrm i^{r_{\vec{b}}}  m_1^{b_1} m_2^{b_2} \cdots m_{2N}^{b_{2N}} \, .
\end{equation}
The imaginary unit $\mathrm i$ is included to ensure that  $ M_{\vec{b}}$ is Hermitian\footnote{The $\mathrm i$ is necessary whenever inverting the product sequence involves an odd number of permutations. That is, the value of $r_{\vec{b}}$ depends on the length $w = \norm{\vec{b}}_1$ of the monomial.
As shown in Appendix~\ref{app:prelims}, if either $w$ or $w - 1$ is a multiple of 4, $r = 0$.
Otherwise, $r = 1$.}.
The $4^N$ operators $\{ M_{\vec{b}} \}_{\vec{b}}$
form a basis of the space of linear operators in the Hilbert space, thus any operator can be uniquely decomposed as a linear combination of them. We will refer to these operators as \textit{Majorana monomials}.
The length of a Majorana monomial, defined as the 1-norm of vector $\vec{b}$, $\norm{\vec{b}}_1 = \sum_{i = 1}^{2N} b_i$, will play a crucial role in this work.
We will use the shorthand $w$-monomial to refer to length-$w$ Majorana monomials.

Majorana monomials share some algebraic similarities with Pauli strings: They are traceless (except for the length-0 monomial, i.e., the identity operator). Any two of them either commute or anti-commute, and they square up to identity $M_{\vec{b}}^2 = \mathbb{I}$.
However, unlike Pauli strings, they are not decomposable as a tensor product of $N$ single-mode operators, as Majorana operators are non-local.
In fact, under fermion-to-qubit transformations, Majorana monomials are generally mapped to non-local Pauli strings~\cite{miller2023bonsai}, with Pauli weights spanning a wide range, even for fixed Majorana monomial length.
It follows that previous guarantees for PP with weight truncation do not apply to Fermionic systems, as Pauli weight does not capture typical operator dynamics.

\medskip

\paragraph*{Majorana monomial propagation.}

We consider the task of simulating expectation values of observables. That is, we want to compute
\begin{equation}
    \label{eq:mean_value}
    f_L(\vec{\theta}) =  \langle H \rangle_{\vec{\theta}} = {\rm Tr}\left[ U_L(\vec{\theta}) \varrho U_L(\vec{\theta})^\dagger H \right]
\end{equation}
given a unitary Fermionic circuit $U(\vec{\theta})$, an initial state $\varrho$, and an observable $H$. The initial state is typically a Fock basis state, $\varrho = \ketbra{n_1 n_2 \dots n_{N}}{n_1 n_2 \dots n_{N}}$. We further suppose the Fermionic circuit,
\begin{equation}\label{eq:unitarycircuit}
U_L(\vec{\theta}) = \prod_{j=1}^L e^{- \mathrm i \theta_j M_{{\vec{b}_j}}/2} \, ,
\end{equation}
consists of a sequence of $L$ Fermionic gates where each gate is generated by a Majorana monomial generator $M_{\vec{b}_j}$ and the $\theta_j$ are real parameters. 
Finally we assume that the observable $H$ has a known expansion into Majorana monomials $H = \sum_{\vec{b}} \alpha_{\vec{b}} M_{\vec{b}} $. Observables $H$ of physical interest are spanned by Majorana monomials of low, even length.
For instance, most Hamiltonians in quantum chemistry and materials contain only length-2 and length-4 terms and thus the number of terms in the expansion will typically scale at worst polynomially in $N$. 

\begin{figure}[t!]
    \centering
    \includegraphics[width=0.9\linewidth]{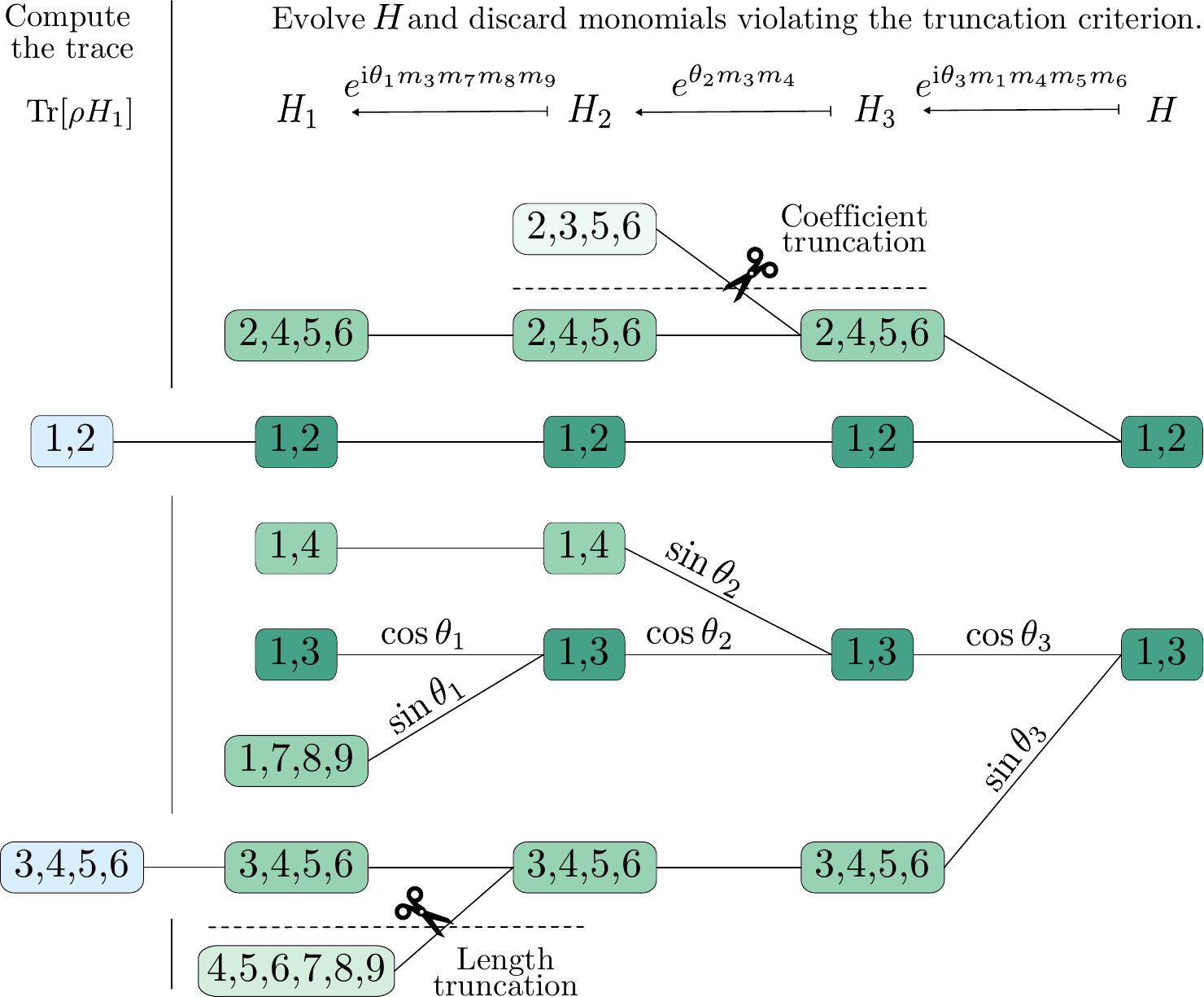}
    \caption{\textbf{Illustration of Majorana Propagation.}
    The observable $H$ is spanned by two length-2 terms, $\mathrm i m_1 m_2$ and $\mathrm i m_1 m_3$.
    The first gate, generated by $m_1 m_4 m_5 m_6$ anti-commutes with both terms, leading to branching into higher length monomials.
    The intensity of the colour represents the coefficient accompanying each monomial.
    Assuming small angles, the sine branches are paler than their precursors.
    The figure further depicts two different truncation events, one where a length-4 monomial is eliminated because of its coefficient reaching a very low value, and a second one where a length-6 monomial exceeds the length threshold $w^{\star} = 4$.
    In the final computation of the trace, only two of the surviving monomials are paired and thus contribute to the estimate.}
    \label{fig:sketch}
\end{figure}

Our proposed Majorana Propagation algorithm, schematically depicted in Fig.~\ref{fig:sketch}, works naturally in the Heisenberg picture. We take each monomial in the expansion of $H$ and back-propagate it under each of the gates in the circuit. 
The evolution of monomials under Fermionic gates takes a conveniently compact form. Namely, in direct analogy with the Pauli case, if the propagated monomial $M_{\vec{b}}$ commutes with the generator of rotation $M_{{\vec{b}_j}}$ then the monomial is left unchanged but otherwise it branches into the sum of two monomials with a cosine or sine coefficient depending on the rotation angles: 
\begin{equation}\label{eq:gatebranch}
   M_{\vec{b}} \xrightarrow[]{e^{- \mathrm i\theta_j M_{{\vec{b}_j}}/2}}  \cos (\theta_j) M_{\vec{b}} + \sin(\theta_j)  \mathrm i M_{{\vec{b}_j}} M_{\vec{b}} \, .  
\end{equation}
Applying Eq.~\eqref{eq:gatebranch} iteratively on each of the gates in Eq.~\eqref{eq:unitarycircuit}, starting with the rotation under $M_{{\vec{b}_L}}$, we end up with a new sum of Majorana monomials $\sum_{\vec{b}} c_{\vec{b}}\left(\vec{\theta} \right) M_{\vec{b}}$ where $c_{\vec{b}}\left(\vec{\theta} \right)$ are the coefficients of the back-propagated monomials $M_{\vec{b}}$. These coefficients capture both the initial length of each of the relevant monomial in the target observable $H$ and the sine and cosine coefficients that have been picked up during the propagation. 
Finally, we compute the overlap of these monomials with the initial state $\varrho$ to compute the expectation value as  
\begin{equation}\label{eq:expectationexpression}
    f\left(\vec{\theta}\right) = \sum_{j} c_{\vec{b}}\left(\vec{\theta} \right) \Tr[ \varrho M_{\vec{b}}] \, . 
\end{equation}
If the rotation angles for the circuit one wishes to simulate are known in advance then the sine and cosine terms take explicit numerical values and the $c_{\vec{b}}\left(\vec{\theta} \right)$ coefficients can be stored numerically. However, if the values of $\vec{\theta}$ are not known in advance, or one is interested in computing the expectation value in Eq.~\eqref{eq:mean_value} for a wide range of $\vec{\theta}$ values, one can instead symbolically represent the coefficients $c_{\vec{b}}\left(\vec{\theta} \right) $. Such \textit{surrogate} models have a substantial time and memory overhead, but after this pre-processing step, enable rapid re-evaluation of the expectation value~\cite{fontana2023classical, rudolph2023classical, lerch2024efficient}. 

While we here focus on unitary Fermionic circuits that have been compiled into products of monomial rotations, as per Eq.~\eqref{eq:unitarycircuit}, we stress that Majorana Propagation can be applied more generally to compute expectation values for circuits composed of a sequence of $L$ linear maps $\mathcal{C} = C_L \circ \cdots \circ C_1$. In that case, one simply computes an analogous expression to Eq.~\eqref{eq:gatebranch} for the action of each $C_j$ on an arbitrary Majorana monomial. 

In its raw form, the number of monomials in the final expression for the expectation value, Eq.~\eqref{eq:expectationexpression}, will blow up exponentially with the number of gates in the circuit that induce a splitting. We therefore employ \textit{truncation} schemes to drastically reduce the computational resources. There are two main families of truncations that can be employed to make the simulation more efficient, namely 1) ones that truncate small (or likely small) coefficients $c_{\vec{b}}$ and 2) ones that truncate any monomial $M_{\vec{b}}$ where $\Tr[\varrho M_{\vec{b}}]$ is small (or likely small). Inevitably these truncations will reduce the accuracy of the simulation but in many cases, the reduction in accuracy can be kept small.

If we are simulating a circuit with fixed angles, i.e., the $\vec{\theta}$ parameter for the circuit is already known, then the simplest truncation strategy is \textit{coefficient truncation}. That is, we can keep track of the coefficients in front of each monomial as we back-propagate each term through the circuit and simply truncate (drop) any terms that have an absolute value below some predefined value. 

Alternatively, in the case of \textit{surrogate} simulation where the angles $\vec{\theta}$ are not known in advance, if one is primarily interested in simulating within some small angle range (e.g., in the case of simulating Trotterization-like circuits) one can use the so-called \textit{small angle} truncation scheme~\cite{lerch2024efficient, beguvsic2023simulating}. Namely, we can leverage the fact that when $\vec{\theta}$ are small then sine contributions to the coefficients $c_{\vec{b}}$ in Eq.~\eqref{eq:expectationexpression} are much smaller than cosine contributions. Thus we can prune paths with many sine coefficients without substantially changing the estimation of the expectation value (see Fig.~\ref{fig:sketch}). In Appendix~\ref{app:smallangle} we provide theoretical guarantees for this approach that follow directly from the analogous findings for Pauli Propagation presented in Ref.~\cite{lerch2024efficient}. 

The second truncation strategy, whereby one aims to cut terms from the sum in Eq.~\eqref{eq:expectationexpression} where $\text{Tr}[\varrho M_{\vec{b}}]$ is likely small, is more subtle. While potentially there are a number of different strategies one could take here, we will make use of the observation that in order for a term $\text{Tr}[\varrho M_{\vec{b}}]$ to be non-zero all the Majorana operators in the monomial $M_{\vec{b}}$ must be \textit{paired}. That is, $\text{Tr}[\varrho M_{\vec{b}}] \neq 0$ if and only if, for every even Majorana operator $m_{2i}$ in the Majorana monomial, $m_{2i + 1}$ is also in the monomial (this can be seen immediately from the definition of the Majorana operator and noting that $\langle \phi | a_i | \phi \rangle = \langle \phi | a_i^\dagger | \phi \rangle = 0$ for any Fock basis state $\vert \phi \rangle$). 
Next we make observations that for a randomly chosen  Majorana monomial the higher its length $w$ the lower the probability that the monomial is paired (for $w < N$). 
This motivates a \textit{length truncation} strategy. At each step, after each gate is applied, all terms corresponding to Majorana monomials with length above a certain threshold are to be discarded.
This idea is depicted in Fig.~\ref{fig:sketch}.

While high length monomials are likely to not contribute to a Fock state this fact alone does not justify Majorana length truncation. Namely, high-length Majorana monomials could flow back into the low-length subspace, so eliminating them along the way could introduce significant errors. In broad strokes, the intuition behind why this does not cause issues is that 1) monomials tend to increase in length throughout the evolution, stabilizing at lengths $\approx N$ where their contribution to the estimation is irrelevant and 2) any monomials that do back flow tend to have small coefficients (due to the conservation of the 2-norm of the operator). We discuss these intuitions in more detail in the \textit{end matter} and formalize them into error guarantees in the following section. 

\medskip

\paragraph*{Error Guarantees.}\label{sec:length_truncation}

Here we present guarantees for simulating typical Majorana circuits using low-length Majorana propagation. Concretely, we consider the Majorana circuit defined in Eq.~\eqref{eq:unitarycircuit} of depth $L$ and suppose that each Majorana generator $M_{{\vec{b}_j}}$ is a randomly chosen length $k = 4$ monomial and each rotation angle is randomly chosen in the range $[0, 2 \pi]$. This class of circuits can be regarded as a mathematically tractable, oversimplified model of certain relevant circuits, such as iteratively constructed variational circuits \cite{grimsley2019adaptive}.

For simplicity in our analysis below we will suppose that the initial observable is a randomly chosen \textit{homogeneous} monomial. By homogeneous we mean that the individual Majorana coefficients are unbiased, uncorrelated and with a variance that only depends on length - for a more formal definition see Def.~\ref{def:homogeneous-distribution} in the Appendices. We believe this is a reasonable toy model for a randomly chosen molecular Hamiltonian. The theorem can also be applied more widely to \textit{certain} classes of biased distributions e.g. a random Majorana mode $O$ of length $o$.

We denote the approximation of $f_L(\vec{\theta})$ obtained by truncating all Majorana operators with length greater than $w_0$ by $f_L^{(w_0)}(\vec{\theta})$. The average approximation error can quantified via the \emph{mean squared error} (MSE),
\begin{equation}
   \Delta(w_0, L) = \exval((f_{L}(\vec{\theta}) - f_L^{(w_0)}(\vec{\theta}))^2) \, .
\end{equation}
This error can be bounded as follows.
\begin{theorem}[Error analysis]\label{thm:main-text-upper-bound-mse}
    Consider a random Majorana circuit as described above, with a Fock basis initial state and a random homogeneous observable $O$. Assume that $w_0 \geq 2$. Then for circuits with $L \leq 2 N \ln(5 e w_0 4^{w_0}) $: 
    \[ \frac{\Delta(w_0, L)}{\exval\left(\|O\|_2^2\right)}  \leq  \frac{1}{2^{N-1}} +  \left(\frac{e w_0 }{ N}\right)^{w_0/2} . \]
    Otherwise, for deeper circuits ($L \geq 2N \ln(5e w_0 4^{w_0})$):     \[ \frac{\Delta(w_0, L)}{\exval\left(\|O\|_2^2\right)}  \leq  \frac{1}{2^{N-1}} + 5ew_0\,\left(\frac{16 e w_0}{ N}\right)^{w_0/2} e^{-L/2N} .\]
    Here $\exval\left(\|O\|_2^2\right) = \exval\left( \Tr[ O O^\dagger ]\right)$ is the squared-2-norm averaged over the randomness in $O$.
\end{theorem}

Theorem \ref{thm:main-text-upper-bound-mse} provides rigorous guarantees on the truncation error incurred when simulating Majorana circuits using a low-length approximation. As illustrated in Fig.~\ref{fig:errors}, in all cases the error decreases exponentially with the truncation length $w_0$, reflecting that circuits can be faithfully captured using few-body Majorana operators. The comparison between the exact error (solid blue) and the theoretical bound (dashed green) shows that the empirical scaling closely follows the predicted behavior, validating the theoretical estimates.

\begin{figure}[t!]
    \centering
    \includegraphics[width=0.98\linewidth]{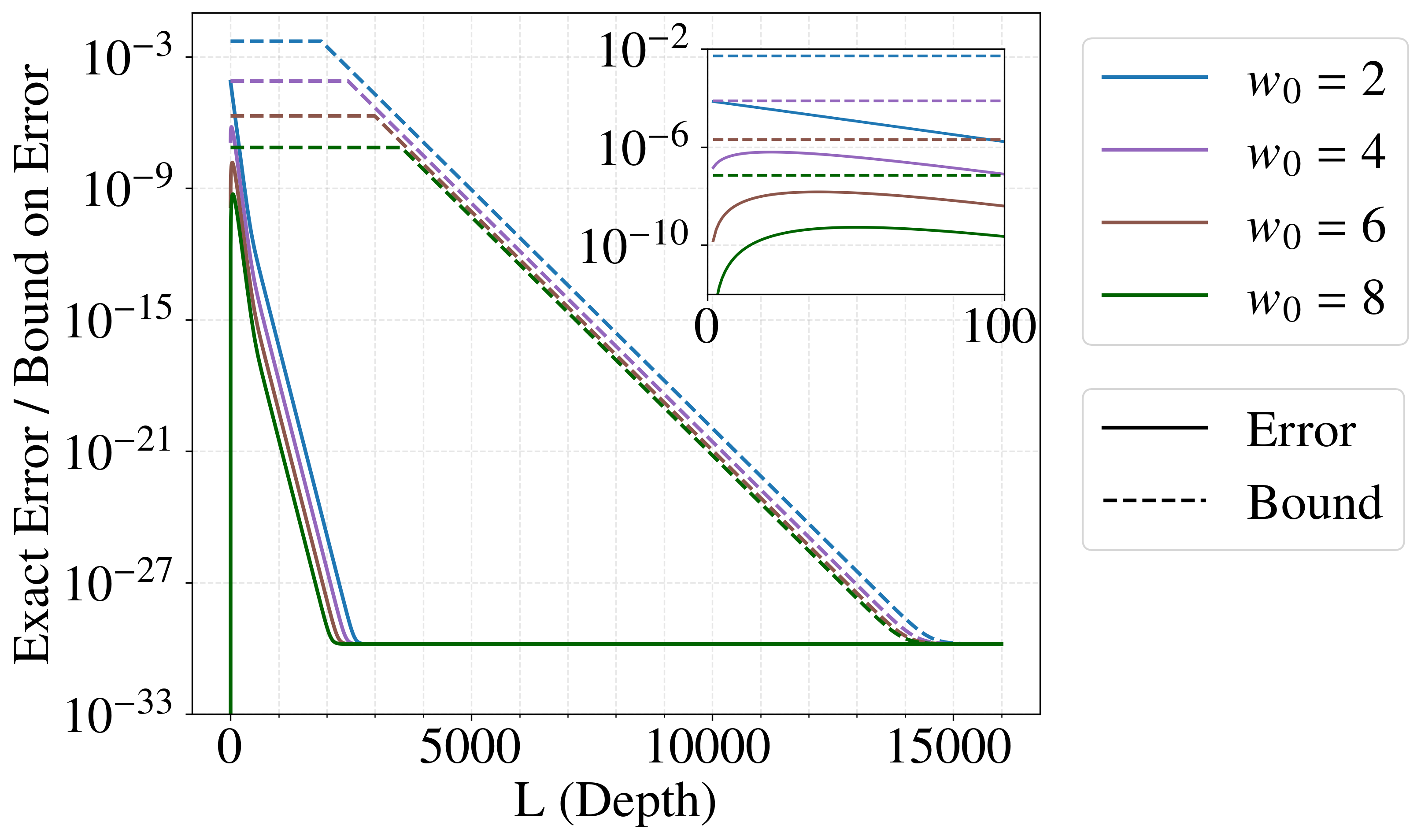}
    \caption{\textbf{Majorana Propagation error analysis.}   Comparison between the exact mean squared simulation error (computed via Theorem~\ref{thm:poly-time-algo-for-squared-error-maintext} in the \textit{end matter}) and the theoretical error bound for a system of $N = 100$ Majorana modes and initial observable Majorana length $= 4$ for different length truncations $w_0$. 
        The solid blue line shows the exact error as a function of circuit depth $L$, while the dashed green line represents the analytical upper bound on this error. The inset zooms in on low depths ($L = 0$ to $L = 100$). Here the observable $O$ is a random Majorana mode of length 4. }
    \label{fig:errors}
\end{figure}

We note that for deeper circuits, in addition to the exponential suppression of the error with the cut off $w_0$, there is also an exponential suppression in circuit depth $L$.  This is a reflection of the fact that the absolute value of expectation values of random circuits decay exponentially with circuit depth to a saturation value that is itself exponentially small in $N$ (a phenomenon analogous to barren plateaus from variational circuits with fixed ansatze~\cite{mcclean2018barren, larocca2024review}). It is satisfying that our error bound, in contrast to prior work~\cite{angrisani2024classically} in the context of Pauli propagation, carefully tracks this depth dependence. 

We further bound the time complexity of our algorithm using the fact that the total number of Majorana monomials with length at most $w_0$ is $\mathcal{O}((2N)^{w_0})$. We can then combine Theorem~\ref{thm:main-text-upper-bound-mse} with Markov's inequality to transform the average error statement into a probabilistic statement. On doing so, we obtain the following Theorem.

\begin{theorem}[Time complexity]
\label{thm:resources}
    Let $U$ be a randomly sampled circuit from an $L$-layered random Majorana rotation circuit on $2 N$ modes, and let $O$ be a random homogeneous observable. Moreover, let $\epsilon, \delta > 0$ such that $\epsilon^{-1}\delta^{-1} \leq 2^{o(N)}$.
    There exists a truncation length $w_0$ for which Majorana propagation runs in time
    $L \cdot \nModes^{\mathcal{O}\left(\log(\epsilon^{-1} \delta^{-1})\right)},$
    and outputs a value $f_L^{(w_0)}(\vtheta)$, such that
    \begin{align}
        \abs{f_L^{(w_0)}(\vtheta) - f_L(\vtheta) } \leq \epsilon,
    \end{align}
    for at least $1 - \delta \mathbb{E}(\| O \|_2^2)$ fraction of the circuits, if $N$ is large enough.
\end{theorem}

Theorem~\ref{thm:resources} thus establishes that for any constant error, only polynomial resources are required to classically simulate the circuit. 
\medskip

\paragraph*{Implementation.}
\begin{figure}[t!]
    \centering
    \includegraphics[]{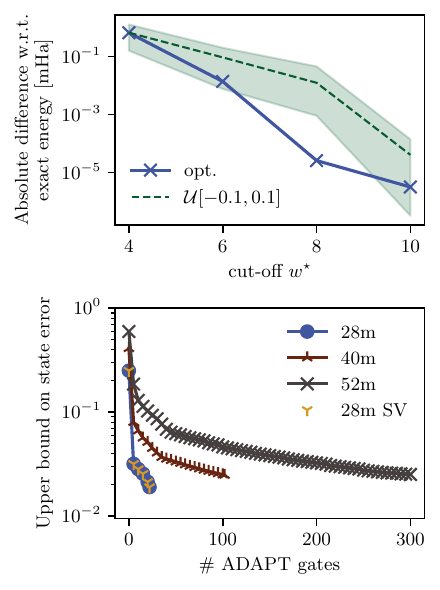}
    \caption{\textbf{Majorana Propagation implementation for ground states.} Top: The blue curve depicts the error in the estimated energy as a function of monomial length cut-off $w^{\star}$ for a 28-mode MP-simulated ADAPT-VQE circuit containing 22 Fermionic gates, computed by comparison against statevector simulation (SV).
    The error of the MP simulation decreases exponentially from below chemical precision (1.6 mHa) for $w^{\star} = 4$ to around numerical precision for $w^{\star} = 10$.
    The green curve shows the same quantity for 5 circuits obtained by uniformly randomizing angles in the range $[-0.1, 0.1]$.
    The shaded area indicates the maximal and minimal errors attained.
    Bottom:
    The ADAPT-VQE inspired algorithm is run classically using MP to simulate the circuits for increasing problem sizes (solid lines).
    The $x$-axis reflects the number of Fermionic gates in the circuit.
    The $y$-axis value is the error in the energy, computed with MP with $w^{\star} = 6$, divided by the spectral gap of the Hamiltonian, which bounds the lack of overlap with the ground state.
     For 28 modes we verify the accuracy of the MP simulation by computing the equivalent error via a SV simulation (yellow markers). SV simulation for the larger problems was not tractable. 
    Only every 5th marker is included to improve readability. }
    \label{fig:adapt_sims}
\end{figure}

We now test Majorana Propagation in the context of variational circuits for quantum chemistry.
Given an input Hamiltonian $H$, the goal is to find a circuit approximating its ground state to a certain degree.
The accuracy of the approximation varies depending on the specific application at hand, but one is typically interested in obtaining non-negligible overlaps with the ground state.

Within the plethora of variational optimization methods for chemistry, adaptive methods have proved to be particularly gate efficient, finding good reference states with relatively modest gate complexity. Here we will explore a classical algorithm inspired by ADAPT-VQE~\cite{grimsley2019adaptive, Q-ADAPT, Grimsley_2023, Anastasiou_2024, 2024reducingresources}, where gates are sequentially selected from a Majoranic pool and added to the circuit as to minimize the energy of the output state.
The operators in the pool are generally low-body, that is, spanned by low-length Majorana monomials.
This adaptive ansatz construction strategy typically results in rather irregular circuits with gates generated by low-length monomials---circuits for which the unstructured circuit class introduced above should be a reasonable model. 

The Hamiltonians considered in this section model the electronic structure of the TLD1433 molecule, a ruthenium-based photosensitizer currently undergoing clinical trials for non-muscle invasive bladder cancer \cite{theralase2024}, with different active sizes comprising $N = \{28, 40, 52\}$ modes.
This is a strongly correlated system which state-of-the-art classical methods, such as DMRG, struggle to solve for larger active spaces.
In our simulations, we use a slightly refined monomial-length-based truncation criterion that will be explained in detail in a forthcoming version of this manuscript.

\begin{table}[t!]

    \centering  \small
        \begin{tabular}{@{}cccccc} 
        \hline
         Number & \multicolumn{2}{c}{Majorana Propagation} 
 & Qiskit& CUDA-Q & CUDA-Q  \\
 of modes & pre-pro. & estim. & SV & SV & MPS \\\hline
              28&  0.05 s & 0.0009 s  &  2.72 h & 12.18 min  & 32.15 min  \\
              40&  1.34 s & 0.0395 s&  -- &  -- &10.10 h  \\ 
              52&  13.85 s & 0.1092 s &  -- &  -- & $>$ 24 h\\
         \hline
    \end{tabular}
\caption{\label{tab:time-benchmarks}
\textbf{Comparison of simulation times.}
Here we compare the times to estimate the energy of the TLD1433 Hamiltonians (of various active space sizes) for an approximate ground state for various simulation methods. 
The approximate ground states have been found by optimizing our ADAPT-VQE inspired ansatz with MP to within 10 mHa error with respect to the DMRG energy benchmark. The two columns for MP refer to the pre-processing time for a given ansatz (i.e., the time to build the surrogate) and to the subsequent energy estimation time (i.e., the time to evaluate the surrogate), respectively. Any additional estimations of the energy for different circuit parameters $\vec{\theta}$ only incur a time given in the second column. We note that for 28 qubits, MP approximates the energy within an absolute error of $\approx 1.39 \cdot 10^{-2}$ mHa with respect to the CUDA-Q statevector exact result, while CUDA-Q MPS method gives an approximation error of $\approx 2.71 $ mHa.}
\end{table}

In Fig.~\ref{fig:adapt_sims} (top), we plot the error in the energy estimated with MP as a function of the length cut-off $w^{\star}$ for an ADAPT-VQE circuit comprising 22 Fermionic gates---equivalent to a 616-CNOT circuit on all-to-all connectivity (see Appendix~\ref{app:numerics})---optimized for the 28-mode active space Hamiltonian.
The error $\varepsilon$ is computed as $\varepsilon = \vert E_{\rm MP}(w^{\star}) - E_{\rm sv} \vert$, where $E_{\rm sv}$ is the exact energy (up to numerical precision) obtained via statevector simulation.
Remarkably, even for $w^{\star} = 4$, MP produces an estimate with an error below chemical precision (1.6 mHa), with the error decreasing approximately exponentially as a function of the cut-off to nearly numerical precision for $w^{\star} = 10$.
A similar behaviour is observed when the angles are randomized within the parameter range in the optimized circuit, although the decay in $w^{\star}$ is slower.

For the larger active space sizes, statevector simulations are out of reach, which makes benchmarking MP challenging.
However, we can take advantage of the fact that we are using MP in the context of ground state chemistry.
State-of-the-art computational chemistry methods, like DMRG, can produce very accurate estimates of the ground state energy for the considered system sizes.
We can therefore compare the energy obtained by training ADAPT-VQE circuits with MP against these benchmarks to indirectly test our simulator.
A well-behaved convergence towards the benchmark ground state energy is an implicit signature of MP's simulation accuracy.

In many applications, the relevant figure of merit quantifying the quality of a reference state $\vert \phi \rangle$ is its overlap with the ground state $\vert \langle \phi \vert \psi_0 \rangle \vert$.
Given the error in the energy, $\vert E - E_0 \vert$, the overlap can be lower-bounded if the spectral gap of the Hamiltonian is known (see Appendix~\ref{app:state_error}).
Thus, using the spectral gap computed with DMRG, we can upper-bound the \textit{state error}, which we define as $1 - \vert \langle \phi \vert \psi_0 \rangle \vert$.
The bottom plot in Fig.~\ref{fig:adapt_sims} shows the upper bound to the state error as a function of the number of Fermionic gates in the circuit as the adaptive training progresses.
All cases show a rapid convergence, reaching an overlap of at least 97\% with the ground state with 22, 102, and 300 Fermionic gates for the three active space sizes (the simulations were manually interrupted upon reaching a absolute error of 10 mHa).
When mapped to qubit space and transpiled for an all-to-all connectivity quantum computer, the equivalent circuits contain 616, 3470, and 13,676 CNOTs, respectively, which highlights the complexity of the obtained states.

As explained above, these results provide an indirect validation of the approach.
In Fig.~\ref{fig:adapt_sims} (bottom), two sources of error are compounded: the error of the circuit itself, and the estimation error incurred by MP.
Thus, it is in principle possible for the estimation error to contribute as to lower the energy, resulting in estimated energies closer to the ground state than the true circuits'.
In order to assess to what extent this is the case, we evaluate the energy of several intermediate circuits obtained throughout the 28-mode ADAPT-VQE simulation with statevector.
As shown in the figure, the state error bounds computed with MP and statevector agree almost perfectly.
This shows that MP is accurate enough to support the successful classical simulation of ADAPT-VQE circuits, leading to a collection of circuits increasingly converging to the ground state.

Finally, Table~\ref{tab:time-benchmarks} reports run-time benchmarks for energy estimation with MP and other methods: two statevector simulators (Qiskit and CUDA-Q), as well as CUDA-Q MPS simulator, a tensor network-based circuit simulator.
The circuits used are the converged ADAPT-VQE circuits from Fig.~\ref{fig:adapt_sims} (bottom), transpiled to qubit space, resulting in the CNOT counts reported above.
While statevector simulators are exact (up to numerical precision), they scale exponentially.
Thus, only the 28-mode circuit could be simulated with these simulators with reasonable resources.
MP and CUDA-Q MPS, on the other hand, produce approximate estimates but run in polynomial time.
Details on the implementation of these methods can be found in Appendix~\ref{app:numerics}.
Interestingly, in addition to simulating these circuits considerably faster than CUDA-Q MPS, the 28-mode result is also nearly two orders of magnitude more accurate, approximating the statevector energy within $\approx 1.39 \cdot 10^{-2}$ mHa, as opposed to the $\approx 2.71 $ mHa error in the MPS simulation.

Crucially, most of the time taken by the MP calculations is in a pre-processing stage, after which the re-evaluation of the energy for different parameter values $\vec{\theta}$ is very fast.
This feature is particularly convenient for variational circuits, which benefits from parameter optimization.

\medskip

\paragraph*{Discussion.}
Our analytical results portray MP as a very natural framework for the simulation of Fermionic circuits.
MP leverages newly identified phenomena in the behaviour of basis monomials, namely monomial length, which is conceptually very different to Pauli weight.
The underlying assumption in our analysis, specifically that circuits are unstructured, with gate generators randomly chosen from a set of bounded-length monomials, seems particularly appropriate to model iteratively constructed ground state approximating circuits.
Our analysis also sheds light on the working principles behind approaches like Majorana-length-damping in tensor network methods for Fermionic systems~\cite{kuo2024energy}.

Our numerical results substantiate that MP can be used to classically simulate ADAPT-VQE to produce circuits with significant overlap with the ground state in strongly correlated systems.
In our simulations, MP was able to outperform state-of-the-art, tensor network circuit simulators.
It would be interesting to compare its performance to more recently proposed approximate simulation methods for Fermionic circuits, like Refs.~\cite{mullinax2023large,mullinax2024classical,reardon2024improved,mocherla2023extending}.
Since MP does not rely on the sparsity of the wave-function in the Fock basis, and since single-excitation gates do not require additional truncation in MP, MP may produce more accurate estimations for circuits with many such excitations leading to dense states. Similarly, we stress that MP is strictly more powerful and efficient than `gsim' approaches~\cite{somma2005quantum, somma2006efficient, Galitski2011Quantum, anschuetz2022efficient, goh2023lie} as the basis of the algebra is effectively generated `on the fly’ during the simulation (rather than in advance) and so is not restricted to free Fermionic models with polynomially scaling Lie algebras \cite{goh2023lie, zimboras2014dynamic}. 

MP is, by design, well-suited to being interfaced with quantum hardware with the aid of Fermion-to-qubit mappings~\cite{miller2023bonsai,miller2024treespilation, neven_mapping, jw_mapping, bk_mapping}.
Through these transformations, the circuits obtained via MP can be readily executed on hardware.
While MP can be used to efficiently approximate expectation values of observables, hardware execution will be needed in many contexts, even when sampling---a task MP cannot be used for---is not required.
For instance, in many computational pipelines, it may generally prove more efficient to use MP to train a circuit up to a limited but sufficient level of accuracy and then use it as a reference for another algorithm, such as QPE, than solving the problem purely classically.

More generally, including contexts like dynamical simulations and machine learning and beyond, hybrid methods where the outputs of a quantum algorithm are further processed with the help of a classical MP simulation have the potential to make the most of size and depth-limited near-term quantum hardware by reducing the load on the quantum device. But even in the fault tolerant era such hybrid approaches will likely remain valuable and MP will prove useful for finding good initial states and processing the outputs of fermion simulations on quantum hardware.

\medskip

\paragraph*{End Matter.}

In order to study Majorana length truncation analytically, we consider unstructured Fermionic circuits with gates generated by Majorana monomials with $\norm{\vec{b}_j}_1$ an even number smaller than or equal to some $k^{\star} \ll N$.
Unstructured here means that the generators are randomly sampled from the pool of such operators, $\{ M_{\vec{b}} : \norm{\vec{b}}_1 \leq k^{\star}, \norm{\vec{b}}_1 \mod 2 = 0 \}$.

\begin{figure}
    \centering
    \includegraphics[width=0.9\columnwidth]{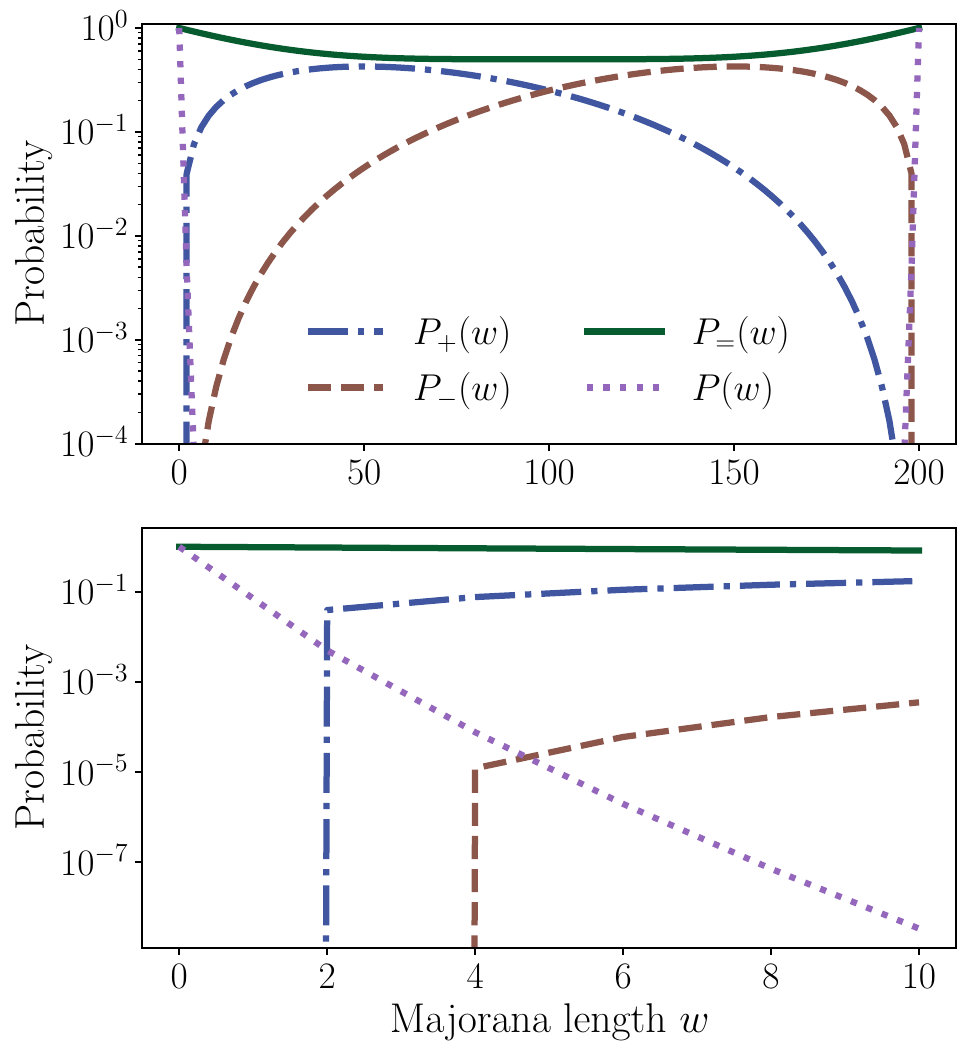}
    \caption{\textbf{Suppression of high-length modes and backflow.} Transition and pairing probabilities for $N = 100$ modes and a randomly chosen length-4 gate for all $w$ lengths from 0 to $2N$ (top) and for $w \leq 10$ (bottom).
    For small $w$, $p_+$ is substantially larger than $p_-$, which tends to drive monomial length upwards in the event of branching.
    The roles of $p_+$ and $p_-$ are symmetrically inverted at $w = N$, so successive branching tends to trap Majorana monomials around that point.
    The highest probability corresponds to no-branching events, $P_{=} = 1 - [p_+ + p_-]$.
    At the end of the Heisenberg evolution, monomials are traced with a Fock basis state.
    The probability for a length-$w$ monomial to contribute to the trace, $\probContribute{w}$, decreases very rapidly in $w$, rendering the contribution of high-$w$ terms irrelevant.}
    \label{fig:analytics}
\end{figure}

The first important point motivating Majorana length truncation is the fact that most high-length terms do not contribute to the expectation value $\langle H \rangle$.
As discussed previously, only paired Majorana monomials yield non-zero expectation values when traced with a Fock basis state.
The probability $\probContribute{w}$ for a randomly chosen $w$-monomial to be paired, is given by
\begin{equation}
    \label{eq:pairing_prob}
    \probContribute{w} = \frac{\binom{N}{w/2}}{\binom{2N}{w}}
\end{equation}
and as depicted in Fig.~\ref{fig:analytics}, is very small unless $w$ is close to $0$ or to $2N$.
This is stated more precisely in the following theorem (see Appendix~\ref{app:exp_suppression} for proof).
\begin{proposition}[The contribution of high length Majorana monomials are exponentially suppressed]\label{thm:suppression}
The probability $\probContribute{w}$ that any randomly chosen $N$-mode Majorana monomial with length $w$, where both $w \in \Theta(N)$ and $(2N - w) \in \Theta(N)$, has a non-zero expectation in any given Fock state is exponentially suppressed in $N$, that is  
\begin{equation}
    \probContribute{w} \in \mathcal{O}(r^{-N}),
\end{equation}
where $r > 1$. Furthermore, the minimal probability is obtained for $w = N$ with 
\begin{equation}
        \min_w \probContribute{w} = \probContribute{N} \in \mathcal{O}(2^{-N}) \, . 
\end{equation}
Conversely, if $w \in \OC(1)$ we have that
\begin{equation}
    \probContribute{w} \in \Omega\left(\frac{1}{\text{poly}(N)}\right) \, . 
\end{equation}
\end{proposition}
While crucial, this fact alone does not justify Majorana length truncation.
Truncating below a certain threshold $w^{\star} \ll N$ disposes of terms with length $w \approx 2N$, if any, which can potentially contribute non-negligibly to the expectation value.
Furthermore, the fact that these Majorana monomials do not contribute significantly to the trace with a Fock basis state at best justifies the truncation at the very end of the evolution, once all gates have been applied.
At intermediate steps, high-length Majorana monomials could flow back into the low-length subspace, so eliminating them along the way could introduce significant errors.
To understand why these events are typically very unlikely, we must turn our attention to the analysis of typical Majorana monomial dynamics.

From the algebraic properties of Majorana operators, Eq.~\eqref{eq:majorana_properties}, it can be seen that the two monomials anti-commute if and only if they have an odd number of Majorana operators in common, that is, iff the number of overlapping operators $s$, which can be computed as $s = \vec{b} \cdot \vec{b}_j$, is odd.
We can draw several interesting conclusions from this fact.

First, let $k \coloneqq \norm{\vec{b}_j}_1$ and $w \coloneqq \norm{\vec{b}}_1$.
As a consequence of the properties in Eq.~\eqref{eq:majorana_properties}, more precisely of the fact that even repetitions of Majorana operators in any product cancel out, the length $w'$ of the monomial in the sine branch, $\mathrm i M_{\vec{b}_j} M_{\vec{b}}$, is the sum of the lengths of the two original monomials, $w + k$, minus twice their overlap, $s$:
\begin{equation}
    \label{eq:output_length}
    w' = w + k - 2s.
\end{equation}
Hence, even-$k$ gates preserve Majorana length parity; in physically relevant scenarios, where all gate generators and monomials in the observable are even, one only needs to consider the even-monomial subspace.

Second, for gates generated by so-called single-excitation, or one-body monomials ($k = 2$), branching can only occur for $s = 1$, and so $w' = w$.
In other words, these gates preserve Majorana length and can therefore be accounted for exactly by the algorithm: no additional truncations are required in order to restrict the operator to the subspace.

Third, according to Eq.~\eqref{eq:output_length}, the length $w'$ of the new monomial in the sine branch is $w' > w$ if $s < k /2$, $w' < w$ if $s > k /2$, and $w' = w$ if $s = k /2$.
Importantly, for $w, k \ll N$, smaller values of $s$ are much more likely than larger ones.
This can be seen intuitively: drawing $k$ distinct Majorana operators out of the set of all $2N$, the most likely event is for none of them to match those within the $w$ operators in $M_{\vec{b}}$.
Consequently, there is an asymmetry in the dynamics whereby branching events tend to increase Majorana length with higher likelihood than they tend to decrease it, preventing backflows of monomials back into the low-length subspace.
Crucially, for $w > N$, the effect is reversed and the dynamics tends to decrease Majorana length, so the dynamics tends to cluster monomials around $w = N$.

This can be seen more formally, since the probability for the overlap to be exactly $s$ can be derived analytically,
\begin{equation}\label{eq:overlap_prob_main}
   P(N, w, k, s) =  \frac{\binom{k}{s} \binom{2N-k}{w -s}}{\binom{2N}{w}},
\end{equation}
from which the probabilities for the length to increase or decrease, $p_+$ and $p_-$, can be computed by simply summing over $s$ for $s < k/2$ and $s > k/2$, respectively.
These quantities are illustrated in Fig.~\ref{fig:analytics}.
We note the invariance of Eq.~\eqref{eq:overlap_prob_main} with respect to $w \rightarrow 2N - w, \, s \rightarrow k - s$, which explains the reversal of the phenomenon for $w > N$.
Finally, using these expressions, it is possible to analyze the behaviour of the ratio $R \coloneqq p_- / p_+$, which quantifies the asymmetry in the dynamics and to show that, for fixed $w$ and $k$, it vanishes quadratically in $N$, making backflows increasingly unlikely as $N$ increases.
All these points are captured formally in the following theorem (proof in Appendix~\ref{app:backflow_suppresion}).

\begin{proposition}[Suppression of backflow]\label{thm:backflow}
Consider the propagation of a length $w$ Majorana monomial $M$ under a length $k$ generator of rotation $M'$ as per Eq.~\eqref{eq:gatebranch}. We assume that $[M', M] \neq 0$ such that the propagation causes branching and we denote the length of the new Majorana monomial, $M' M$, as $w'$. Let $p_+$ and $p_-$ denote the probabilities that the propagation increases ($w' > w$) and decreases ($w < w'$) operator length. The following statements hold:
\begin{enumerate}
    \item Any rotation under a $k=2$ length generator $M$ leaves length $w$ unchanged,  $p_+ = p_- = 0$. 
    \item For any $w \in \OC(1)$ and $k \in \OC(1)$ with $2 \leq k \leq w$ such that  $w k < N$ we have that 
    \begin{equation}
        R \coloneqq p_-/p_+ \in \OC(1/N^2) \, .
    \end{equation}
\end{enumerate}
\end{proposition}

This analysis paints a very clear picture of Majorana monomial dynamics for typical circuits in the Heisenberg picture: as the operator evolves through the gates in the circuit, operator branching leads to ever-increasing Majorana length, with very low likelihood of length-decreasing events, until length $w \approx N$ is reached.
After that point, branching tends to decrease length, trapping monomials around $w = N$.

The algorithm introduced in this work takes advantage of this dynamics.
As explained above, the monomials clustering around $w = N$ contribute exponentially little to the trace with the Fock basis state, so they can be neglected.
Due to the inversion in the branching imbalance for $w > N$, the terms with $w \approx 2N$ are never populated, so they need not be considered in the simulation.
Moreover, since backflow probability throughout the dynamics is very small, they can be neglected early on, as soon as they leave the low-length subspace.

We finally need to account for the fact that while backflow is suppressed in theory the contribution of many unlikely backflows could add up to a substantial contribution to the expectation value. However, in practice this does not occur as the total two norm of the evolved operator is conserved under unitary evolutions and so the coefficients of the terms that do backflow are small. 

In Appendix~\ref{app:endtoendanalysis} we tie together these intuitions to prove that the average error of a length-truncated Majorana propagation simulation is exponentially suppressed in truncation length $w_0$. To do so we employ a Markov chain analysis that is conceptually rather different to the proof strategies employed to analyze path propagation methods previously~\cite{angrisani2024classically}. In particular it allows us to first obtain an efficient classical algorithm for computing the mean squared error \textit{exactly}. 
\begin{theorem}\label{thm:poly-time-algo-for-squared-error-maintext}
    The mean squared error $\exval((f_L - f_L^{(w_0)})^2)$ for a random Majorana circuit as defined in Def.~\ref{def:circuitmodel} can be computed exactly classically in polynomial time. 
\end{theorem}
\noindent It is this algorithm that we use to plot the exact average error in  Fig.~\ref{fig:errors}. We then proceed to bound the error. Crucially, in contrast to proofs for PP~\cite{angrisani2024classically}, we obtain a \textit{depth dependent} error bound as shown in Theorem~\ref{thm:main-text-upper-bound-mse} and plotted in Fig.~\ref{fig:errors}. In fact, our algorithm can also be applied more generally for computing the average of the expectation value. It would be interesting to investigate whether this strategy could be applied to propagation methods in other contexts to compute coarse-grained properties and/or derive tighter bounds. 

\medskip

\paragraph*{Acknowledgements}
Work on ``Quantum Computing for Photon-Drug Interactions in Cancer Prevention and Treatment'' is supported by Wellcome Leap as part of the Q4Bio Program.
We thank Pi A.~B.~Hasse, Stefan Knecht, Fabijan Pavosevic, Martina Stella, and Fabio Tarocco for providing the TLD1433 Hamiltonians and the corresponding DMRG benchmarks.
We also thank Ludmila A.~S.~Botelho and Roberto Di Remigio for their coding optimization advice.

\paragraph*{Note}
During the completion of this work, we became aware of independent work by Matteo D'Anna and Jannes Nys where they employ a related approach to study dynamics in Fermionic lattice models.

\paragraph*{Competing interests}
Elements of this work are included in patent applications filed by Algorithmiq Oy currently pending with the European Patent Office.

\paragraph*{Author contributions}
GGP conceived the method.
AM implemented the method.
ZH and GGP formalized the initial theory and wrote the first version of the manuscript. JF conducted the end-to-end error analysis under the supervision of ZH and ZZ. 
AM, OS, AN, RC, and AG implemented and executed the ADAPT-VQE simulations under AG's direction.
AM, OS, and AG produced the figures.
All authors participated to discussions and to the writing of the manuscript.

\bibliography{notes/ref, notes/quantum, notes/QITE} 

\begin{thebibliography}{73}%
\makeatletter
\providecommand \@ifxundefined [1]{%
 \@ifx{#1\undefined}
}%
\providecommand \@ifnum [1]{%
 \ifnum #1\expandafter \@firstoftwo
 \else \expandafter \@secondoftwo
 \fi
}%
\providecommand \@ifx [1]{%
 \ifx #1\expandafter \@firstoftwo
 \else \expandafter \@secondoftwo
 \fi
}%
\providecommand \natexlab [1]{#1}%
\providecommand \enquote  [1]{``#1''}%
\providecommand \bibnamefont  [1]{#1}%
\providecommand \bibfnamefont [1]{#1}%
\providecommand \citenamefont [1]{#1}%
\providecommand \href@noop [0]{\@secondoftwo}%
\providecommand \href [0]{\begingroup \@sanitize@url \@href}%
\providecommand \@href[1]{\@@startlink{#1}\@@href}%
\providecommand \@@href[1]{\endgroup#1\@@endlink}%
\providecommand \@sanitize@url [0]{\catcode `\\12\catcode `\$12\catcode
  `\&12\catcode `\#12\catcode `\^12\catcode `\_12\catcode `\%12\relax}%
\providecommand \@@startlink[1]{}%
\providecommand \@@endlink[0]{}%
\providecommand \url  [0]{\begingroup\@sanitize@url \@url }%
\providecommand \@url [1]{\endgroup\@href {#1}{\urlprefix }}%
\providecommand \urlprefix  [0]{URL }%
\providecommand \Eprint [0]{\href }%
\providecommand \doibase [0]{https://doi.org/}%
\providecommand \selectlanguage [0]{\@gobble}%
\providecommand \bibinfo  [0]{\@secondoftwo}%
\providecommand \bibfield  [0]{\@secondoftwo}%
\providecommand \translation [1]{[#1]}%
\providecommand \BibitemOpen [0]{}%
\providecommand \bibitemStop [0]{}%
\providecommand \bibitemNoStop [0]{.\EOS\space}%
\providecommand \EOS [0]{\spacefactor3000\relax}%
\providecommand \BibitemShut  [1]{\csname bibitem#1\endcsname}%
\let\auto@bib@innerbib\@empty
\bibitem [{\citenamefont {Preskill}(2018)}]{preskill2018quantum}%
  \BibitemOpen
  \bibfield  {author} {\bibinfo {author} {\bibfnamefont {J.}~\bibnamefont
  {Preskill}},\ }\bibfield  {title} {\bibinfo {title} {Quantum computing in the
  {NISQ} era and beyond},\ }\href {https://doi.org/10.22331/q-2018-08-06-79}
  {\bibfield  {journal} {\bibinfo  {journal} {Quantum}\ }\textbf {\bibinfo
  {volume} {2}},\ \bibinfo {pages} {79} (\bibinfo {year} {2018})}\BibitemShut
  {NoStop}%
\bibitem [{\citenamefont {Rall}\ \emph {et~al.}(2019)\citenamefont {Rall},
  \citenamefont {Liang}, \citenamefont {Cook},\ and\ \citenamefont
  {Kretschmer}}]{rall2019simulation}%
  \BibitemOpen
  \bibfield  {author} {\bibinfo {author} {\bibfnamefont {P.}~\bibnamefont
  {Rall}}, \bibinfo {author} {\bibfnamefont {D.}~\bibnamefont {Liang}},
  \bibinfo {author} {\bibfnamefont {J.}~\bibnamefont {Cook}},\ and\ \bibinfo
  {author} {\bibfnamefont {W.}~\bibnamefont {Kretschmer}},\ }\bibfield  {title}
  {\bibinfo {title} {Simulation of qubit quantum circuits via pauli
  propagation},\ }\href {https://doi.org/10.1103/PhysRevA.99.062337} {\bibfield
   {journal} {\bibinfo  {journal} {Physical Review A}\ }\textbf {\bibinfo
  {volume} {99}},\ \bibinfo {pages} {062337} (\bibinfo {year}
  {2019})}\BibitemShut {NoStop}%
\bibitem [{\citenamefont {Aharonov}\ \emph {et~al.}(2023)\citenamefont
  {Aharonov}, \citenamefont {Gao}, \citenamefont {Landau}, \citenamefont
  {Liu},\ and\ \citenamefont {Vazirani}}]{aharonov2022polynomial}%
  \BibitemOpen
  \bibfield  {author} {\bibinfo {author} {\bibfnamefont {D.}~\bibnamefont
  {Aharonov}}, \bibinfo {author} {\bibfnamefont {X.}~\bibnamefont {Gao}},
  \bibinfo {author} {\bibfnamefont {Z.}~\bibnamefont {Landau}}, \bibinfo
  {author} {\bibfnamefont {Y.}~\bibnamefont {Liu}},\ and\ \bibinfo {author}
  {\bibfnamefont {U.}~\bibnamefont {Vazirani}},\ }\bibfield  {title} {\bibinfo
  {title} {A polynomial-time classical algorithm for noisy random circuit
  sampling},\ }\href {https://doi.org/10.1145/3564246.3585234} {\bibfield
  {journal} {\bibinfo  {journal} {Proceedings of the 55th Annual ACM Symposium
  on Theory of Computing}\ ,\ \bibinfo {pages} {945}} (\bibinfo {year}
  {2023})}\BibitemShut {NoStop}%
\bibitem [{\citenamefont {Rakovszky}\ \emph {et~al.}(2022)\citenamefont
  {Rakovszky}, \citenamefont {Von~Keyserlingk},\ and\ \citenamefont
  {Pollmann}}]{rakovszky2022dissipation}%
  \BibitemOpen
  \bibfield  {author} {\bibinfo {author} {\bibfnamefont {T.}~\bibnamefont
  {Rakovszky}}, \bibinfo {author} {\bibfnamefont {C.}~\bibnamefont
  {Von~Keyserlingk}},\ and\ \bibinfo {author} {\bibfnamefont {F.}~\bibnamefont
  {Pollmann}},\ }\bibfield  {title} {\bibinfo {title} {Dissipation-assisted
  operator evolution method for capturing hydrodynamic transport},\ }\href@noop
  {} {\bibfield  {journal} {\bibinfo  {journal} {Physical Review B}\ }\textbf
  {\bibinfo {volume} {105}},\ \bibinfo {pages} {075131} (\bibinfo {year}
  {2022})}\BibitemShut {NoStop}%
\bibitem [{\citenamefont {Begu{\v{s}}i{\'c}}\ \emph {et~al.}(2025)\citenamefont
  {Begu{\v{s}}i{\'c}}, \citenamefont {Hejazi},\ and\ \citenamefont
  {Chan}}]{beguvsic2023simulating}%
  \BibitemOpen
  \bibfield  {author} {\bibinfo {author} {\bibfnamefont {T.}~\bibnamefont
  {Begu{\v{s}}i{\'c}}}, \bibinfo {author} {\bibfnamefont {K.}~\bibnamefont
  {Hejazi}},\ and\ \bibinfo {author} {\bibfnamefont {G.~K.}\ \bibnamefont
  {Chan}},\ }\bibfield  {title} {\bibinfo {title} {Simulating quantum circuit
  expectation values by {C}lifford perturbation theory},\ }\bibfield  {journal}
  {\bibinfo  {journal} {The Journal of Chemical Physics}\ }\textbf {\bibinfo
  {volume} {162}},\ \href {https://doi.org/10.1063/5.0269149}
  {10.1063/5.0269149} (\bibinfo {year} {2025})\BibitemShut {NoStop}%
\bibitem [{\citenamefont {Fontana}\ \emph {et~al.}(2025)\citenamefont
  {Fontana}, \citenamefont {Rudolph}, \citenamefont {Duncan}, \citenamefont
  {Rungger},\ and\ \citenamefont {C{\^\i}rstoiu}}]{fontana2023classical}%
  \BibitemOpen
  \bibfield  {author} {\bibinfo {author} {\bibfnamefont {E.}~\bibnamefont
  {Fontana}}, \bibinfo {author} {\bibfnamefont {M.~S.}\ \bibnamefont
  {Rudolph}}, \bibinfo {author} {\bibfnamefont {R.}~\bibnamefont {Duncan}},
  \bibinfo {author} {\bibfnamefont {I.}~\bibnamefont {Rungger}},\ and\ \bibinfo
  {author} {\bibfnamefont {C.}~\bibnamefont {C{\^\i}rstoiu}},\ }\bibfield
  {title} {\bibinfo {title} {Classical simulations of noisy variational quantum
  circuits},\ }\href
  {https://doi.org/https://doi.org/10.1038/s41534-024-00955-1} {\bibfield
  {journal} {\bibinfo  {journal} {npj Quantum Information}\ }\textbf {\bibinfo
  {volume} {11}},\ \bibinfo {pages} {1} (\bibinfo {year} {2025})}\BibitemShut
  {NoStop}%
\bibitem [{\citenamefont {Shao}\ \emph {et~al.}(2024)\citenamefont {Shao},
  \citenamefont {Wei}, \citenamefont {Cheng},\ and\ \citenamefont
  {Liu}}]{shao2023simulating}%
  \BibitemOpen
  \bibfield  {author} {\bibinfo {author} {\bibfnamefont {Y.}~\bibnamefont
  {Shao}}, \bibinfo {author} {\bibfnamefont {F.}~\bibnamefont {Wei}}, \bibinfo
  {author} {\bibfnamefont {S.}~\bibnamefont {Cheng}},\ and\ \bibinfo {author}
  {\bibfnamefont {Z.}~\bibnamefont {Liu}},\ }\bibfield  {title} {\bibinfo
  {title} {Simulating noisy variational quantum algorithms: A polynomial
  approach},\ }\href {https://doi.org/10.1103/PhysRevLett.133.120603}
  {\bibfield  {journal} {\bibinfo  {journal} {Physical Review Letters}\
  }\textbf {\bibinfo {volume} {133}},\ \bibinfo {pages} {120603} (\bibinfo
  {year} {2024})}\BibitemShut {NoStop}%
\bibitem [{\citenamefont {Rudolph}\ \emph {et~al.}(2023)\citenamefont
  {Rudolph}, \citenamefont {Fontana}, \citenamefont {Holmes},\ and\
  \citenamefont {Cincio}}]{rudolph2023classical}%
  \BibitemOpen
  \bibfield  {author} {\bibinfo {author} {\bibfnamefont {M.~S.}\ \bibnamefont
  {Rudolph}}, \bibinfo {author} {\bibfnamefont {E.}~\bibnamefont {Fontana}},
  \bibinfo {author} {\bibfnamefont {Z.}~\bibnamefont {Holmes}},\ and\ \bibinfo
  {author} {\bibfnamefont {L.}~\bibnamefont {Cincio}},\ }\bibfield  {title}
  {\bibinfo {title} {Classical surrogate simulation of quantum systems with
  {LOWESA}},\ }\href {https://arxiv.org/abs/2308.09109} {\bibfield  {journal}
  {\bibinfo  {journal} {arXiv preprint arXiv:2308.09109}\ } (\bibinfo {year}
  {2023})}\BibitemShut {NoStop}%
\bibitem [{\citenamefont {Schuster}\ \emph {et~al.}(2024)\citenamefont
  {Schuster}, \citenamefont {Yin}, \citenamefont {Gao},\ and\ \citenamefont
  {Yao}}]{schuster2024polynomial}%
  \BibitemOpen
  \bibfield  {author} {\bibinfo {author} {\bibfnamefont {T.}~\bibnamefont
  {Schuster}}, \bibinfo {author} {\bibfnamefont {C.}~\bibnamefont {Yin}},
  \bibinfo {author} {\bibfnamefont {X.}~\bibnamefont {Gao}},\ and\ \bibinfo
  {author} {\bibfnamefont {N.~Y.}\ \bibnamefont {Yao}},\ }\bibfield  {title}
  {\bibinfo {title} {A polynomial-time classical algorithm for noisy quantum
  circuits},\ }\bibfield  {journal} {\bibinfo  {journal} {arXiv preprint
  arXiv:2407.12768}\ }\href
  {https://doi.org/https://doi.org/10.48550/arXiv.2407.12768}
  {https://doi.org/10.48550/arXiv.2407.12768} (\bibinfo {year}
  {2024})\BibitemShut {NoStop}%
\bibitem [{\citenamefont {Angrisani}\ \emph {et~al.}(2024)\citenamefont
  {Angrisani}, \citenamefont {Schmidhuber}, \citenamefont {Rudolph},
  \citenamefont {Cerezo}, \citenamefont {Holmes},\ and\ \citenamefont
  {Huang}}]{angrisani2024classically}%
  \BibitemOpen
  \bibfield  {author} {\bibinfo {author} {\bibfnamefont {A.}~\bibnamefont
  {Angrisani}}, \bibinfo {author} {\bibfnamefont {A.}~\bibnamefont
  {Schmidhuber}}, \bibinfo {author} {\bibfnamefont {M.~S.}\ \bibnamefont
  {Rudolph}}, \bibinfo {author} {\bibfnamefont {M.}~\bibnamefont {Cerezo}},
  \bibinfo {author} {\bibfnamefont {Z.}~\bibnamefont {Holmes}},\ and\ \bibinfo
  {author} {\bibfnamefont {H.-Y.}\ \bibnamefont {Huang}},\ }\bibfield  {title}
  {\bibinfo {title} {Classically estimating observables of noiseless quantum
  circuits},\ }\href@noop {} {\bibfield  {journal} {\bibinfo  {journal} {arXiv
  preprint arXiv:2409.01706}\ } (\bibinfo {year} {2024})}\BibitemShut {NoStop}%
\bibitem [{\citenamefont {Gonz{\'a}lez-Garc{\'\i}a}\ \emph
  {et~al.}(2024)\citenamefont {Gonz{\'a}lez-Garc{\'\i}a}, \citenamefont
  {Cirac},\ and\ \citenamefont {Trivedi}}]{gonzalez2024pauli}%
  \BibitemOpen
  \bibfield  {author} {\bibinfo {author} {\bibfnamefont {G.}~\bibnamefont
  {Gonz{\'a}lez-Garc{\'\i}a}}, \bibinfo {author} {\bibfnamefont {J.~I.}\
  \bibnamefont {Cirac}},\ and\ \bibinfo {author} {\bibfnamefont
  {R.}~\bibnamefont {Trivedi}},\ }\bibfield  {title} {\bibinfo {title} {Pauli
  path simulations of noisy quantum circuits beyond average case},\ }\href
  {https://doi.org/10.48550/arXiv.2407.16068} {\bibfield  {journal} {\bibinfo
  {journal} {arXiv preprint arXiv:2407.16068}\ } (\bibinfo {year}
  {2024})}\BibitemShut {NoStop}%
\bibitem [{\citenamefont {Lerch}\ \emph {et~al.}(2024)\citenamefont {Lerch},
  \citenamefont {Puig}, \citenamefont {Rudolph}, \citenamefont {Angrisani},
  \citenamefont {Jones}, \citenamefont {Cerezo}, \citenamefont {Thanasilp},\
  and\ \citenamefont {Holmes}}]{lerch2024efficient}%
  \BibitemOpen
  \bibfield  {author} {\bibinfo {author} {\bibfnamefont {S.}~\bibnamefont
  {Lerch}}, \bibinfo {author} {\bibfnamefont {R.}~\bibnamefont {Puig}},
  \bibinfo {author} {\bibfnamefont {M.}~\bibnamefont {Rudolph}}, \bibinfo
  {author} {\bibfnamefont {A.}~\bibnamefont {Angrisani}}, \bibinfo {author}
  {\bibfnamefont {T.}~\bibnamefont {Jones}}, \bibinfo {author} {\bibfnamefont
  {M.}~\bibnamefont {Cerezo}}, \bibinfo {author} {\bibfnamefont
  {S.}~\bibnamefont {Thanasilp}},\ and\ \bibinfo {author} {\bibfnamefont
  {Z.}~\bibnamefont {Holmes}},\ }\bibfield  {title} {\bibinfo {title}
  {Efficient quantum-enhanced classical simulation for patches of quantum
  landscapes},\ }\bibfield  {journal} {\bibinfo  {journal} {arXiv preprint
  arXiv:2411.19896}\ }\href {https://doi.org/10.48550/arXiv.2411.19896}
  {10.48550/arXiv.2411.19896} (\bibinfo {year} {2024})\BibitemShut {NoStop}%
\bibitem [{\citenamefont {Cirstoiu}(2024)}]{cirstoiu2024fourier}%
  \BibitemOpen
  \bibfield  {author} {\bibinfo {author} {\bibfnamefont {C.}~\bibnamefont
  {Cirstoiu}},\ }\bibfield  {title} {\bibinfo {title} {A fourier analysis
  framework for approximate classical simulations of quantum circuits},\ }\href
  {10.48550/arXiv.2410.13856} {\bibfield  {journal} {\bibinfo  {journal} {arXiv
  preprint arXiv:2410.13856}\ } (\bibinfo {year} {2024})}\BibitemShut {NoStop}%
\bibitem [{\citenamefont {Angrisani}\ \emph {et~al.}(2025)\citenamefont
  {Angrisani}, \citenamefont {Mele}, \citenamefont {Rudolph}, \citenamefont
  {Cerezo},\ and\ \citenamefont {Holmes}}]{angrisani2025simulating}%
  \BibitemOpen
  \bibfield  {author} {\bibinfo {author} {\bibfnamefont {A.}~\bibnamefont
  {Angrisani}}, \bibinfo {author} {\bibfnamefont {A.~A.}\ \bibnamefont {Mele}},
  \bibinfo {author} {\bibfnamefont {M.~S.}\ \bibnamefont {Rudolph}}, \bibinfo
  {author} {\bibfnamefont {M.}~\bibnamefont {Cerezo}},\ and\ \bibinfo {author}
  {\bibfnamefont {Z.}~\bibnamefont {Holmes}},\ }\bibfield  {title} {\bibinfo
  {title} {Simulating quantum circuits with arbitrary local noise using pauli
  propagation},\ }\bibfield  {journal} {\bibinfo  {journal} {arXiv preprint
  arXiv:2501.13101}\ }\href {https://doi.org/10.48550/arXiv.2501.13101}
  {10.48550/arXiv.2501.13101} (\bibinfo {year} {2025})\BibitemShut {NoStop}%
\bibitem [{\citenamefont {Fuller}\ \emph {et~al.}(2025)\citenamefont {Fuller},
  \citenamefont {Tran}, \citenamefont {Lykov}, \citenamefont {Johnson},
  \citenamefont {Rossmannek}, \citenamefont {Wei}, \citenamefont {He},
  \citenamefont {Kim}, \citenamefont {Vu}, \citenamefont {Sharma} \emph
  {et~al.}}]{fuller2025improved}%
  \BibitemOpen
  \bibfield  {author} {\bibinfo {author} {\bibfnamefont {B.}~\bibnamefont
  {Fuller}}, \bibinfo {author} {\bibfnamefont {M.~C.}\ \bibnamefont {Tran}},
  \bibinfo {author} {\bibfnamefont {D.}~\bibnamefont {Lykov}}, \bibinfo
  {author} {\bibfnamefont {C.}~\bibnamefont {Johnson}}, \bibinfo {author}
  {\bibfnamefont {M.}~\bibnamefont {Rossmannek}}, \bibinfo {author}
  {\bibfnamefont {K.~X.}\ \bibnamefont {Wei}}, \bibinfo {author} {\bibfnamefont
  {A.}~\bibnamefont {He}}, \bibinfo {author} {\bibfnamefont {Y.}~\bibnamefont
  {Kim}}, \bibinfo {author} {\bibfnamefont {D.}~\bibnamefont {Vu}}, \bibinfo
  {author} {\bibfnamefont {K.}~\bibnamefont {Sharma}}, \emph {et~al.},\
  }\bibfield  {title} {\bibinfo {title} {Improved quantum computation using
  operator backpropagation},\ }\bibfield  {journal} {\bibinfo  {journal} {arXiv
  preprint arXiv:2502.01897}\ }\href
  {https://doi.org/https://doi.org/10.48550/arXiv.2502.01897}
  {https://doi.org/10.48550/arXiv.2502.01897} (\bibinfo {year}
  {2025})\BibitemShut {NoStop}%
\bibitem [{\citenamefont {Rudolph}\ \emph {et~al.}(2025)\citenamefont
  {Rudolph}, \citenamefont {Jones}, \citenamefont {Teng}, \citenamefont
  {Angrisani},\ and\ \citenamefont {Holmes}}]{rudolph2025pauli}%
  \BibitemOpen
  \bibfield  {author} {\bibinfo {author} {\bibfnamefont {M.~S.}\ \bibnamefont
  {Rudolph}}, \bibinfo {author} {\bibfnamefont {T.}~\bibnamefont {Jones}},
  \bibinfo {author} {\bibfnamefont {Y.}~\bibnamefont {Teng}}, \bibinfo {author}
  {\bibfnamefont {A.}~\bibnamefont {Angrisani}},\ and\ \bibinfo {author}
  {\bibfnamefont {Z.}~\bibnamefont {Holmes}},\ }\bibfield  {title} {\bibinfo
  {title} {Pauli propagation: A computational framework for simulating quantum
  systems},\ }\href {https://arxiv.org/abs/2505.21606} {\bibfield  {journal}
  {\bibinfo  {journal} {arXiv preprint arXiv:2505.21606}\ } (\bibinfo {year}
  {2025})}\BibitemShut {NoStop}%
\bibitem [{\citenamefont {Yu}\ \emph {et~al.}(2025)\citenamefont {Yu},
  \citenamefont {Moreno}, \citenamefont {Iosue}, \citenamefont {Bertels},
  \citenamefont {Claudino}, \citenamefont {Fuller}, \citenamefont
  {Groszkowski}, \citenamefont {Humble}, \citenamefont {Jurcevic},
  \citenamefont {Kirby} \emph {et~al.}}]{yu2025quantum}%
  \BibitemOpen
  \bibfield  {author} {\bibinfo {author} {\bibfnamefont {J.}~\bibnamefont
  {Yu}}, \bibinfo {author} {\bibfnamefont {J.~R.}\ \bibnamefont {Moreno}},
  \bibinfo {author} {\bibfnamefont {J.~T.}\ \bibnamefont {Iosue}}, \bibinfo
  {author} {\bibfnamefont {L.}~\bibnamefont {Bertels}}, \bibinfo {author}
  {\bibfnamefont {D.}~\bibnamefont {Claudino}}, \bibinfo {author}
  {\bibfnamefont {B.}~\bibnamefont {Fuller}}, \bibinfo {author} {\bibfnamefont
  {P.}~\bibnamefont {Groszkowski}}, \bibinfo {author} {\bibfnamefont {T.~S.}\
  \bibnamefont {Humble}}, \bibinfo {author} {\bibfnamefont {P.}~\bibnamefont
  {Jurcevic}}, \bibinfo {author} {\bibfnamefont {W.}~\bibnamefont {Kirby}},
  \emph {et~al.},\ }\bibfield  {title} {\bibinfo {title} {Quantum-centric
  algorithm for sample-based krylov diagonalization},\ }\href@noop {}
  {\bibfield  {journal} {\bibinfo  {journal} {arXiv preprint arXiv:2501.09702}\
  } (\bibinfo {year} {2025})}\BibitemShut {NoStop}%
\bibitem [{\citenamefont {Yeter-Aydeniz}\ \emph {et~al.}(2020)\citenamefont
  {Yeter-Aydeniz}, \citenamefont {Pooser},\ and\ \citenamefont
  {Siopsis}}]{yeter2020practical}%
  \BibitemOpen
  \bibfield  {author} {\bibinfo {author} {\bibfnamefont {K.}~\bibnamefont
  {Yeter-Aydeniz}}, \bibinfo {author} {\bibfnamefont {R.~C.}\ \bibnamefont
  {Pooser}},\ and\ \bibinfo {author} {\bibfnamefont {G.}~\bibnamefont
  {Siopsis}},\ }\bibfield  {title} {\bibinfo {title} {{Practical quantum
  computation of chemical and nuclear energy levels using quantum imaginary
  time evolution and Lanczos algorithms}},\ }\href
  {https://doi.org/10.1038/s41534-020-00290-1} {\bibfield  {journal} {\bibinfo
  {journal} {npj Quantum Information}\ }\textbf {\bibinfo {volume} {6}},\
  \bibinfo {pages} {63} (\bibinfo {year} {2020})}\BibitemShut {NoStop}%
\bibitem [{\citenamefont {Yeter-Aydeniz}\ \emph {et~al.}(2021)\citenamefont
  {Yeter-Aydeniz}, \citenamefont {Gard}, \citenamefont {Jakowski},
  \citenamefont {Majumder}, \citenamefont {Barron}, \citenamefont {Siopsis},
  \citenamefont {Humble},\ and\ \citenamefont
  {Pooser}}]{yeter2021benchmarking}%
  \BibitemOpen
  \bibfield  {author} {\bibinfo {author} {\bibfnamefont {K.}~\bibnamefont
  {Yeter-Aydeniz}}, \bibinfo {author} {\bibfnamefont {B.~T.}\ \bibnamefont
  {Gard}}, \bibinfo {author} {\bibfnamefont {J.}~\bibnamefont {Jakowski}},
  \bibinfo {author} {\bibfnamefont {S.}~\bibnamefont {Majumder}}, \bibinfo
  {author} {\bibfnamefont {G.~S.}\ \bibnamefont {Barron}}, \bibinfo {author}
  {\bibfnamefont {G.}~\bibnamefont {Siopsis}}, \bibinfo {author} {\bibfnamefont
  {T.~S.}\ \bibnamefont {Humble}},\ and\ \bibinfo {author} {\bibfnamefont
  {R.~C.}\ \bibnamefont {Pooser}},\ }\bibfield  {title} {\bibinfo {title}
  {Benchmarking quantum chemistry computations with variational, imaginary time
  evolution, and krylov space solver algorithms},\ }\href
  {https://doi.org/10.1002/qute.202100012} {\bibfield  {journal} {\bibinfo
  {journal} {Advanced Quantum Technologies}\ }\textbf {\bibinfo {volume} {4}},\
  \bibinfo {pages} {2100012} (\bibinfo {year} {2021})}\BibitemShut {NoStop}%
\bibitem [{\citenamefont {LaRose}\ \emph {et~al.}(2019)\citenamefont {LaRose},
  \citenamefont {Tikku}, \citenamefont {O'Neel-Judy}, \citenamefont {Cincio},\
  and\ \citenamefont {Coles}}]{larose2019variational}%
  \BibitemOpen
  \bibfield  {author} {\bibinfo {author} {\bibfnamefont {R.}~\bibnamefont
  {LaRose}}, \bibinfo {author} {\bibfnamefont {A.}~\bibnamefont {Tikku}},
  \bibinfo {author} {\bibfnamefont {{\'E}.}~\bibnamefont {O'Neel-Judy}},
  \bibinfo {author} {\bibfnamefont {L.}~\bibnamefont {Cincio}},\ and\ \bibinfo
  {author} {\bibfnamefont {P.~J.}\ \bibnamefont {Coles}},\ }\bibfield  {title}
  {\bibinfo {title} {Variational quantum state diagonalization},\ }\href
  {https://doi.org/10.1038/s41534-019-0167-6} {\bibfield  {journal} {\bibinfo
  {journal} {npj Quantum Information}\ }\textbf {\bibinfo {volume} {5}},\
  \bibinfo {pages} {1} (\bibinfo {year} {2019})}\BibitemShut {NoStop}%
\bibitem [{\citenamefont {Parrish}\ \emph {et~al.}(2019)\citenamefont
  {Parrish}, \citenamefont {Hohenstein}, \citenamefont {McMahon},\ and\
  \citenamefont {Mart{\'\i}nez}}]{parrish2019quantum}%
  \BibitemOpen
  \bibfield  {author} {\bibinfo {author} {\bibfnamefont {R.~M.}\ \bibnamefont
  {Parrish}}, \bibinfo {author} {\bibfnamefont {E.~G.}\ \bibnamefont
  {Hohenstein}}, \bibinfo {author} {\bibfnamefont {P.~L.}\ \bibnamefont
  {McMahon}},\ and\ \bibinfo {author} {\bibfnamefont {T.~J.}\ \bibnamefont
  {Mart{\'\i}nez}},\ }\bibfield  {title} {\bibinfo {title} {Quantum computation
  of electronic transitions using a variational quantum eigensolver},\ }\href
  {https://doi.org/10.1103/PhysRevLett.122.230401} {\bibfield  {journal}
  {\bibinfo  {journal} {Physical review letters}\ }\textbf {\bibinfo {volume}
  {122}},\ \bibinfo {pages} {230401} (\bibinfo {year} {2019})}\BibitemShut
  {NoStop}%
\bibitem [{\citenamefont {Stair}\ \emph {et~al.}(2020)\citenamefont {Stair},
  \citenamefont {Huang},\ and\ \citenamefont
  {Evangelista}}]{stair2020multireference}%
  \BibitemOpen
  \bibfield  {author} {\bibinfo {author} {\bibfnamefont {N.~H.}\ \bibnamefont
  {Stair}}, \bibinfo {author} {\bibfnamefont {R.}~\bibnamefont {Huang}},\ and\
  \bibinfo {author} {\bibfnamefont {F.~A.}\ \bibnamefont {Evangelista}},\
  }\bibfield  {title} {\bibinfo {title} {{A multireference quantum krylov
  algorithm for strongly correlated electrons}},\ }\href
  {https://doi.org/https://doi.org/10.1021/acs.jctc.9b01125} {\bibfield
  {journal} {\bibinfo  {journal} {Journal of chemical theory and computation}\
  }\textbf {\bibinfo {volume} {16}},\ \bibinfo {pages} {2236} (\bibinfo {year}
  {2020})}\BibitemShut {NoStop}%
\bibitem [{\citenamefont {Kokail}\ \emph {et~al.}(2019)\citenamefont {Kokail},
  \citenamefont {Maier}, \citenamefont {van Bijnen}, \citenamefont {Brydges},
  \citenamefont {Joshi}, \citenamefont {Jurcevic}, \citenamefont {Muschik},
  \citenamefont {Silvi}, \citenamefont {Blatt}, \citenamefont {Roos} \emph
  {et~al.}}]{kokail2019self}%
  \BibitemOpen
  \bibfield  {author} {\bibinfo {author} {\bibfnamefont {C.}~\bibnamefont
  {Kokail}}, \bibinfo {author} {\bibfnamefont {C.}~\bibnamefont {Maier}},
  \bibinfo {author} {\bibfnamefont {R.}~\bibnamefont {van Bijnen}}, \bibinfo
  {author} {\bibfnamefont {T.}~\bibnamefont {Brydges}}, \bibinfo {author}
  {\bibfnamefont {M.~K.}\ \bibnamefont {Joshi}}, \bibinfo {author}
  {\bibfnamefont {P.}~\bibnamefont {Jurcevic}}, \bibinfo {author}
  {\bibfnamefont {C.~A.}\ \bibnamefont {Muschik}}, \bibinfo {author}
  {\bibfnamefont {P.}~\bibnamefont {Silvi}}, \bibinfo {author} {\bibfnamefont
  {R.}~\bibnamefont {Blatt}}, \bibinfo {author} {\bibfnamefont {C.~F.}\
  \bibnamefont {Roos}}, \emph {et~al.},\ }\bibfield  {title} {\bibinfo {title}
  {Self-verifying variational quantum simulation of lattice models},\ }\href
  {https://doi.org/10.1038/s41586-019-1177-4} {\bibfield  {journal} {\bibinfo
  {journal} {Nature}\ }\textbf {\bibinfo {volume} {569}},\ \bibinfo {pages}
  {355} (\bibinfo {year} {2019})}\BibitemShut {NoStop}%
\bibitem [{\citenamefont {Arute}\ \emph {et~al.}(2020)\citenamefont {Arute},
  \citenamefont {Arya}, \citenamefont {Babbush}, \citenamefont {Bacon},
  \citenamefont {Bardin}, \citenamefont {Barends}, \citenamefont {Boixo},
  \citenamefont {Broughton}, \citenamefont {Buckley} \emph
  {et~al.}}]{google2020hartree}%
  \BibitemOpen
  \bibfield  {author} {\bibinfo {author} {\bibfnamefont {F.}~\bibnamefont
  {Arute}}, \bibinfo {author} {\bibfnamefont {K.}~\bibnamefont {Arya}},
  \bibinfo {author} {\bibfnamefont {R.}~\bibnamefont {Babbush}}, \bibinfo
  {author} {\bibfnamefont {D.}~\bibnamefont {Bacon}}, \bibinfo {author}
  {\bibfnamefont {J.~C.}\ \bibnamefont {Bardin}}, \bibinfo {author}
  {\bibfnamefont {R.}~\bibnamefont {Barends}}, \bibinfo {author} {\bibfnamefont
  {S.}~\bibnamefont {Boixo}}, \bibinfo {author} {\bibfnamefont
  {M.}~\bibnamefont {Broughton}}, \bibinfo {author} {\bibfnamefont {B.~B.}\
  \bibnamefont {Buckley}}, \emph {et~al.},\ }\bibfield  {title} {\bibinfo
  {title} {Hartree-fock on a superconducting qubit quantum computer},\ }\href
  {https://doi.org/https://doi.org/10.1126/science.abb9811} {\bibfield
  {journal} {\bibinfo  {journal} {Science}\ }\textbf {\bibinfo {volume}
  {369}},\ \bibinfo {pages} {1084} (\bibinfo {year} {2020})}\BibitemShut
  {NoStop}%
\bibitem [{\citenamefont {Xiaoyue}\ \emph {et~al.}(2024)\citenamefont
  {Xiaoyue}, \citenamefont {Robbiati}, \citenamefont {Pasquale}, \citenamefont
  {Pedicillo}, \citenamefont {Wright}, \citenamefont {Carrazza},\ and\
  \citenamefont {Gluza}}]{xiaoyue2024strategies}%
  \BibitemOpen
  \bibfield  {author} {\bibinfo {author} {\bibfnamefont {L.}~\bibnamefont
  {Xiaoyue}}, \bibinfo {author} {\bibfnamefont {M.}~\bibnamefont {Robbiati}},
  \bibinfo {author} {\bibfnamefont {A.}~\bibnamefont {Pasquale}}, \bibinfo
  {author} {\bibfnamefont {E.}~\bibnamefont {Pedicillo}}, \bibinfo {author}
  {\bibfnamefont {A.}~\bibnamefont {Wright}}, \bibinfo {author} {\bibfnamefont
  {S.}~\bibnamefont {Carrazza}},\ and\ \bibinfo {author} {\bibfnamefont
  {M.}~\bibnamefont {Gluza}},\ }\bibfield  {title} {\bibinfo {title}
  {Strategies for optimizing double-bracket quantum algorithms},\ }\href@noop
  {} {\bibfield  {journal} {\bibinfo  {journal} {arXiv preprint
  arXiv:2408.07431}\ } (\bibinfo {year} {2024})}\BibitemShut {NoStop}%
\bibitem [{\citenamefont {Cerezo}\ \emph {et~al.}(2021)\citenamefont {Cerezo},
  \citenamefont {Arrasmith}, \citenamefont {Babbush}, \citenamefont {Benjamin},
  \citenamefont {Endo}, \citenamefont {Fujii}, \citenamefont {McClean},
  \citenamefont {Mitarai}, \citenamefont {Yuan}, \citenamefont {Cincio} \emph
  {et~al.}}]{cerezo2021variational}%
  \BibitemOpen
  \bibfield  {author} {\bibinfo {author} {\bibfnamefont {M.}~\bibnamefont
  {Cerezo}}, \bibinfo {author} {\bibfnamefont {A.}~\bibnamefont {Arrasmith}},
  \bibinfo {author} {\bibfnamefont {R.}~\bibnamefont {Babbush}}, \bibinfo
  {author} {\bibfnamefont {S.~C.}\ \bibnamefont {Benjamin}}, \bibinfo {author}
  {\bibfnamefont {S.}~\bibnamefont {Endo}}, \bibinfo {author} {\bibfnamefont
  {K.}~\bibnamefont {Fujii}}, \bibinfo {author} {\bibfnamefont {J.~R.}\
  \bibnamefont {McClean}}, \bibinfo {author} {\bibfnamefont {K.}~\bibnamefont
  {Mitarai}}, \bibinfo {author} {\bibfnamefont {X.}~\bibnamefont {Yuan}},
  \bibinfo {author} {\bibfnamefont {L.}~\bibnamefont {Cincio}}, \emph
  {et~al.},\ }\bibfield  {title} {\bibinfo {title} {Variational quantum
  algorithms},\ }\href@noop {} {\bibfield  {journal} {\bibinfo  {journal}
  {Nature Reviews Physics}\ }\textbf {\bibinfo {volume} {3}},\ \bibinfo {pages}
  {625} (\bibinfo {year} {2021})}\BibitemShut {NoStop}%
\bibitem [{\citenamefont {Nielsen}\ and\ \citenamefont
  {Chuang}(2000)}]{nielsen2000quantum}%
  \BibitemOpen
  \bibfield  {author} {\bibinfo {author} {\bibfnamefont {M.~A.}\ \bibnamefont
  {Nielsen}}\ and\ \bibinfo {author} {\bibfnamefont {I.~L.}\ \bibnamefont
  {Chuang}},\ }\href@noop {} {\emph {\bibinfo {title} {Quantum Computation and
  Quantum Information}}}\ (\bibinfo  {publisher} {Cambridge University Press},\
  \bibinfo {address} {Cambridge},\ \bibinfo {year} {2000})\BibitemShut
  {NoStop}%
\bibitem [{\citenamefont {Temme}\ \emph {et~al.}(2011)\citenamefont {Temme},
  \citenamefont {Osborne}, \citenamefont {Vollbrecht}, \citenamefont {Poulin},\
  and\ \citenamefont {Verstraete}}]{temme2011quantum}%
  \BibitemOpen
  \bibfield  {author} {\bibinfo {author} {\bibfnamefont {K.}~\bibnamefont
  {Temme}}, \bibinfo {author} {\bibfnamefont {T.~J.}\ \bibnamefont {Osborne}},
  \bibinfo {author} {\bibfnamefont {K.~G.}\ \bibnamefont {Vollbrecht}},
  \bibinfo {author} {\bibfnamefont {D.}~\bibnamefont {Poulin}},\ and\ \bibinfo
  {author} {\bibfnamefont {F.}~\bibnamefont {Verstraete}},\ }\bibfield  {title}
  {\bibinfo {title} {Quantum metropolis sampling},\ }\href
  {https://doi.org/https://doi.org/10.1038/nature09770} {\bibfield  {journal}
  {\bibinfo  {journal} {Nature}\ }\textbf {\bibinfo {volume} {471}},\ \bibinfo
  {pages} {87} (\bibinfo {year} {2011})}\BibitemShut {NoStop}%
\bibitem [{\citenamefont {Kitaev}(1995)}]{kitaev1995quantum}%
  \BibitemOpen
  \bibfield  {author} {\bibinfo {author} {\bibfnamefont {A.~Y.}\ \bibnamefont
  {Kitaev}},\ }\bibfield  {title} {\bibinfo {title} {Quantum measurements and
  the abelian stabilizer problem},\ }\bibfield  {journal} {\bibinfo  {journal}
  {arXiv preprint quant-ph/9511026}\ }\href
  {https://doi.org/10.48550/arXiv.quant-ph/9511026}
  {10.48550/arXiv.quant-ph/9511026} (\bibinfo {year} {1995})\BibitemShut
  {NoStop}%
\bibitem [{\citenamefont {Brassard}\ \emph {et~al.}(2002)\citenamefont
  {Brassard}, \citenamefont {Hoyer}, \citenamefont {Mosca},\ and\ \citenamefont
  {Tapp}}]{brassard2002quantum}%
  \BibitemOpen
  \bibfield  {author} {\bibinfo {author} {\bibfnamefont {G.}~\bibnamefont
  {Brassard}}, \bibinfo {author} {\bibfnamefont {P.}~\bibnamefont {Hoyer}},
  \bibinfo {author} {\bibfnamefont {M.}~\bibnamefont {Mosca}},\ and\ \bibinfo
  {author} {\bibfnamefont {A.}~\bibnamefont {Tapp}},\ }\bibfield  {title}
  {\bibinfo {title} {Quantum amplitude amplification and estimation},\ }\href
  {https://doi.org/https://doi.org/10.1090/conm/305/05215} {\bibfield
  {journal} {\bibinfo  {journal} {Contemporary Mathematics}\ }\textbf {\bibinfo
  {volume} {305}},\ \bibinfo {pages} {53} (\bibinfo {year} {2002})}\BibitemShut
  {NoStop}%
\bibitem [{\citenamefont {Tan}\ \emph {et~al.}(2020)\citenamefont {Tan},
  \citenamefont {Bowmick},\ and\ \citenamefont {Sengupta}}]{tan2020quantum}%
  \BibitemOpen
  \bibfield  {author} {\bibinfo {author} {\bibfnamefont {K.~C.}\ \bibnamefont
  {Tan}}, \bibinfo {author} {\bibfnamefont {D.}~\bibnamefont {Bowmick}},\ and\
  \bibinfo {author} {\bibfnamefont {P.}~\bibnamefont {Sengupta}},\ }\href
  {https://arxiv.org/abs/2010.00949} {\bibinfo {title} {Quantum stochastic
  series expansion methods}} (\bibinfo {year} {2020}),\ \Eprint
  {https://arxiv.org/abs/2010.00949} {arXiv:2010.00949} \BibitemShut {NoStop}%
\bibitem [{\citenamefont {Low}\ and\ \citenamefont
  {Chuang}(2019)}]{Low2019hamiltonian}%
  \BibitemOpen
  \bibfield  {author} {\bibinfo {author} {\bibfnamefont {G.~H.}\ \bibnamefont
  {Low}}\ and\ \bibinfo {author} {\bibfnamefont {I.~L.}\ \bibnamefont
  {Chuang}},\ }\bibfield  {title} {\bibinfo {title} {Hamiltonian simulation by
  qubitization},\ }\href {https://doi.org/10.22331/q-2019-07-12-163} {\bibfield
   {journal} {\bibinfo  {journal} {Quantum}\ }\textbf {\bibinfo {volume} {3}},\
  \bibinfo {pages} {163} (\bibinfo {year} {2019})}\BibitemShut {NoStop}%
\bibitem [{\citenamefont {Gily{\'e}n}\ \emph {et~al.}(2019)\citenamefont
  {Gily{\'e}n}, \citenamefont {Su}, \citenamefont {Low},\ and\ \citenamefont
  {Wiebe}}]{gilyen2019quantum}%
  \BibitemOpen
  \bibfield  {author} {\bibinfo {author} {\bibfnamefont {A.}~\bibnamefont
  {Gily{\'e}n}}, \bibinfo {author} {\bibfnamefont {Y.}~\bibnamefont {Su}},
  \bibinfo {author} {\bibfnamefont {G.~H.}\ \bibnamefont {Low}},\ and\ \bibinfo
  {author} {\bibfnamefont {N.}~\bibnamefont {Wiebe}},\ }\bibfield  {title}
  {\bibinfo {title} {Quantum singular value transformation and beyond:
  exponential improvements for quantum matrix arithmetics},\ }in\ \href
  {https://doi.org/10.1145/3313276.3316366} {\emph {\bibinfo {booktitle}
  {Proceedings of the 51st Annual ACM SIGACT Symposium on Theory of
  Computing}}}\ (\bibinfo {year} {2019})\ pp.\ \bibinfo {pages}
  {193--204}\BibitemShut {NoStop}%
\bibitem [{\citenamefont {Ge}\ \emph {et~al.}(2019)\citenamefont {Ge},
  \citenamefont {Tura},\ and\ \citenamefont {Cirac}}]{ge2019faster}%
  \BibitemOpen
  \bibfield  {author} {\bibinfo {author} {\bibfnamefont {Y.}~\bibnamefont
  {Ge}}, \bibinfo {author} {\bibfnamefont {J.}~\bibnamefont {Tura}},\ and\
  \bibinfo {author} {\bibfnamefont {J.~I.}\ \bibnamefont {Cirac}},\ }\bibfield
  {title} {\bibinfo {title} {Faster ground state preparation and high-precision
  ground energy estimation with fewer qubits},\ }\href
  {https://doi.org/https://doi.org/10.1063/1.5027484} {\bibfield  {journal}
  {\bibinfo  {journal} {Journal of Mathematical Physics}\ }\textbf {\bibinfo
  {volume} {60}},\ \bibinfo {pages} {022202} (\bibinfo {year}
  {2019})}\BibitemShut {NoStop}%
\bibitem [{\citenamefont {Motta}\ \emph {et~al.}(2020)\citenamefont {Motta},
  \citenamefont {Sun}, \citenamefont {Tan}, \citenamefont {O’Rourke},
  \citenamefont {Ye}, \citenamefont {Minnich}, \citenamefont {Brandao},\ and\
  \citenamefont {Chan}}]{motta2020determining}%
  \BibitemOpen
  \bibfield  {author} {\bibinfo {author} {\bibfnamefont {M.}~\bibnamefont
  {Motta}}, \bibinfo {author} {\bibfnamefont {C.}~\bibnamefont {Sun}}, \bibinfo
  {author} {\bibfnamefont {A.~T.}\ \bibnamefont {Tan}}, \bibinfo {author}
  {\bibfnamefont {M.~J.}\ \bibnamefont {O’Rourke}}, \bibinfo {author}
  {\bibfnamefont {E.}~\bibnamefont {Ye}}, \bibinfo {author} {\bibfnamefont
  {A.~J.}\ \bibnamefont {Minnich}}, \bibinfo {author} {\bibfnamefont {F.~G.}\
  \bibnamefont {Brandao}},\ and\ \bibinfo {author} {\bibfnamefont {G.~K.-L.}\
  \bibnamefont {Chan}},\ }\bibfield  {title} {\bibinfo {title} {Determining
  eigenstates and thermal states on a quantum computer using quantum imaginary
  time evolution},\ }\href {https://doi.org/10.1038/s41567-019-0704-4}
  {\bibfield  {journal} {\bibinfo  {journal} {Nature Physics}\ }\textbf
  {\bibinfo {volume} {16}},\ \bibinfo {pages} {205} (\bibinfo {year}
  {2020})}\BibitemShut {NoStop}%
\bibitem [{\citenamefont {Motlagh}\ \emph {et~al.}(2024)\citenamefont
  {Motlagh}, \citenamefont {Zini}, \citenamefont {Arrazola},\ and\
  \citenamefont {Wiebe}}]{motlagh2024ground}%
  \BibitemOpen
  \bibfield  {author} {\bibinfo {author} {\bibfnamefont {D.}~\bibnamefont
  {Motlagh}}, \bibinfo {author} {\bibfnamefont {M.~S.}\ \bibnamefont {Zini}},
  \bibinfo {author} {\bibfnamefont {J.~M.}\ \bibnamefont {Arrazola}},\ and\
  \bibinfo {author} {\bibfnamefont {N.}~\bibnamefont {Wiebe}},\ }\bibfield
  {title} {\bibinfo {title} {Ground state preparation via dynamical cooling},\
  }\bibfield  {journal} {\bibinfo  {journal} {arXiv preprint arXiv:2404.05810}\
  }\href {https://doi.org/https://doi.org/10.48550/arXiv.2404.05810}
  {https://doi.org/10.48550/arXiv.2404.05810} (\bibinfo {year}
  {2024})\BibitemShut {NoStop}%
\bibitem [{\citenamefont {Gluza}\ \emph {et~al.}(2024)\citenamefont {Gluza},
  \citenamefont {Son}, \citenamefont {Tiang}, \citenamefont {Suzuki},
  \citenamefont {Holmes},\ and\ \citenamefont {Ng}}]{gluza2024double}%
  \BibitemOpen
  \bibfield  {author} {\bibinfo {author} {\bibfnamefont {M.}~\bibnamefont
  {Gluza}}, \bibinfo {author} {\bibfnamefont {J.}~\bibnamefont {Son}}, \bibinfo
  {author} {\bibfnamefont {B.~H.}\ \bibnamefont {Tiang}}, \bibinfo {author}
  {\bibfnamefont {Y.}~\bibnamefont {Suzuki}}, \bibinfo {author} {\bibfnamefont
  {Z.}~\bibnamefont {Holmes}},\ and\ \bibinfo {author} {\bibfnamefont {N.~H.}\
  \bibnamefont {Ng}},\ }\bibfield  {title} {\bibinfo {title} {Double-bracket
  quantum algorithms for quantum imaginary-time evolution},\ }\bibfield
  {journal} {\bibinfo  {journal} {arXiv preprint arXiv:2412.04554}\ }\href
  {https://doi.org/https://doi.org/10.48550/arXiv.2412.04554}
  {https://doi.org/10.48550/arXiv.2412.04554} (\bibinfo {year}
  {2024})\BibitemShut {NoStop}%
\bibitem [{\citenamefont {Suzuki}\ \emph {et~al.}(2025)\citenamefont {Suzuki},
  \citenamefont {Tiang}, \citenamefont {Son}, \citenamefont {Ng}, \citenamefont
  {Holmes},\ and\ \citenamefont {Gluza}}]{suzuki2025double}%
  \BibitemOpen
  \bibfield  {author} {\bibinfo {author} {\bibfnamefont {Y.}~\bibnamefont
  {Suzuki}}, \bibinfo {author} {\bibfnamefont {B.~H.}\ \bibnamefont {Tiang}},
  \bibinfo {author} {\bibfnamefont {J.}~\bibnamefont {Son}}, \bibinfo {author}
  {\bibfnamefont {N.~H.}\ \bibnamefont {Ng}}, \bibinfo {author} {\bibfnamefont
  {Z.}~\bibnamefont {Holmes}},\ and\ \bibinfo {author} {\bibfnamefont
  {M.}~\bibnamefont {Gluza}},\ }\bibfield  {title} {\bibinfo {title}
  {Double-bracket algorithm for quantum signal processing without
  post-selection},\ }\bibfield  {journal} {\bibinfo  {journal} {arXiv preprint
  arXiv:2504.01077}\ }\href
  {https://doi.org/https://doi.org/10.48550/arXiv.2504.01077}
  {https://doi.org/10.48550/arXiv.2504.01077} (\bibinfo {year}
  {2025})\BibitemShut {NoStop}%
\bibitem [{\citenamefont {Lee}\ \emph {et~al.}(2023)\citenamefont {Lee},
  \citenamefont {Lee}, \citenamefont {Zhai}, \citenamefont {Tong},
  \citenamefont {Dalzell}, \citenamefont {Kumar}, \citenamefont {Helms},
  \citenamefont {Gray}, \citenamefont {Cui}, \citenamefont {Liu} \emph
  {et~al.}}]{lee2023evaluating}%
  \BibitemOpen
  \bibfield  {author} {\bibinfo {author} {\bibfnamefont {S.}~\bibnamefont
  {Lee}}, \bibinfo {author} {\bibfnamefont {J.}~\bibnamefont {Lee}}, \bibinfo
  {author} {\bibfnamefont {H.}~\bibnamefont {Zhai}}, \bibinfo {author}
  {\bibfnamefont {Y.}~\bibnamefont {Tong}}, \bibinfo {author} {\bibfnamefont
  {A.~M.}\ \bibnamefont {Dalzell}}, \bibinfo {author} {\bibfnamefont
  {A.}~\bibnamefont {Kumar}}, \bibinfo {author} {\bibfnamefont
  {P.}~\bibnamefont {Helms}}, \bibinfo {author} {\bibfnamefont
  {J.}~\bibnamefont {Gray}}, \bibinfo {author} {\bibfnamefont {Z.-H.}\
  \bibnamefont {Cui}}, \bibinfo {author} {\bibfnamefont {W.}~\bibnamefont
  {Liu}}, \emph {et~al.},\ }\bibfield  {title} {\bibinfo {title} {Evaluating
  the evidence for exponential quantum advantage in ground-state quantum
  chemistry},\ }\href {https://doi.org/10.1038/s41467-023-37587-6} {\bibfield
  {journal} {\bibinfo  {journal} {Nature Communications}\ }\textbf {\bibinfo
  {volume} {14}},\ \bibinfo {pages} {1952} (\bibinfo {year}
  {2023})}\BibitemShut {NoStop}%
\bibitem [{\citenamefont {Zimbor{\'a}s}\ \emph {et~al.}(2025)\citenamefont
  {Zimbor{\'a}s}, \citenamefont {Koczor}, \citenamefont {Holmes}, \citenamefont
  {Borrelli}, \citenamefont {Gily{\'e}n}, \citenamefont {Huang}, \citenamefont
  {Cai}, \citenamefont {Ac{\'\i}n}, \citenamefont {Aolita}, \citenamefont
  {Banchi} \emph {et~al.}}]{zimboras2025myths}%
  \BibitemOpen
  \bibfield  {author} {\bibinfo {author} {\bibfnamefont {Z.}~\bibnamefont
  {Zimbor{\'a}s}}, \bibinfo {author} {\bibfnamefont {B.}~\bibnamefont
  {Koczor}}, \bibinfo {author} {\bibfnamefont {Z.}~\bibnamefont {Holmes}},
  \bibinfo {author} {\bibfnamefont {E.-M.}\ \bibnamefont {Borrelli}}, \bibinfo
  {author} {\bibfnamefont {A.}~\bibnamefont {Gily{\'e}n}}, \bibinfo {author}
  {\bibfnamefont {H.-Y.}\ \bibnamefont {Huang}}, \bibinfo {author}
  {\bibfnamefont {Z.}~\bibnamefont {Cai}}, \bibinfo {author} {\bibfnamefont
  {A.}~\bibnamefont {Ac{\'\i}n}}, \bibinfo {author} {\bibfnamefont
  {L.}~\bibnamefont {Aolita}}, \bibinfo {author} {\bibfnamefont
  {L.}~\bibnamefont {Banchi}}, \emph {et~al.},\ }\bibfield  {title} {\bibinfo
  {title} {Myths around quantum computation before full fault tolerance: What
  no-go theorems rule out and what they don't},\ }\bibfield  {journal}
  {\bibinfo  {journal} {arXiv preprint arXiv:2501.05694}\ }\href
  {https://doi.org/https://doi.org/10.48550/arXiv.2501.05694}
  {https://doi.org/10.48550/arXiv.2501.05694} (\bibinfo {year}
  {2025})\BibitemShut {NoStop}%
\bibitem [{\citenamefont {Miller}\ \emph {et~al.}(2023)\citenamefont {Miller},
  \citenamefont {Zimbor{\'a}s}, \citenamefont {Knecht}, \citenamefont
  {Maniscalco},\ and\ \citenamefont
  {Garc{\'\i}a-P{\'e}rez}}]{miller2023bonsai}%
  \BibitemOpen
  \bibfield  {author} {\bibinfo {author} {\bibfnamefont {A.}~\bibnamefont
  {Miller}}, \bibinfo {author} {\bibfnamefont {Z.}~\bibnamefont
  {Zimbor{\'a}s}}, \bibinfo {author} {\bibfnamefont {S.}~\bibnamefont
  {Knecht}}, \bibinfo {author} {\bibfnamefont {S.}~\bibnamefont {Maniscalco}},\
  and\ \bibinfo {author} {\bibfnamefont {G.}~\bibnamefont
  {Garc{\'\i}a-P{\'e}rez}},\ }\bibfield  {title} {\bibinfo {title} {Bonsai
  algorithm: Grow your own fermion-to-qubit mappings},\ }\href@noop {}
  {\bibfield  {journal} {\bibinfo  {journal} {PRX Quantum}\ }\textbf {\bibinfo
  {volume} {4}},\ \bibinfo {pages} {030314} (\bibinfo {year}
  {2023})}\BibitemShut {NoStop}%
\bibitem [{\citenamefont {Miller}\ \emph {et~al.}(2024)\citenamefont {Miller},
  \citenamefont {Glos},\ and\ \citenamefont
  {Zimbor{\'a}s}}]{miller2024treespilation}%
  \BibitemOpen
  \bibfield  {author} {\bibinfo {author} {\bibfnamefont {A.}~\bibnamefont
  {Miller}}, \bibinfo {author} {\bibfnamefont {A.}~\bibnamefont {Glos}},\ and\
  \bibinfo {author} {\bibfnamefont {Z.}~\bibnamefont {Zimbor{\'a}s}},\
  }\bibfield  {title} {\bibinfo {title} {Treespilation: Architecture-and
  state-optimised fermion-to-qubit mappings},\ }\href@noop {} {\bibfield
  {journal} {\bibinfo  {journal} {arXiv preprint arXiv:2403.03992}\ } (\bibinfo
  {year} {2024})}\BibitemShut {NoStop}%
\bibitem [{\citenamefont {Jiang}\ \emph {et~al.}(2020)\citenamefont {Jiang},
  \citenamefont {Kalev}, \citenamefont {Mruczkiewicz},\ and\ \citenamefont
  {Neven}}]{neven_mapping}%
  \BibitemOpen
  \bibfield  {author} {\bibinfo {author} {\bibfnamefont {Z.}~\bibnamefont
  {Jiang}}, \bibinfo {author} {\bibfnamefont {A.}~\bibnamefont {Kalev}},
  \bibinfo {author} {\bibfnamefont {W.}~\bibnamefont {Mruczkiewicz}},\ and\
  \bibinfo {author} {\bibfnamefont {H.}~\bibnamefont {Neven}},\ }\bibfield
  {title} {\bibinfo {title} {Optimal fermion-to-qubit mapping via ternary trees
  with applications to reduced quantum states learning},\ }\href@noop {}
  {\bibfield  {journal} {\bibinfo  {journal} {Quantum}\ }\textbf {\bibinfo
  {volume} {4}},\ \bibinfo {pages} {276} (\bibinfo {year} {2020})}\BibitemShut
  {NoStop}%
\bibitem [{\citenamefont {Jordan}\ and\ \citenamefont
  {Wigner}(1928)}]{jw_mapping}%
  \BibitemOpen
  \bibfield  {author} {\bibinfo {author} {\bibfnamefont {P.}~\bibnamefont
  {Jordan}}\ and\ \bibinfo {author} {\bibfnamefont {E.}~\bibnamefont
  {Wigner}},\ }\bibfield  {title} {\bibinfo {title} {Über das {Paulische}
  Äquivalenzverbot},\ }\href {https://doi.org/10.1007/BF01331938} {\bibfield
  {journal} {\bibinfo  {journal} {Zeitschrift für Physik}\ }\textbf {\bibinfo
  {volume} {47}},\ \bibinfo {pages} {631} (\bibinfo {year} {1928})}\BibitemShut
  {NoStop}%
\bibitem [{\citenamefont {Bravyi}\ and\ \citenamefont
  {Kitaev}(2002)}]{bk_mapping}%
  \BibitemOpen
  \bibfield  {author} {\bibinfo {author} {\bibfnamefont {S.~B.}\ \bibnamefont
  {Bravyi}}\ and\ \bibinfo {author} {\bibfnamefont {A.~Y.}\ \bibnamefont
  {Kitaev}},\ }\bibfield  {title} {\bibinfo {title} {Fermionic quantum
  computation},\ }\href
  {https://doi.org/https://doi.org/10.1006/aphy.2002.6254} {\bibfield
  {journal} {\bibinfo  {journal} {Annals of Physics}\ }\textbf {\bibinfo
  {volume} {298}},\ \bibinfo {pages} {210} (\bibinfo {year}
  {2002})}\BibitemShut {NoStop}%
\bibitem [{\citenamefont {Chiew}\ and\ \citenamefont
  {Strelchuk}(2023)}]{chiew2023discovering}%
  \BibitemOpen
  \bibfield  {author} {\bibinfo {author} {\bibfnamefont {M.}~\bibnamefont
  {Chiew}}\ and\ \bibinfo {author} {\bibfnamefont {S.}~\bibnamefont
  {Strelchuk}},\ }\bibfield  {title} {\bibinfo {title} {Discovering optimal
  fermion-qubit mappings through algorithmic enumeration},\ }\href@noop {}
  {\bibfield  {journal} {\bibinfo  {journal} {Quantum}\ }\textbf {\bibinfo
  {volume} {7}},\ \bibinfo {pages} {1145} (\bibinfo {year} {2023})}\BibitemShut
  {NoStop}%
\bibitem [{\citenamefont {Chien}\ and\ \citenamefont
  {Klassen}(2022)}]{chien2022optimizing}%
  \BibitemOpen
  \bibfield  {author} {\bibinfo {author} {\bibfnamefont {R.~W.}\ \bibnamefont
  {Chien}}\ and\ \bibinfo {author} {\bibfnamefont {J.}~\bibnamefont
  {Klassen}},\ }\bibfield  {title} {\bibinfo {title} {Optimizing fermionic
  encodings for both hamiltonian and hardware},\ }\href@noop {} {\bibfield
  {journal} {\bibinfo  {journal} {arXiv preprint arXiv:2210.05652}\ } (\bibinfo
  {year} {2022})}\BibitemShut {NoStop}%
\bibitem [{\citenamefont {McClean}\ \emph {et~al.}(2017)\citenamefont
  {McClean}, \citenamefont {Kimchi-Schwartz}, \citenamefont {Carter},\ and\
  \citenamefont {De~Jong}}]{mcclean2017hybrid}%
  \BibitemOpen
  \bibfield  {author} {\bibinfo {author} {\bibfnamefont {J.~R.}\ \bibnamefont
  {McClean}}, \bibinfo {author} {\bibfnamefont {M.~E.}\ \bibnamefont
  {Kimchi-Schwartz}}, \bibinfo {author} {\bibfnamefont {J.}~\bibnamefont
  {Carter}},\ and\ \bibinfo {author} {\bibfnamefont {W.~A.}\ \bibnamefont
  {De~Jong}},\ }\bibfield  {title} {\bibinfo {title} {Hybrid quantum-classical
  hierarchy for mitigation of decoherence and determination of excited
  states},\ }\href
  {https://journals.aps.org/pra/abstract/10.1103/PhysRevA.95.042308} {\bibfield
   {journal} {\bibinfo  {journal} {Physical Review A}\ }\textbf {\bibinfo
  {volume} {95}},\ \bibinfo {pages} {042308} (\bibinfo {year}
  {2017})}\BibitemShut {NoStop}%
\bibitem [{\citenamefont {Stilck~Fran{\c{c}}a}\ and\ \citenamefont
  {Garcia-Patron}(2021)}]{franca2020limitations}%
  \BibitemOpen
  \bibfield  {author} {\bibinfo {author} {\bibfnamefont {D.}~\bibnamefont
  {Stilck~Fran{\c{c}}a}}\ and\ \bibinfo {author} {\bibfnamefont
  {R.}~\bibnamefont {Garcia-Patron}},\ }\bibfield  {title} {\bibinfo {title}
  {Limitations of optimization algorithms on noisy quantum devices},\ }\href
  {https://doi.org/10.1038/s41567-021-01356-3} {\bibfield  {journal} {\bibinfo
  {journal} {Nature Physics}\ }\textbf {\bibinfo {volume} {17}},\ \bibinfo
  {pages} {1221} (\bibinfo {year} {2021})}\BibitemShut {NoStop}%
\bibitem [{\citenamefont {Anschuetz}\ and\ \citenamefont
  {Kiani}(2022)}]{anschuetz2022beyond}%
  \BibitemOpen
  \bibfield  {author} {\bibinfo {author} {\bibfnamefont {E.~R.}\ \bibnamefont
  {Anschuetz}}\ and\ \bibinfo {author} {\bibfnamefont {B.~T.}\ \bibnamefont
  {Kiani}},\ }\bibfield  {title} {\bibinfo {title} {Beyond barren plateaus:
  Quantum variational algorithms are swamped with traps},\ }\href
  {https://doi.org/10.1038/s41467-022-35364-5} {\bibfield  {journal} {\bibinfo
  {journal} {Nature Communications}\ }\textbf {\bibinfo {volume} {13}},\
  \bibinfo {pages} {7760} (\bibinfo {year} {2022})}\BibitemShut {NoStop}%
\bibitem [{\citenamefont {Larocca}\ and\ \citenamefont
  {Havlicek}(2024)}]{larocca2024quantum}%
  \BibitemOpen
  \bibfield  {author} {\bibinfo {author} {\bibfnamefont {M.}~\bibnamefont
  {Larocca}}\ and\ \bibinfo {author} {\bibfnamefont {V.}~\bibnamefont
  {Havlicek}},\ }\bibfield  {title} {\bibinfo {title} {Quantum algorithms for
  representation-theoretic multiplicities},\ }\href
  {https://arxiv.org/abs/2407.17649} {\bibfield  {journal} {\bibinfo  {journal}
  {arXiv preprint arXiv:2407.17649}\ } (\bibinfo {year} {2024})}\BibitemShut
  {NoStop}%
\bibitem [{\citenamefont {Cerezo}\ \emph {et~al.}(2025)\citenamefont {Cerezo},
  \citenamefont {Larocca}, \citenamefont {Garc{\'\i}a-Mart{\'\i}n},
  \citenamefont {Diaz}, \citenamefont {Braccia}, \citenamefont {Fontana},
  \citenamefont {Rudolph}, \citenamefont {Bermejo}, \citenamefont {Ijaz},
  \citenamefont {Thanasilp} \emph {et~al.}}]{cerezo2023does}%
  \BibitemOpen
  \bibfield  {author} {\bibinfo {author} {\bibfnamefont {M.}~\bibnamefont
  {Cerezo}}, \bibinfo {author} {\bibfnamefont {M.}~\bibnamefont {Larocca}},
  \bibinfo {author} {\bibfnamefont {D.}~\bibnamefont
  {Garc{\'\i}a-Mart{\'\i}n}}, \bibinfo {author} {\bibfnamefont {N.~L.}\
  \bibnamefont {Diaz}}, \bibinfo {author} {\bibfnamefont {P.}~\bibnamefont
  {Braccia}}, \bibinfo {author} {\bibfnamefont {E.}~\bibnamefont {Fontana}},
  \bibinfo {author} {\bibfnamefont {M.~S.}\ \bibnamefont {Rudolph}}, \bibinfo
  {author} {\bibfnamefont {P.}~\bibnamefont {Bermejo}}, \bibinfo {author}
  {\bibfnamefont {A.}~\bibnamefont {Ijaz}}, \bibinfo {author} {\bibfnamefont
  {S.}~\bibnamefont {Thanasilp}}, \emph {et~al.},\ }\bibfield  {title}
  {\bibinfo {title} {Does provable absence of barren plateaus imply classical
  simulability?},\ }\href {https://doi.org/10.1038/s41467-025-63099-6}
  {\bibfield  {journal} {\bibinfo  {journal} {Nature Communications}\ }\textbf
  {\bibinfo {volume} {16}},\ \bibinfo {pages} {7907} (\bibinfo {year}
  {2025})}\BibitemShut {NoStop}%
\bibitem [{\citenamefont {Grimsley}\ \emph {et~al.}(2019)\citenamefont
  {Grimsley}, \citenamefont {Economou}, \citenamefont {Barnes},\ and\
  \citenamefont {Mayhall}}]{grimsley2019adaptive}%
  \BibitemOpen
  \bibfield  {author} {\bibinfo {author} {\bibfnamefont {H.~R.}\ \bibnamefont
  {Grimsley}}, \bibinfo {author} {\bibfnamefont {S.~E.}\ \bibnamefont
  {Economou}}, \bibinfo {author} {\bibfnamefont {E.}~\bibnamefont {Barnes}},\
  and\ \bibinfo {author} {\bibfnamefont {N.~J.}\ \bibnamefont {Mayhall}},\
  }\bibfield  {title} {\bibinfo {title} {An adaptive variational algorithm for
  exact molecular simulations on a quantum computer},\ }\href
  {https://doi.org/10.1038/s41467-019-10988-2} {\bibfield  {journal} {\bibinfo
  {journal} {Nature {C}ommunications}\ }\textbf {\bibinfo {volume} {10}},\
  \bibinfo {pages} {1} (\bibinfo {year} {2019})}\BibitemShut {NoStop}%
\bibitem [{\citenamefont {Bermejo}\ \emph {et~al.}(2024)\citenamefont
  {Bermejo}, \citenamefont {Braccia}, \citenamefont {Rudolph}, \citenamefont
  {Holmes}, \citenamefont {Cincio},\ and\ \citenamefont
  {Cerezo}}]{bermejo2024quantum}%
  \BibitemOpen
  \bibfield  {author} {\bibinfo {author} {\bibfnamefont {P.}~\bibnamefont
  {Bermejo}}, \bibinfo {author} {\bibfnamefont {P.}~\bibnamefont {Braccia}},
  \bibinfo {author} {\bibfnamefont {M.~S.}\ \bibnamefont {Rudolph}}, \bibinfo
  {author} {\bibfnamefont {Z.}~\bibnamefont {Holmes}}, \bibinfo {author}
  {\bibfnamefont {L.}~\bibnamefont {Cincio}},\ and\ \bibinfo {author}
  {\bibfnamefont {M.}~\bibnamefont {Cerezo}},\ }\bibfield  {title} {\bibinfo
  {title} {Quantum convolutional neural networks are (effectively) classically
  simulable},\ }\bibfield  {journal} {\bibinfo  {journal} {arXiv preprint
  arXiv:2408.12739}\ }\href
  {https://doi.org/https://doi.org/10.48550/arXiv.2408.12739}
  {https://doi.org/10.48550/arXiv.2408.12739} (\bibinfo {year}
  {2024})\BibitemShut {NoStop}%
\bibitem [{\citenamefont {McClean}\ \emph {et~al.}(2018)\citenamefont
  {McClean}, \citenamefont {Boixo}, \citenamefont {Smelyanskiy}, \citenamefont
  {Babbush},\ and\ \citenamefont {Neven}}]{mcclean2018barren}%
  \BibitemOpen
  \bibfield  {author} {\bibinfo {author} {\bibfnamefont {J.~R.}\ \bibnamefont
  {McClean}}, \bibinfo {author} {\bibfnamefont {S.}~\bibnamefont {Boixo}},
  \bibinfo {author} {\bibfnamefont {V.~N.}\ \bibnamefont {Smelyanskiy}},
  \bibinfo {author} {\bibfnamefont {R.}~\bibnamefont {Babbush}},\ and\ \bibinfo
  {author} {\bibfnamefont {H.}~\bibnamefont {Neven}},\ }\bibfield  {title}
  {\bibinfo {title} {Barren plateaus in quantum neural network training
  landscapes},\ }\href {https://doi.org/10.1038/s41467-018-07090-4} {\bibfield
  {journal} {\bibinfo  {journal} {Nature {C}ommunications}\ }\textbf {\bibinfo
  {volume} {9}},\ \bibinfo {pages} {1} (\bibinfo {year} {2018})}\BibitemShut
  {NoStop}%
\bibitem [{\citenamefont {Larocca}\ \emph {et~al.}(2025)\citenamefont
  {Larocca}, \citenamefont {Thanasilp}, \citenamefont {Wang}, \citenamefont
  {Sharma}, \citenamefont {Biamonte}, \citenamefont {Coles}, \citenamefont
  {Cincio}, \citenamefont {McClean}, \citenamefont {Holmes},\ and\
  \citenamefont {Cerezo}}]{larocca2024review}%
  \BibitemOpen
  \bibfield  {author} {\bibinfo {author} {\bibfnamefont {M.}~\bibnamefont
  {Larocca}}, \bibinfo {author} {\bibfnamefont {S.}~\bibnamefont {Thanasilp}},
  \bibinfo {author} {\bibfnamefont {S.}~\bibnamefont {Wang}}, \bibinfo {author}
  {\bibfnamefont {K.}~\bibnamefont {Sharma}}, \bibinfo {author} {\bibfnamefont
  {J.}~\bibnamefont {Biamonte}}, \bibinfo {author} {\bibfnamefont {P.~J.}\
  \bibnamefont {Coles}}, \bibinfo {author} {\bibfnamefont {L.}~\bibnamefont
  {Cincio}}, \bibinfo {author} {\bibfnamefont {J.~R.}\ \bibnamefont {McClean}},
  \bibinfo {author} {\bibfnamefont {Z.}~\bibnamefont {Holmes}},\ and\ \bibinfo
  {author} {\bibfnamefont {M.}~\bibnamefont {Cerezo}},\ }\bibfield  {title}
  {\bibinfo {title} {A review of barren plateaus in variational quantum
  computing},\ }\href {https://doi.org/10.1038/s42254-025-00813-9} {\bibfield
  {journal} {\bibinfo  {journal} {Nature Reviews Physics}\ }\textbf {\bibinfo
  {volume} {3}},\ \bibinfo {pages} {625–644} (\bibinfo {year}
  {2025})}\BibitemShut {NoStop}%
\bibitem [{\citenamefont {Tang}\ \emph {et~al.}(2021)\citenamefont {Tang},
  \citenamefont {Shkolnikov}, \citenamefont {Barron}, \citenamefont {Grimsley},
  \citenamefont {Mayhall}, \citenamefont {Barnes},\ and\ \citenamefont
  {Economou}}]{Q-ADAPT}%
  \BibitemOpen
  \bibfield  {author} {\bibinfo {author} {\bibfnamefont {H.~L.}\ \bibnamefont
  {Tang}}, \bibinfo {author} {\bibfnamefont {V.}~\bibnamefont {Shkolnikov}},
  \bibinfo {author} {\bibfnamefont {G.~S.}\ \bibnamefont {Barron}}, \bibinfo
  {author} {\bibfnamefont {H.~R.}\ \bibnamefont {Grimsley}}, \bibinfo {author}
  {\bibfnamefont {N.~J.}\ \bibnamefont {Mayhall}}, \bibinfo {author}
  {\bibfnamefont {E.}~\bibnamefont {Barnes}},\ and\ \bibinfo {author}
  {\bibfnamefont {S.~E.}\ \bibnamefont {Economou}},\ }\bibfield  {title}
  {\bibinfo {title} {Qubit-adapt-vqe: An adaptive algorithm for constructing
  hardware-efficient ans\"atze on a quantum processor},\ }\href
  {https://doi.org/10.1103/PRXQuantum.2.020310} {\bibfield  {journal} {\bibinfo
   {journal} {PRX Quantum}\ }\textbf {\bibinfo {volume} {2}},\ \bibinfo {pages}
  {020310} (\bibinfo {year} {2021})}\BibitemShut {NoStop}%
\bibitem [{\citenamefont {Grimsley}\ \emph {et~al.}(2023)\citenamefont
  {Grimsley}, \citenamefont {Barron}, \citenamefont {Barnes}, \citenamefont
  {Economou},\ and\ \citenamefont {Mayhall}}]{Grimsley_2023}%
  \BibitemOpen
  \bibfield  {author} {\bibinfo {author} {\bibfnamefont {H.~R.}\ \bibnamefont
  {Grimsley}}, \bibinfo {author} {\bibfnamefont {G.~S.}\ \bibnamefont
  {Barron}}, \bibinfo {author} {\bibfnamefont {E.}~\bibnamefont {Barnes}},
  \bibinfo {author} {\bibfnamefont {S.~E.}\ \bibnamefont {Economou}},\ and\
  \bibinfo {author} {\bibfnamefont {N.~J.}\ \bibnamefont {Mayhall}},\
  }\bibfield  {title} {\bibinfo {title} {Adaptive, problem-tailored variational
  quantum eigensolver mitigates rough parameter landscapes and barren
  plateaus},\ }\bibfield  {journal} {\bibinfo  {journal} {npj Quantum
  Information}\ }\textbf {\bibinfo {volume} {9}},\ \href
  {https://doi.org/10.1038/s41534-023-00681-0} {10.1038/s41534-023-00681-0}
  (\bibinfo {year} {2023})\BibitemShut {NoStop}%
\bibitem [{\citenamefont {Anastasiou}\ \emph {et~al.}(2024)\citenamefont
  {Anastasiou}, \citenamefont {Chen}, \citenamefont {Mayhall}, \citenamefont
  {Barnes},\ and\ \citenamefont {Economou}}]{Anastasiou_2024}%
  \BibitemOpen
  \bibfield  {author} {\bibinfo {author} {\bibfnamefont {P.~G.}\ \bibnamefont
  {Anastasiou}}, \bibinfo {author} {\bibfnamefont {Y.}~\bibnamefont {Chen}},
  \bibinfo {author} {\bibfnamefont {N.~J.}\ \bibnamefont {Mayhall}}, \bibinfo
  {author} {\bibfnamefont {E.}~\bibnamefont {Barnes}},\ and\ \bibinfo {author}
  {\bibfnamefont {S.~E.}\ \bibnamefont {Economou}},\ }\bibfield  {title}
  {\bibinfo {title} {Tetris-adapt-vqe: An adaptive algorithm that yields
  shallower, denser circuit ansätze},\ }\bibfield  {journal} {\bibinfo
  {journal} {Physical Review Research}\ }\textbf {\bibinfo {volume} {6}},\
  \href {https://doi.org/10.1103/physrevresearch.6.013254}
  {10.1103/physrevresearch.6.013254} (\bibinfo {year} {2024})\BibitemShut
  {NoStop}%
\bibitem [{\citenamefont {Ramôa}\ \emph {et~al.}(2024)\citenamefont {Ramôa},
  \citenamefont {Anastasiou}, \citenamefont {Santos}, \citenamefont {Mayhall},
  \citenamefont {Barnes},\ and\ \citenamefont
  {Economou}}]{2024reducingresources}%
  \BibitemOpen
  \bibfield  {author} {\bibinfo {author} {\bibfnamefont {M.}~\bibnamefont
  {Ramôa}}, \bibinfo {author} {\bibfnamefont {P.~G.}\ \bibnamefont
  {Anastasiou}}, \bibinfo {author} {\bibfnamefont {L.~P.}\ \bibnamefont
  {Santos}}, \bibinfo {author} {\bibfnamefont {N.~J.}\ \bibnamefont {Mayhall}},
  \bibinfo {author} {\bibfnamefont {E.}~\bibnamefont {Barnes}},\ and\ \bibinfo
  {author} {\bibfnamefont {S.~E.}\ \bibnamefont {Economou}},\ }\href
  {https://arxiv.org/abs/2407.08696} {\bibinfo {title} {Reducing the resources
  required by adapt-vqe using coupled exchange operators and improved
  subroutines}} (\bibinfo {year} {2024}),\ \Eprint
  {https://arxiv.org/abs/2407.08696} {arXiv:2407.08696 [quant-ph]} \BibitemShut
  {NoStop}%
\bibitem [{\citenamefont {Inc.}(2024)}]{theralase2024}%
  \BibitemOpen
  \bibfield  {author} {\bibinfo {author} {\bibfnamefont {T.~T.}\ \bibnamefont
  {Inc.}},\ }\href
  {https://theralase.com/anti-cancer-therapy/clinical-studies/} {\bibinfo
  {title} {Clinical studies overview – {TLD-1433} ({R}uvidar™)}} (\bibinfo
  {year} {2024}),\ \bibinfo {note} {accessed: 2024-02-27}\BibitemShut {NoStop}%
\bibitem [{\citenamefont {Kuo}\ \emph {et~al.}(2024)\citenamefont {Kuo},
  \citenamefont {Ware}, \citenamefont {Lunts}, \citenamefont {Hafezi},\ and\
  \citenamefont {White}}]{kuo2024energy}%
  \BibitemOpen
  \bibfield  {author} {\bibinfo {author} {\bibfnamefont {E.-J.}\ \bibnamefont
  {Kuo}}, \bibinfo {author} {\bibfnamefont {B.}~\bibnamefont {Ware}}, \bibinfo
  {author} {\bibfnamefont {P.}~\bibnamefont {Lunts}}, \bibinfo {author}
  {\bibfnamefont {M.}~\bibnamefont {Hafezi}},\ and\ \bibinfo {author}
  {\bibfnamefont {C.~D.}\ \bibnamefont {White}},\ }\bibfield  {title} {\bibinfo
  {title} {Energy diffusion in weakly interacting chains with fermionic
  dissipation assisted operator evolution},\ }\href@noop {} {\bibfield
  {journal} {\bibinfo  {journal} {Physical Review B}\ }\textbf {\bibinfo
  {volume} {110}},\ \bibinfo {pages} {075149} (\bibinfo {year}
  {2024})}\BibitemShut {NoStop}%
\bibitem [{\citenamefont {Mullinax}\ and\ \citenamefont
  {Tubman}(2023)}]{mullinax2023large}%
  \BibitemOpen
  \bibfield  {author} {\bibinfo {author} {\bibfnamefont {J.~W.}\ \bibnamefont
  {Mullinax}}\ and\ \bibinfo {author} {\bibfnamefont {N.~M.}\ \bibnamefont
  {Tubman}},\ }\bibfield  {title} {\bibinfo {title} {Large-scale sparse
  wavefunction circuit simulator for applications with the variational quantum
  eigensolver},\ }\href@noop {} {\bibfield  {journal} {\bibinfo  {journal}
  {arXiv preprint arXiv:2301.05726}\ } (\bibinfo {year} {2023})}\BibitemShut
  {NoStop}%
\bibitem [{\citenamefont {Mullinax}\ \emph {et~al.}(2024)\citenamefont
  {Mullinax}, \citenamefont {Anastasiou}, \citenamefont {Larson}, \citenamefont
  {Economou},\ and\ \citenamefont {Tubman}}]{mullinax2024classical}%
  \BibitemOpen
  \bibfield  {author} {\bibinfo {author} {\bibfnamefont {J.~W.}\ \bibnamefont
  {Mullinax}}, \bibinfo {author} {\bibfnamefont {P.~G.}\ \bibnamefont
  {Anastasiou}}, \bibinfo {author} {\bibfnamefont {J.}~\bibnamefont {Larson}},
  \bibinfo {author} {\bibfnamefont {S.~E.}\ \bibnamefont {Economou}},\ and\
  \bibinfo {author} {\bibfnamefont {N.~M.}\ \bibnamefont {Tubman}},\ }\bibfield
   {title} {\bibinfo {title} {Classical pre-optimization approach for
  adapt-vqe: Maximizing the potential of high-performance computing resources
  to improve quantum simulation of chemical applications},\ }\href@noop {}
  {\bibfield  {journal} {\bibinfo  {journal} {arXiv preprint arXiv:2411.07920}\
  } (\bibinfo {year} {2024})}\BibitemShut {NoStop}%
\bibitem [{\citenamefont {Reardon-Smith}\ \emph {et~al.}(2024)\citenamefont
  {Reardon-Smith}, \citenamefont {Oszmaniec},\ and\ \citenamefont
  {Korzekwa}}]{reardon2024improved}%
  \BibitemOpen
  \bibfield  {author} {\bibinfo {author} {\bibfnamefont {O.}~\bibnamefont
  {Reardon-Smith}}, \bibinfo {author} {\bibfnamefont {M.}~\bibnamefont
  {Oszmaniec}},\ and\ \bibinfo {author} {\bibfnamefont {K.}~\bibnamefont
  {Korzekwa}},\ }\bibfield  {title} {\bibinfo {title} {Improved simulation of
  quantum circuits dominated by free fermionic operations},\ }\href@noop {}
  {\bibfield  {journal} {\bibinfo  {journal} {Quantum}\ }\textbf {\bibinfo
  {volume} {8}},\ \bibinfo {pages} {1549} (\bibinfo {year} {2024})}\BibitemShut
  {NoStop}%
\bibitem [{\citenamefont {Mocherla}\ \emph {et~al.}(2023)\citenamefont
  {Mocherla}, \citenamefont {Lao},\ and\ \citenamefont
  {Browne}}]{mocherla2023extending}%
  \BibitemOpen
  \bibfield  {author} {\bibinfo {author} {\bibfnamefont {A.}~\bibnamefont
  {Mocherla}}, \bibinfo {author} {\bibfnamefont {L.}~\bibnamefont {Lao}},\ and\
  \bibinfo {author} {\bibfnamefont {D.~E.}\ \bibnamefont {Browne}},\ }\bibfield
   {title} {\bibinfo {title} {Extending matchgate simulation methods to
  universal quantum circuits},\ }\href@noop {} {\bibfield  {journal} {\bibinfo
  {journal} {arXiv preprint arXiv:2302.02654}\ } (\bibinfo {year}
  {2023})}\BibitemShut {NoStop}%
\bibitem [{\citenamefont {Somma}(2005)}]{somma2005quantum}%
  \BibitemOpen
  \bibfield  {author} {\bibinfo {author} {\bibfnamefont {R.~D.}\ \bibnamefont
  {Somma}},\ }\bibfield  {title} {\bibinfo {title} {Quantum computation,
  complexity, and many-body physics},\ }\href
  {https://arxiv.org/abs/quant-ph/0512209} {\bibfield  {journal} {\bibinfo
  {journal} {arXiv preprint quant-ph/0512209}\ } (\bibinfo {year}
  {2005})}\BibitemShut {NoStop}%
\bibitem [{\citenamefont {Somma}\ \emph {et~al.}(2006)\citenamefont {Somma},
  \citenamefont {Barnum}, \citenamefont {Ortiz},\ and\ \citenamefont
  {Knill}}]{somma2006efficient}%
  \BibitemOpen
  \bibfield  {author} {\bibinfo {author} {\bibfnamefont {R.}~\bibnamefont
  {Somma}}, \bibinfo {author} {\bibfnamefont {H.}~\bibnamefont {Barnum}},
  \bibinfo {author} {\bibfnamefont {G.}~\bibnamefont {Ortiz}},\ and\ \bibinfo
  {author} {\bibfnamefont {E.}~\bibnamefont {Knill}},\ }\bibfield  {title}
  {\bibinfo {title} {Efficient solvability of {H}amiltonians and limits on the
  power of some quantum computational models},\ }\href
  {https://doi.org/https://doi.org/10.1103/PhysRevLett.97.190501} {\bibfield
  {journal} {\bibinfo  {journal} {Physical Review Letters}\ }\textbf {\bibinfo
  {volume} {97}},\ \bibinfo {pages} {190501} (\bibinfo {year}
  {2006})}\BibitemShut {NoStop}%
\bibitem [{\citenamefont {Galitski}(2011)}]{Galitski2011Quantum}%
  \BibitemOpen
  \bibfield  {author} {\bibinfo {author} {\bibfnamefont {V.}~\bibnamefont
  {Galitski}},\ }\bibfield  {title} {\bibinfo {title} {Quantum-to-classical
  correspondence and hubbard-stratonovich dynamical systems: A lie-algebraic
  approach},\ }\href {https://doi.org/10.1103/PhysRevA.84.012118} {\bibfield
  {journal} {\bibinfo  {journal} {Phys. Rev. A}\ }\textbf {\bibinfo {volume}
  {84}},\ \bibinfo {pages} {012118} (\bibinfo {year} {2011})}\BibitemShut
  {NoStop}%
\bibitem [{\citenamefont {Anschuetz}\ \emph {et~al.}(2023)\citenamefont
  {Anschuetz}, \citenamefont {Bauer}, \citenamefont {Kiani},\ and\
  \citenamefont {Lloyd}}]{anschuetz2022efficient}%
  \BibitemOpen
  \bibfield  {author} {\bibinfo {author} {\bibfnamefont {E.~R.}\ \bibnamefont
  {Anschuetz}}, \bibinfo {author} {\bibfnamefont {A.}~\bibnamefont {Bauer}},
  \bibinfo {author} {\bibfnamefont {B.~T.}\ \bibnamefont {Kiani}},\ and\
  \bibinfo {author} {\bibfnamefont {S.}~\bibnamefont {Lloyd}},\ }\bibfield
  {title} {\bibinfo {title} {Efficient classical algorithms for simulating
  symmetric quantum systems},\ }\href
  {https://doi.org/10.22331/q-2023-11-28-1189} {\bibfield  {journal} {\bibinfo
  {journal} {Quantum}\ }\textbf {\bibinfo {volume} {7}},\ \bibinfo {pages}
  {1189} (\bibinfo {year} {2023})}\BibitemShut {NoStop}%
\bibitem [{\citenamefont {Goh}\ \emph {et~al.}(2025)\citenamefont {Goh},
  \citenamefont {Larocca}, \citenamefont {Cincio}, \citenamefont {Cerezo},\
  and\ \citenamefont {Sauvage}}]{goh2023lie}%
  \BibitemOpen
  \bibfield  {author} {\bibinfo {author} {\bibfnamefont {M.~L.}\ \bibnamefont
  {Goh}}, \bibinfo {author} {\bibfnamefont {M.}~\bibnamefont {Larocca}},
  \bibinfo {author} {\bibfnamefont {L.}~\bibnamefont {Cincio}}, \bibinfo
  {author} {\bibfnamefont {M.}~\bibnamefont {Cerezo}},\ and\ \bibinfo {author}
  {\bibfnamefont {F.}~\bibnamefont {Sauvage}},\ }\bibfield  {title} {\bibinfo
  {title} {Lie-algebraic classical simulations for quantum computing},\ }\href
  {https://doi.org/10.1103/3y65-f5w6} {\bibfield  {journal} {\bibinfo
  {journal} {Physical Review Research}\ }\textbf {\bibinfo {volume} {7}},\
  \bibinfo {pages} {033266} (\bibinfo {year} {2025})}\BibitemShut {NoStop}%
\bibitem [{\citenamefont {Zimbor{\'a}s}\ \emph {et~al.}(2014)\citenamefont
  {Zimbor{\'a}s}, \citenamefont {Zeier}, \citenamefont {Keyl},\ and\
  \citenamefont {Schulte-Herbr{\"u}ggen}}]{zimboras2014dynamic}%
  \BibitemOpen
  \bibfield  {author} {\bibinfo {author} {\bibfnamefont {Z.}~\bibnamefont
  {Zimbor{\'a}s}}, \bibinfo {author} {\bibfnamefont {R.}~\bibnamefont {Zeier}},
  \bibinfo {author} {\bibfnamefont {M.}~\bibnamefont {Keyl}},\ and\ \bibinfo
  {author} {\bibfnamefont {T.}~\bibnamefont {Schulte-Herbr{\"u}ggen}},\
  }\bibfield  {title} {\bibinfo {title} {A dynamic systems approach to fermions
  and their relation to spins},\ }\href@noop {} {\bibfield  {journal} {\bibinfo
   {journal} {EPJ Quantum Technology}\ }\textbf {\bibinfo {volume} {1}},\
  \bibinfo {pages} {11} (\bibinfo {year} {2014})}\BibitemShut {NoStop}%
\bibitem [{\citenamefont {Bettaque}\ and\ \citenamefont
  {Swingle}(2024)}]{bettaque2024structure}%
  \BibitemOpen
  \bibfield  {author} {\bibinfo {author} {\bibfnamefont {V.}~\bibnamefont
  {Bettaque}}\ and\ \bibinfo {author} {\bibfnamefont {B.}~\bibnamefont
  {Swingle}},\ }\bibfield  {title} {\bibinfo {title} {The structure of the
  majorana clifford group},\ }\href@noop {} {\bibfield  {journal} {\bibinfo
  {journal} {arXiv preprint arXiv:2407.11319}\ } (\bibinfo {year}
  {2024})}\BibitemShut {NoStop}%
\end{thebibliography}%

\clearpage 

\appendix
\onecolumngrid

\mtcaddpart{}

\parttoc 

\section{Preliminaries}\label{app:prelims}

\subsection{Introduction to fermionic systems}
An $N$-mode Fermionic system in second quantization can be described in terms of $N$ creation operators $\{ a_i^\dagger \}_{i = 1}^{N}$ and annihilation operators $\{ a_i \}_{i = 1}^{N}$ that satisfy the canonical anticommutation relations:
\begin{equation}
    \label{eq:Fermionic_anticommutation}
    \{ a_i, a_j \} = \{ a_i^\dagger, a_j^\dagger \} = 0, \quad \{ a_i^\dagger, a_j \} = 1 \delta_{ij}.
\end{equation}
Mathematically, the $N$-mode Fermionic system is equivalent to the $N$-dimensional Fock space $\mathcal{F}(\mathbb{C}^N)$, a $2^N$-dimensional Hilbert space spanned by the so-called Fock basis.
The operators defined above allow us to define the basis as follows. First, the Fermionic vacuum $\ket{\text{vac}_\text{f}}$ is defined to be the unique vector such that $a_j \ket{\text{vac}_\text{f}} = 0$ for all $j = 1, \dots, N$. 
The remaining basis elements, namely \textit{Fock} states, can be constructed by considering all possible combinations of occupation numbers $n_j \in \{0, 1\}$:
\begin{align}
    \ket{n_1 n_2 \dots n_N} \coloneqq \prod_{j=1}^{N} (a_j^\dagger)^{n_j} \ket{\text{vac}_\text{f}}.
\end{align}

Creation and annihilation operators are not the only operators that can define the Fermionic space. It is also common to define an equivalent set of so-called Majorana operators $\{ m_k \}_{k = 1}^{2N}$ as
\begin{gather}
    \label{eq:Fermion_to_Majorana}
    m_{2 j - 1} \coloneqq a_{j}^\dagger + a_{j},\\ m_{2 j} \coloneqq  \mathrm{i} (a_{j}^\dagger - a_{j}).
\end{gather}
Such operators obey many useful properties, such as being unitary and self-adjoint, obeying relations
\begin{equation}
    m_i^\dagger=m_i,\quad\{m_i,m_j\}=2\delta_{ij}.
\end{equation}
These operators play an analogous role in the Fermionic Clifford group as Pauli operators do in the qubit Clifford group.
For a comprehensive introduction to the structure of the Majorana Clifford group see Ref.~\cite{bettaque2024structure}.

The above ways of defining Fermionic systems allow us to provide two equivalent forms of an $N$-mode second-quantized Fermionic Hamiltonian:
\begin{equation}
\begin{split}
    \mathcal{H}_f &= \sum_{ij}^{N} h_{ij} a_i^{\dagger} {a_j} + \sum_{ijkl}^{N} h_{ijkl} a_i^{\dagger} a_j^{\dagger} {a_k} {a_l} \\
    &=\sum_{ij}^{2N} \mathrm i c_{ij} {m_i} {m_j} + \sum_{ijkl}^{2N} c_{ijkl} \, {m_i} {m_j} {m_k} {m_l}.
\end{split}
\label{eq:hamiltonian}
\end{equation}
for coefficients $h_{ij}$, $h_{ijkl}$,  $c_{ij}$ and $c_{ijkl}$.
The equivalence between these forms comes directly from the linear dependency presented in Eq.~\eqref{eq:Fermion_to_Majorana}. 

Any Fermionic operator can be uniquely expressed as a linear combination of Majorana monomials. Hence, up to a sign, each unique product of Majorana operators can be associated with a $2N$-dimensional binary vector $\vec{b} = (b_1, \ldots, b_{2N}), b_i \in \{0, 1\}$ through the expression
\begin{equation}
     M_{\vec{b}} = \mathrm i^{r_{\vec{b}}}  m_1^{b_1} m_2^{b_2} \cdots m_{2N}^{b_{2N}} \, .
\end{equation}
We define the length of a Majorana monomial as the 1-norm of vector $\vec{b}$, $\norm{\vec{b}}_1 = \sum_{i = 1}^{2N} b_i$. Let us relabel a Majorana monomial of length $w$ as $m_{x_1}...\,m_{x_w}$. All monomials referenced throughout this document have been made Hermitian by multiplying the anti-Hermitian ones by the imaginary unit $\mathrm i$. Concretely, 
\begin{equation}
    (m_{x_1}...\,m_{x_w})^{\dagger} = m_{x_w}...\,m_{x_1} = (-1)^{\sum_{j=1}^{w-1} w - j} m_{x_1}...\,m_{x_w} = (-1)^{w(w-1)/2} m_{x_1}...\,m_{x_w}
\end{equation}
Thus if $w$ or $w-1$ is a multiple of 4 then $ m_{x_1}...\,m_{x_w}$ is already Hermitian and $r = 0$, but otherwise $ m_{x_1}...\,m_{x_w}$ is antihermitian and $r =1$ so that the factor of $\mathrm i$ makes the monomial Hermitian. 

Notice that, owing to the anti-commutation relations between Majorana operators, any repeated terms in a monomial cancel out, so each operator can appear at most once in a monomial. Also, any monomial is equal to itself upon reshuffling of its elements up to a sign, so we will assume $x_1 < x_2 < \ldots < x_{w}$ throughout without loss of generality.

\subsection{Algorithm overview}

We are interested in computing 
\begin{equation}\label{eq:expectation}
   f(\vec{\theta}) := \langle H \rangle_{\vec{\theta}} =  {\rm Tr} \left[U(\vec{\theta}) \varrho U(\vec{\theta})^\dagger H \right] 
\end{equation}
where $\varrho = \ketbra{n_1 n_2 \dots n_{N}}{n_1 n_2 \dots n_{N}}$, with $n_i \in \{0, 1\}$, is a Fock basis state. The circuit
\begin{equation}\label{eq:circuit}
    U(\vec{\theta}) = \prod_{j=1}^L e^{- \mathrm i \theta_j M_{{\vec{b}_j}}/2}
\end{equation}
consists of a sequence of $L$ Fermionic gates where each gate is generated by a Hermitian Fermionic generator $M_{{\vec{b}_j}}$ which is a $w$-monomial. We will assume $w$ even and $w \leq k$ for all gates and for some pre-defined value of $k$ (typically $k=4$). We will assume that the observable $H$ is composed of a sum of polynomially many $w$-monomials. That is,  $H = \sum_{\vec{b}} \alpha_{\vec{b}} M_{\vec{b}} $  where the $M_{\vec{b}}$ are $w$-monomials with $w \leq k'$. Typically, $k' = 4$, for interacting Hamiltonians.

Our Majorana Propagation method, analogously to Pauli propagation algorithms, works in the Heisenberg picture where the initial operator (here the observable $H$) is often sparse in the Majorana basis. In this basis, the operator is back-propagated through the circuit, i.e., we compute $U(\vec{\theta})^\dagger  H U(\vec{\theta})$, and finally overlap with the initial state $\varrho$. 
It is well known that the effect of applying a unitary gate generated by a single Pauli string $P$ to another Pauli string $Q$ is given by, 
\begin{equation}\label{eq:PauliRot}
    e^{ \mathrm i \theta P / 2} Q e^{- \mathrm  i \theta P / 2} = \begin{cases} Q, & [P, Q] = 0, \\ \cos (\theta) Q +  \mathrm i \sin(\theta) PQ, & \{P, Q \} = 0.\end{cases}
\end{equation}
We recall that the standard textbook derivation of Eq.~\eqref{eq:PauliRot} can be computed by Taylor-expanding $ e^{  \mathrm i \theta P / 2}$ and using $P^2 = I$. For Majorana monomials we similarly have $M^2 = I$, and thus the analogous expression holds here. Namely, we have 
\begin{equation}\label{eq:gate}
    e^{ \mathrm i \theta M' / 2} M e^{-  \mathrm i \theta M' / 2} = \begin{cases} M, & [M', M] = 0, \\ \cos (\theta) M +  \mathrm i \sin(\theta) M' M, & \{M', M \} = 0.\end{cases}
\end{equation}

The Heisenberg evolved operator, $U(\vec{\theta})^\dagger  H U(\vec{\theta})$, can be computed by applying each gate to each of the monomials in the observable in turn using Eq.~\eqref{eq:gate}. The final expectation function takes the form
\begin{equation}\label{eq:path_expectation}
    f\left(\vec{\theta}\right) = \sum_{j} c_{\vec{b}}\left(\vec{\theta} \right) \Tr[ \varrho M_{\vec{b}}] \,,
\end{equation}
where $c_{\vec{b}}\left(\vec{\theta} \right) $ are the coefficients of the back-propagated Majorana operators $M_{\vec{b}}$. These coefficients capture both the initial length of each of the relevant monomial in the target observable $H$ and the sine and cosine coefficients that have been picked up during the propagation.

\section{Coefficient and small angle truncations}\label{app:smallangle}

Here we suppose that we are interested in simulating a small angle subregion of the surrogate landscape.
Let us then define
\begin{equation}\label{eq:hypercube}
	\vol(r) := \{ \vec{\theta};\ \forall  i \, \, \theta_i \in [ -r ,  r ]   \},
\end{equation}
as the hypercube of parameter space centred around the point $\vec{0}$, and 
\begin{equation}\label{eq:uniformdist}
    \uni(r) := \text{Unif}[\vol(r)]
\end{equation}
as a uniform distribution over the hypercube $\vol(r)$. It can be shown that for any randomly chosen $\vec{\theta} \sim  \uni(r)$ the error induced by the small angle truncation strategy can be kept small. Combining this with the observation that small angle truncation keeps the number of paths tractable, one obtains the following guarantee on the efficiency of small angle truncation Majorana Propagation. 

\setcounter{theorem}{0}
\begin{proposition}[Time complexity of small-angle Majorana Propagation]\label{thm:uncorrfull}
Consider an expectation function $f(\vec{\theta}) = \Tr[ U(\vec{\theta}) \varrho U(\vec{\theta})^\dagger H]$ with $\varrho$ the initial state, an observable $H= \sum_{\vec{b}} a_{\vec{b}} M_{\vec{b}}$, where at most  $N \in \mathcal{O}(\poly(n))$ coefficients $a_j$ are non zero, with $\norm{\vec{a}}_1 \in \OC(1)$, and a parametrized circuit of the form of Eq.~\eqref{eq:circuit}. If the expectation values $\Tr[\varrho M_{\vec{b}}]$ of the initial state $\varrho$ with any Majorana operator $M_{\vec{b}}$ can be efficiently computed, then a runtime of 
\begin{equation}
    t \in \OC\left(N \cdot \nparams^{\log{\left(\frac{1}{\epsilon^2\delta}\right)}}\right)\;,
\end{equation}
suffices to classically simulate $f(\vec{\theta})$ up to an error $\epsilon$ with probability at least $1-\delta$ for random $\vec{\theta}\sim\uni(r)$ with 
\begin{equation}
    r \in \OC\left(  \frac{1}{\sqrt{\nparams}} \right) \, .
\end{equation}
Similarly, with a runtime of
\begin{equation}
    t \in \OC\left(N \cdot \nparams^{\log{\left(\frac{1}{\epsilon}\right)}}\right)\;
\end{equation}
we can classically simulate $f(\vec{\theta})$ up to an error $\epsilon$
for all $\vec{\theta}$ in the hypercube $\vol( r)$ with 
\begin{equation}
    r \in \OC\left(  \frac{1}{\nparams} \right) \, . 
\end{equation}
\end{proposition}
\begin{proof}
The proof follows from substituting Eq.~\eqref{eq:gate} and $|\Tr[\varrho M_{\vec{b}}] | \leq 1$ into the proof of Theorem 3 in Ref.~\cite{lerch2024efficient}.
\end{proof}

\section{Length Truncation: Intuition}

In the previous section we detailed coefficient and small-angle truncation strategies which cut terms from the sum in Eq.~\eqref{eq:path_expectation} where the prefactor $c_{\vec{b}}(\theta)$ that are either small or, in the case of the latter, likely to be small. Here we detail an alternative truncation strategy that aims to cut terms from the sum in Eq.~\eqref{eq:path_expectation} where $\text{Tr}[\varrho M_{\vec{b}}]$ is likely small. 

The core intuition here is that in order for a term $\text{Tr}[\varrho M_{\vec{b}}]$, where $\varrho$ is a Fock state, to be non-zero all the Majorana operators in the monomial $M_{\vec{b}}$ are paired. That is, $\text{Tr}[\varrho M_
{\vec{b}}]= 0$ if and only if for every even Majorana operator $m_{2i}$ in the monomial $M_{\vec{b}}$, the Majorana operator $m_{2i + 1}$ is also in the monomial $M_{\vec{b}}$ (this can be seen immediately from the definition of the Majorana operator). Next we observe that in general the higher the length $w$ of the Majorana monomial (for $w < N$) the lower the probability that the monomial is paired. 

This motivates a \textit{length truncation} strategy. At each step, after each gate is applied, all terms corresponding to Majorana monomials with length above a certain threshold are to be discarded. In typical circuits, the method is expected to result in an accurate computation since high-length monomials truncated along the way would most likely keep increasing in length along the way, therefore sharply decreasing their relative importance for the quantity under evaluation. This idea is analyzed with a greater level of detail and precision in the next two subsections.

\subsection{The contribution of typical high length Majorana monomials are exponentially suppressed}\label{app:exp_suppression}

In this subsection we substantiate the claim that the higher length Majorana monomials $M_{\vec{b}}$, that arise during a propagation, are more likely to annihilate when we compute their overlap with the initial Fock state. 
Intuitively, this can be seen by supposing that, due to the circuit’s complexity, all high-length $w$-monomials carry a similar length. There are $\binom{2N}{w}$ such monomials. The number of these that are paired is $\binom{N}{w/2}$. The probability $\probContribute{w} $ that any randomly chosen length $w$ Monomial $M_{\vec{b}}$ is paired (and therefore $\Tr[\varrho M_{\vec{b}} ] \neq 0$) is given by 
\begin{equation}\label{eq:fullypaired}
    \probContribute{w} = \frac{\binom{N}{w/2}}{\binom{2N}{w}} \, .
\end{equation}
Figure~\ref{fig:analytics} depicts the pairing probability as a function of $w$ for $N=100$. It is visually clear that only $w \approx 0$ and $w \approx 2N$ contribute significantly, all other terms contributing many orders of magnitude less. This observation is formalized by the following theorem where we show that the contribution of any random chosen Monomial with both $w \in \Theta(N)$ and $(2N - w) \in \Theta(N)$ is exponentially in $N$ suppressed. Conversely, the contribution of terms with $w \in \mathcal{O}(1)$ are only polynomially suppressed. 

\setcounter{theorem}{0}
\begin{theorem}[The contribution of high length Majoranas are exponentially suppressed]\label{thm:suppression-app}
The probability $\probContribute{w}$ that any randomly chosen $N$-mode Majorana monomonial with length $w$, where both $w \in \Theta(N)$ and $(2N - w) \in \Theta(N)$, has a non-zero expectation in any given Fock state is exponentially suppressed in $N$, that is
\begin{equation}
    \probContribute{w} \in \mathcal{O}(r^{-N})
\end{equation}
where $r > 1$. Furthermore, the minimal probability is obtained for $w = N$ with 
\begin{equation}
        \min_w \probContribute{w} = \probContribute{N} \in \mathcal{O}(2^{-N}) \, . 
\end{equation}
Conversely, if $w \in \OC(1)$ we have that
\begin{equation}
    \probContribute{w} \in \Omega\left(\frac{1}{\text{poly}(N)}\right) \, . 
\end{equation}

\end{theorem}

\medskip 

\begin{proof}
    As argued above, only paired Majorana modes have a non-vanishing expectation in any Fock state and this occurs with the probability $\probContribute{w}$ stated in Eq.~\eqref{eq:fullypaired}. We will start by showing that for $w = c N$ 
    \begin{equation}
    \probContribute{w} \in \mathcal{O}\left(r^{-N}\right)
    \end{equation}
    for some $r > 1$. 
    To do so we employ the following inequality 
    \begin{align}\label{eq:binom_ineq}
    \left(\frac{m}{k}\right)^k \leq \binom{m}{k} \leq  \left(\frac{em}{k}\right)^k \,,
    \end{align}
    to obtain 
    \begin{equation}\label{eq:upperboundonPw}
      \probContribute{w} \leq \left(\frac{e w}{2 N}\right)^{w/2}   \, . 
    \end{equation}
    Thus if we set $w/2 = c N$ ($0 < c < 1$) we have that 
    \begin{equation}
        \probContribute{ w } \leq p^{c N} 
    \end{equation}
    where $p = e c$. Hence if $c < 1/e$, corresponding to $w < 2 N / e  < 0.74 N$, we have that $0 < p < 1$ and thus $\probContribute{w} = r^{-N}$ with $r = p^{-c} > 1$.
    To quantify the scaling of $\probContribute{ w }$ in the region $0.74 N \leq w \leq N $ we note that $\probContribute{w}$ decreases monotonically with $w$ to its minimal value at $\probContribute{N}$ and hence 
    \begin{equation}
        \probContribute{ w } \leq r^{-N} 
    \end{equation}
    with $r = p^{-c}$ for some $0 < p < 1$ for any $w = N c$, $0 < c < 1$.
    Finally, since $\probContribute{w}$ is symmetric with respect to $w = N$, $\probContribute{w} = \probContribute{2N - w}$, the result applies to any $w = 2 N c$, as claimed. 
        
    To bound the minimal value of $\probContribute{w}$ at $w = N$ we can use the following Binomial coefficient bound 
    \begin{equation}
       \frac{4^N}{\sqrt{\pi ( N + 1/3)}} \leq  \binom{2N}{N} \leq \frac{4^N}{\sqrt{\pi ( N + 1/4)}}  \, .
    \end{equation}
    Thus we can upper bound the minimal value of $\probContribute{w}$ as 
    \begin{equation}
        \min_w \probContribute{w} = \probContribute{ N }  \leq 4^{- N/2} \sqrt{\frac{N + 1/3}{N/2 + 1/4}} \leq 4^{- N/2} \sqrt{\frac{3N}{N}} \in \mathcal{O}(2^{-N}) \, . 
    \end{equation}
    
    Finally, we can find a lower bound $\probContribute{w}$ for $w \in \OC(1)$ using Eq.~\eqref{eq:binom_ineq} to give
    \begin{equation}
        \probContribute{w} \geq \left(\frac{2N}{w}\right)^{w/2} \left(\frac{e 2 N}{w}\right)^{-w} =  \left(\frac{2 N}{w}\right)^{-w/2} e^{-w} =  \left(\frac{w}{2e^2N}\right)^{w/2} \geq \left(\frac{w}{15 N}\right)^{w/2} \, .
    \end{equation}
    Thus if $w \in \OC(1)$ we have that $\probContribute{w} \in \Omega\left(\frac{1}{\text{poly}(N)}\right)$.
\end{proof}

\subsection{The probability of backflow is suppressed}\label{app:backflow_suppresion}

In the previous subsection we saw that only very low or very high length monomials contribute substantially to Fock state expectation values. Two points must now be addressed:
\begin{itemize}
    \item[1.] The fact that it is safe to disregard high-length $w$-monomials in the final observable does not imply that it is safe to disregard high-length $w$-monomials at intermediate steps; high-length $w$-monomials at step $T$ may still evolve towards low-length terms in $O_D$, so their removal at step $T$ may result in large errors.
    \item[2.] High-length terms could evolve towards $w = 2N$, where their contribution matters, so truncating high-length elements early in the simulation may result in large errors.
\end{itemize}

As we shall now see, the answer to these concerns is that a typical dynamics leads to operators clustering around $w = N$, and the probability for their length to increase all the way to $w \approx2N$ or to decrease back to $w \approx 0$ is negligible.

\begin{theorem}[Suppression of backflow]\label{thm:backflow-app}
Consider the propagation of a length $w$ Majorana monomial $M$ under a length $k$ generator of rotation $M'$ as per Eq.~\eqref{eq:gate}. We assume that $[M', M] \neq 0$ such that the propagation causes branching and we denote the length of the new Majorana monomial, $M' M$, as $w'$. Let $p_+$ and $p_-$ denote the probabilities that the propagation increases ($w' > w$) and decreases ($w < w'$) operator length. The following statements hold:
\begin{enumerate}
    \item Any rotation under a $k=2$ length generator $M$ leaves the initial length $w$ unchanged,  $p_+ = p_- = 0$. 

    \item Consider a uniformly randomly chosen $N$-mode Majorana monomonial of length $w$ and a randomly chosen generator of length $k=4$. The ratio of the probabilities that propagation decreases versus increases operator length is: 
\begin{equation}
     R := \frac{p_-}{p_+}   = \frac{(w - 1)(w -2) }{(2N - 1 - w) (2N - 2 - w)} \, .
\end{equation}
Thus for any $w \ll N$ the backflow is suppressed with the suppression quadratic in $N$ for  $w \in \mathcal{O}(1)$. 

    \item For any $w \in \OC(1)$ and $k \in \OC(1)$ with $2 \leq k \leq w$ such that  $w k < N$ we have that 
    \begin{equation}
        R \in \OC(1/N^2) \, .
    \end{equation}
\end{enumerate}
\end{theorem}

\begin{proof}

\textit{Part 1)} Let us denote the number of overlapping Majorana operators between $M$ and $M'$ as $s$. The commutator $[ M, M' ] \neq 0$, i.e., the propagation only causes branching, if and only if the number of Majorana operators common to both monomials is odd. Thus in what follows we can assume $s$ is odd. 

Next, since even repetitions of Majorana operators in a monomial cancel out, the length of the monomial appearing in the branching term, $M' M$, must be smaller than the sum of each monomial’s length. More precisely, the length $w'$ of the term $M' M$ is $w' = k + w - 2s$. We can thus see that if $s < k / 2$, the branching process increases Majorana length, $w' > w$, while $w' < w$ if $s > k / 2$ and the length is left unchanged, $w' = w$ if $s = k/2$. 
We can therefore immediately see that if $k = 2$ then $s = k/2 = 1$ and propagation does not increase monomial length (proving the first claim in Theorem~\ref{thm:backflow-app}). \\ 

\textit{Part 2)} 
Next, as before let us consider the `random Majorana' limit where we assume that all Majorana monomials of any given length are equally likely. In this case, the probability that any two randomly chosen Majorana monomials of lengths $w$ and $k$ respectively overlap on precisely $s$ modes is given by 
\begin{equation}\label{eq:overlapprob}
   P(N, w, k, s) =  \frac{\binom{k}{s} \binom{2N-k}{w -s}}{\binom{2N}{w}} \, .
\end{equation}
To see this we can start by assuming that the first Majorana monomial of \(k\) elements is fixed. There is no loss of generality in doing so since we will count the selections for the second string relative to this fixed set. The second Majorana monomial is chosen by selecting \(w\) elements out of \(2N\). Hence, the total number of possible selections is $\binom{2N}{w}$. To have exactly \(s\) overlapping elements between the two strings, the following must occur:
\begin{enumerate}
    \item Choose \(s\) overlapping elements: From the fixed string of \(k\) elements, choose \(s\) elements. This can be done in $\binom{k}{s}$ ways.
    \item Choose the remaining \(w-s\) non-overlapping elements: The remaining \(w-s\) elements of the second string must come from the \(2N-k\) elements that are \emph{not} in the fixed string. This selection can be made in $\binom{2N-k}{w-s}$ ways.
\end{enumerate}
Thus, the total number of favorable outcomes is:
\begin{equation}
    \binom{k}{s} \binom{2N-k}{w-s} . 
\end{equation}
The probability that the two strings overlap in exactly \(s\) elements is the ratio of the number of favorable outcomes to the total number of outcomes:
\begin{equation}
    P(N, w, k, s) = \frac{\binom{k}{s} \binom{2N-k}{w-s}}{\binom{2N}{w}}.
\end{equation}
This proves Eq.~\eqref{eq:overlapprob}.

Next, we focus on the commonly encountered scenario where $k=4$. In this case, we can compute an exact compact expression for the ratio between the probability that any branching event decreases Majorana length as compared to increases it. Namely, as if $s < k / 2$, the branching process increases Majorana length and decreases it if $s > k / 2$, in this case this ratio is given by 
\begin{equation}
    R := \frac{p_-}{p_+} = \frac{P(N, w, 4, 3)}{P(N, w, 4, 1)} \, .
\end{equation}
If we substitute Eq.~\eqref{eq:overlapprob} into $R$ we find that 
\begin{equation}
    R = \frac{P(N, w, 4, 3)}{P(N, w, 4, 1)} = \frac{\binom{4}{3} \binom{2N-4}{w -3}}{\binom{4}{1} \binom{2N-4}{w -1}} = \frac{ \binom{2N-4}{w -3}}{\binom{2N-4}{w -1}} \, ,
\end{equation}
where we use that as binomial coefficients are symmetric we have $\binom{4}{3}= \binom{4}{1}$.
Using the explicit expression for the binomial coefficients this simplifies to 
    \begin{equation}
    R = \frac{(w - 1)! (2N - 3 - w)! }{(w - 3)! (2N - 1 - w)! }  = \frac{(w - 1)(w -2) }{(2N - 1 - w) (2N - 2 - w)} \, .
\end{equation}
Thus for $w \in \OC(1)$ we see that $  R $ is quadratically suppressed in $N$. \\

\textit{Part 3)} 
A similar suppression is obtained whenever $k$ and $w$ are much smaller than $2N$. In this case we have that the probability that the length decreases is the sum of $P(N, k, w, s)$ for all odd values of $s$ greater than $k/2$ (up to a maximum value of $k$), \begin{equation}\label{eq:P-}
    p_- = \sum_{k/2 < 2j+1 < k} P(N, w, k, 2j+1) \, .
\end{equation}
Similarly, the probability that the length increases is 
the sum of $P(N, k, w, s)$ for all odd values of $s$ less than $k/2$, i.e., 
\begin{equation}\label{eq:P+}
    p_- = \sum_{0 < 2j+1 < k/2} P(N, w, k, 2j+1) \, .
\end{equation}
There are the same number of terms in the sums in Eq.~\eqref{eq:P-} and Eq.~\eqref{eq:P+} and so we can bound $R$
 as 
 \begin{equation}\label{eq:Rbound-app}
    R \leq \frac{P_{\rm max}}{P_{\rm min}} \, 
\end{equation}
where 
\begin{equation}\label{eq:pmax_pmin}
\begin{aligned}
        &P_{\rm max} := \max_{j \in \mathbb{Z}^+} \{ P(N, w, k, 2j+1) \}_{k/2 < 2j+1 < k} \\ 
        &P_{\rm min} := \min_{j \in \mathbb{Z}^+} \{ P(N, w, k, 2j+1) \}_{0 < 2j+1 < k/2} \, .
\end{aligned}
\end{equation}

In the limit that $2N$ is much greater than $w$ and $k$ (and $w \geq k$) the probability $P(N, k, w, s)$ decreases monotonically in $s$. This can be seen by noting that in this limit the ratio of $P(N, k, w, s+1)$ and $P(N, k, w, s)$ is less than 1 for all values of $0 < s < k$. Namely, using the identities
\begin{equation}
    \frac{\binom{k}{s+1}}{\binom{k}{s}} = \frac{k-s}{s+1} \quad \text{and} \quad \frac{\binom{2N-k}{w-s-1}}{\binom{2N-k}{w-s}} = \frac{w-s}{2N-k-w+s+1},
\end{equation}
we have:
\begin{equation}
\begin{aligned}\label{eq:ratio}
       \frac{P(N, k, w, s+1)}{P(N, k, w, s)} =  \frac{(k-s)(w-s)}{(s+1)(2N-k-w+s+1)} 
\end{aligned}
\end{equation}
which is clearly less than 1 for all $0 < s < k$ in the limit that $w$ and $k$ are much less than $N$. 

More precisely, Eq.~\eqref{eq:ratio} is less than 1 (and $P(N, k, w, s)$ decreases monotonically in $s$) in the limit that $w k < N$ (for $k\geq 2$ and $w \geq 2$). To see this we first note that Eq.~\eqref{eq:ratio} takes the maximal value when $s=0$. This can be seen by taking the log of Eq.~\eqref{eq:ratio} to give 
\begin{equation}
    F(s):= \ln(k-s)+\ln(w-s)-\ln(s+1)-\ln\Bigl(2N-k-w+s+1\Bigr).
\end{equation}
Differentiating with respect to $s$ then yields
\begin{equation}\label{eq:derivative}
    F'(s) = -\frac{1}{k-s} -\frac{1}{w-s} -\frac{1}{s+1} -\frac{1}{2N-k-w+s+1}.
\end{equation}
We can assume that $s < k$ (if $s = k$ then the ratio in Eq.~\eqref{eq:ratio} evaluates to 0 and so is less than 1 as required) and $s < w$ (as we already assumed $k \leq w$). We now further assume that $2N - k - w + s +1 > 0$ for all $s$, which holds if 
\begin{equation}
    w + k < 2N + 1.
\end{equation}
Thus each of the denominators in Eq.~\eqref{eq:derivative} are positive, making each term negative. Thus, \(F'(s) < 0\) and so the ratio in  Eq.~\eqref{eq:ratio} is strictly decreasing in \(s\). Because \(R(s)\) is strictly decreasing, its maximum is attained when \(s\) is smallest (assuming \(s\ge0\)), i.e., at \(s=0\). We then require that 
\begin{equation}
\begin{aligned}
&\frac{k\,w}{2N-k-w+1} < 1.
\end{aligned}
\end{equation}
This holds if 
\begin{equation}\label{eq:intermediarybound}
     k w+k+w < 2N+1 \, .
\end{equation}
A looser but more compact bound can be obtained by noting that for $k \geq 2$ and $w \geq 2$ we have $k + w \leq k w $ and so the bound in Eq.~\eqref{eq:intermediarybound} holds if 
\begin{equation}
     w k  <  N \, .
\end{equation}
We have thus established that $P(N, k, w, s)$ is monotonically decreasing if $w k < N$. 

Returning to our original goal of bounding the ratio of forwards and backflow via Eq.~\eqref{eq:pmax_pmin}, we can use the fact that $P(N, k, w, s)$ is monotonically decreasing to conclude that
\begin{align}
    &P_{\rm max} =  P(N, w, k, s_{\rm max}) \\
    &P_{\rm min} =  P(N, w, k, s_{\rm min})
\end{align}
where $s_{\rm max}$ ($s_{\rm min}$) is the odd number that is closest to $k/2$ while strictly greater than (less than) $k/2$. For example, if $k$ is a multiple of $4$, we have $s_{\rm max} = k/2 + 1$ and $s_{\rm min} = k/2 - 1$. In complete generality we have that 
\begin{equation}
\begin{cases} s_{\rm max} = k/2 +1,  s_{\rm min} = k/2 -1 & k = 4l \\ 
s_{\rm max} = k/2 +1/2,  s_{\rm min} = k/2 -3/2 & k = 4l+1 \\
s_{\rm max} = k/2 +2,  s_{\rm min} = k/2 -2 & k = 4l + 2 \\
s_{\rm max} = k/2 +3/2,  s_{\rm min} = k/2 -1/2 & k = 4l  + 3 \, ,
\end{cases}
\end{equation}
where $l$ is some positive integer.

We can use these expressions in Eq.~\eqref{eq:Rbound-app} to bound $R$. Namely, for the case where $k$ is a multiple of $4$ we have 
 \begin{equation}
    R \leq \frac{P(N, k, w, k/2 +1 )}{P(N, k, w, k/2-1)} = \frac{(w - k/2 + 1)(w -k/2) }{(2N - k/2 - w +1) (2N - k/2 - w)} \, .
\end{equation}
The ratio can be similarly bounded in each of the other three cases above. The exact expression varies from case to case but in all cases we find a suppression that is at least quadratic in $N$ for constant $k$ and $w$. 
\end{proof}

\section{Length truncation: error analysis}\label{app:endtoendanalysis}

In this section we will present the proofs of Theorem~\ref{thm:main-text-upper-bound-mse}, Theorem~\ref{thm:resources} and Theorem~\ref{thm:poly-time-algo-for-squared-error-maintext} in the main text. That is, we provide a formal end-to-end error analysis of Majorana propagation applied to a random ADAPT-VQE inspired circuit model. 

\subsection{Preliminaries}

In this Appendix the calculations become more involved with the notation correspondingly becoming more cumbersome and so we have chosen to change notation conventions in places. For example, Majorana monomials that are used as generators in rotations we will now denote with $G$. In this subsection we will detail the definitions we will use in this appendix.

\begin{definition}
    We define $\Gamma_w$ to be the set of Majoranas of length $w$, and $\Gamma$ to be the set of all Majoranas.
\end{definition}

\begin{definition}[Random circuit model]\label{def:circuitmodel}
    We describe the random circuits we will analyse. We work over an arbitrary number $\nModes$ of fermions. We start with an observable $H =O_0$ sampled from a homogeneous distribution (defined below in definition~\ref{def:homogeneous-distribution}, when more notations are introduced). We then construct $\circuitDepth$ layers (indexing them in decreasing order in the Schrödinger picture, towards simplicity in the Heisenberg picture). Layer $j \in \{0, \ldots, \circuitDepth-1\}$ is composed of a uniformly random generator $\majoranaGate{j} \in \Gamma_k$, with $k = 4$, and some $\theta_i \in [0, 2\pi[$ uniformly at random, i.e:
    \begin{equation}\label{eq:singleupdate}
        O_{j+1} = \exp\left(i \frac{\theta}{2}\majoranaGate{j}\right) O_j \exp\left(-i \frac{\theta}{2}\majoranaGate{j}\right).
    \end{equation}

    Every random variable is sampled independently. We finally use as initial state an arbitrary Fock state $\ket{\fockState} := \ket{n_1 n_2\ldots n_N}$, measuring $O_{\circuitDepth}$. This whole procedure samples a random observable $O_0$ and circuit $U$ and evaluates: 
    \[f_{\circuitDepth}(U, O_0) = \Tr(U \ket{\fockState}\bra{\fockState} U^\dagger O_0) = \Tr(\ket{\fockState} \bra{\fockState} O_{\circuitDepth}).\]

    Note that we now write $f_L(\vec{\theta}) = f_{\circuitDepth}(U, O_0)$, making the parameters explicit.
\end{definition}

Our goal is to analyze the quality of the Majorana propagation algorithm on this circuit, when we only keep Majoranas of length $w \leq w_0$, for some even $w_0$.

\begin{definition}[Majorana-length partial 2-norm]\label{def:majorana-length-partial-norm}
    Let $P$ be an operator. We can always express it in the Majorana basis, $P = \sum_{\majorana \in \Gamma} c_\majorana \majorana$. We define its Majorana length-$w$ partial 2-norm to be:
    \[\partialNorm{P}{w}^2 = \sum_{\majorana \in \Gamma_w} |c_\majorana|^2.\]

    Note that this is such that $\sum_{w = 0}^{2\nModes} \partialNorm{P}{w}^2 = \|P\|^2$, where $\|P\|^2 = \Tr(P^{\dagger} P) = \sum_\majorana |c_\majorana|^2$ is the Schatten 2-norm.
\end{definition}

\begin{definition}
    We define the coefficients $c_{i, \majorana}$ to be the decomposition of our propagated observables in the Majorana basis:
    \[O_i = \sum_{\majorana \in \Gamma} c_{i,\majorana}\majorana.\]

    Since $O_i$ is Hermitian, we have $c_{i,\majorana} \in \mathbb{R}$. 
\end{definition}

\begin{definition}\label{def:homogeneous-distribution}
    We say that an observable $O_i$ has a homogeneous distribution if its coefficient in the Majorana basis are such that, for any $\majorana, \majoranaP \in \Gamma$:
    \begin{itemize}
        \item (Unbiased) $\exval(c_{i\majorana}) = 0$;
        \item (Orthogonal) If $\majorana \neq \majoranaP$, then $\Cov(c_{i\majorana}, c_{i\majoranaP}) = \exval(c_{i\majorana}c_{i\majoranaP}) = 0$;
        \item (Length-homogeneous) $\Var(c_{i\majorana}) = \exval(c_{i\majorana}^2)$ only depends on $|\majorana|$.
    \end{itemize}

    Note that we assumed above in definition~\ref{def:circuitmodel} that $O_0$ respects this definition.
\end{definition}

\begin{definition}\label{def:truncated-observable}
    Let $w_0 \in \{0, 1, \ldots, 2\nModes\}$. We define $O_i^{(w_0)}$ to be the observable after the $i$th gate in Majorana propagation, where we truncate all Majoranas of length strictly greater than $w_0$. We also define $c_{i, \majorana}^{(w_0)}$ to be its coefficients in the Majorana basis and $f_{\circuitDepth}^{(w_0)}(U, O_0)$ the function computed by this procedure:
    \[f_{\circuitDepth}^{(w_0)}(U, O_0) = \Tr(\ket{\fockState} \bra{\fockState} O_{\circuitDepth}^{(w_0)}).\]
    
    Note that $O_i^{(2\nModes)} = O_i$. Hence, all the following propositions also apply to untruncated Majorana propagation.
\end{definition}

\begin{remark}
    Let $2 \leq o \leq 2\nModes- 2$ be an arbitrary even integer. Consider a distribution of initial observables where we let $O_0$ to be an element picked uniformly  at random from $\Gamma_o$. This is not a homogeneous distribution since $\exval(c_{0\majorana}) \neq 0$ for any $\majorana \in \Gamma_o$.

    Now, note that we can consider a different distribution, where $O_0$ is sampled uniformly at random from $\Gamma_o$, but we then also give it a $\pm$ sign uniformly at random. This distribution is homogeneous and hence all our result applies.

    However, note the two distributions give exactly the same mean-squared error $\exval((f_{\circuitDepth}(U, O_0) - f_{\circuitDepth}^{(w_0)}(U, O_0)^2)$ since everything is linear and the $\pm$ sign gets killed by the square. Hence, all our results below also apply to this non-homogeneous distribution.

\end{remark}

\subsection{A classically efficient algorithm to compute the mean square simulation error}

In this subsection we derive a recursive expression for the averaged error that computes the averaged error at time $i+1$ based on the averaged error at time $i$. Applying this expression iteratively from $i = 0$ to $i= \circuitDepth$ thus allows us to efficiently compute the averaged error classically. \\

Central to all the analysis that follows will be the following lemma that quantifies how the coefficients are updated after each layer of the random circuit, i.e., after each random Majorana rotation,  as specified in Eq.~\eqref{eq:singleupdate}. Before stating this lemma, we require a crucial definition.

\begin{definition}[Phaseless product]
    We define the phaseless product operation $\phaselessprod: \Gamma \times \Gamma \mapsto \Gamma$ such that, for any $M_{b_1}, M_{b_2} \in \Gamma$, there exists some $c \in \{\pm 1, \pm i\}$ such that:
    \[M_{b_1} \phaselessprod M_{b_2} = c M_{b_1}\cdot M_{b_2} \in  \Gamma.\]

    This is just a multiplication, where we get rid of any coefficient. This is moreover well-defined, since we simply have $M_{b_1} \phaselessprod M_{b_2} = M_{b_1 + b_2}$.
\end{definition}

With this new definition, we are now ready to state our lemma.

\begin{lemma}[Propagation rule]\label{lma:c-propagation-rule}
    Let $\majorana \in \Gamma$, $w_0$ and $i$ be arbitrary. If $|\majorana| > w_0$, then $c_{i+1, \majorana}^{(w_0)} = 0$. Otherwise:
    \[\begin{split}
        c_{i+1, \majorana}^{(w_0)} = &\ I([\majorana, \majoranaGate{i}] = 0) c_{i\majorana}^{(w_0)} + I(\{\majorana, \majoranaGate{i}\} = 0) c_{i\majorana}^{(w_0)} \cos(\theta_i) \\& + \sum_{\majoranaOther \in \Gamma} I(\{\majoranaOther, \majoranaGate{i}\} = 0 \land \majoranaGate{i} \phaselessprod \majoranaOther = \majorana) \sin(\theta_i) c_{i\majoranaOther}^{(w_0)} \cdot i \cdot s_{\majoranaGate{i}\majoranaOther},
    \end{split}\]
    where $s_{\majoranaGateNoParam \majoranaOther}$ is such that $\majoranaGateNoParam \cdot \majoranaOther = s_{\majoranaGateNoParam \majoranaOther} \majoranaGateNoParam \phaselessprod \majoranaOther$. Note that $s_{\majoranaGateNoParam \majoranaOther} \in \{-1, 1, -i, i\}$ in the general case, but $s_{\majoranaGateNoParam \majoranaOther} \in \{-i, i\}$ when $\majoranaGateNoParam$ and $\majoranaOther$ anticommute (to preserve the Hermitianness of $\exp\left(i \frac{\theta}{2}\majoranaGateNoParam\right) \majoranaOther \exp\left(-i \frac{\theta}{2}\majoranaGateNoParam\right)$).
\end{lemma}
\begin{proof}
    We derive this formula by looking at the Majorana propagation algorithm, with truncation $w_0$. If $|\majorana| > w_0$ then, by construction, we must have $c_{i+1, \majorana}^{(w_0)} = 0$. Otherwise, we can get contribution from different types of paths.

    If $\majorana$ and $\majoranaGate{i}$ commute, then we get full contribution of $c_{i\majorana}^{(w_0)}$. If they anticommute, then this contribution is scaled by $\cos(\theta_i)$. We moreover get contribution from all the branching paths, i.e.~all other $\majoranaOther \in \Gamma$ that anticommute with $\majoranaGate{i}$ and that get mapped to $\majorana$ after application of the gate. This is moreover scaled by a $is_{\majoranaGate{i}\majoranaOther}\sin(\theta_i)$ factor.
\end{proof}

\subsubsection{Orthogonality and mean squared values of path coefficients}

Analogously to the case in other path propagation methods~\cite{angrisani2024classically} we will use the orthogonality of the paths (i.e, the coefficients corresponding to different monomials) to substantially simplify our error analysis. However, we are working with a different random circuit model we need to re-derive these orthogonality properties explicitly here. These properties are captured by following Lemma~\ref{lma:c-moments}. This lemma in fact goes beyond proving orthogonality and also establishes several key properties of the coefficients that will become useful later. 

\begin{lemma}\label{lma:c-moments}
    Let $\majorana, \majoranaP \in \Gamma$, $w = |\majorana|$ be the length of $\majorana$, $w_0$ be the truncation value, and keep the time step $i$ arbitrary. Then, $O_i^{(w_0)}$ has a homogeneous distribution (definition~\ref{def:homogeneous-distribution}). More precisely, the coefficients have the following moment properties:
    \begin{enumerate}
        \item The coefficients are unbiased: $\exval(c_{i\majorana}^{(w_0)}) = 0$ \, .
        \item The coefficients are orthogonal: $\exval(c_{i\majorana}^{(w_0)}c_{i\majoranaP}^{(w_0)}) = 0$ for $\majorana \neq \majoranaP$. 
        \item The mean squared coefficient of a monomial $\majorana$ depends on its length $w = |\majorana|$ but not the particular $\majorana$: 
        \begin{equation}
            \exval(c_{i\majorana}^{(w_0)}c_{i\majorana}^{(w_0)})  = \frac{\exval(\partialNorm{O_i^{(w_0)}}{w}^2)}{\binom{2\nModes}{w}} \, ,
        \end{equation}
        where $\partialNorm{P}{w}$ is the partial 2-norm as defined in definition \ref{def:majorana-length-partial-norm}.
        \item The correlation between truncated and untruncated coefficients is determined by the correlations between truncated coefficient: $\exval(c_{i\majorana} c_{i\majoranaP}^{(w_0)}) = \exval(c_{i\majorana}^{(w_0)}c_{i\majoranaP}^{(w_0)})$. 
    \end{enumerate}
\end{lemma}
\begin{proof}
    We prove each case separately, by induction on $i$.
    \begin{enumerate}
        \item We prove $\exval(c_{i\majorana}^{(w_0)}) = 0$ first. The initial case is trivial if $|\majorana| > w_0$. If $|\majorana| \leq w_0$, then $c_{0\majorana}^{(w_0)} = c_{0\majorana}$ and hence the initial case holds by the homogeneous distribution of $O_0$ (definition~\ref{def:homogeneous-distribution}). The inductive case then trivially holds thanks to lemma~\ref{lma:c-propagation-rule}.
        \item \label{enum:proof-c-uncorrelated}  Suppose $\majorana \neq \majoranaP$. We aim to show that $\exval(c_{i\majorana}^{(w_0)}c_{i\majoranaP}^{(w_0)}) = 0$. If $|\majorana| > w_0$ or $|\majoranaP| > w_0$, then this is trivial. We prove the case $|\majorana| \leq w_0$ and $|\majoranaP| \leq w_0$ by induction on $i$. The initial case is again easy by definition of $O_0$. For the inductive case, directly simplifying most terms thanks to the fact $\exval_{\theta}(\cos(\theta)) = \exval_{\theta}(\sin(\theta)) = \exval_{\theta}(\sin(\theta)\cos(\theta)) = 0$, and writing $\exval_j = \exval_{\theta_j, \majoranaGate{j}}$ and $I(\majoranaOther, \majoranaGate{i} \mapsto \majorana) = I(\{\majoranaOther, \majoranaGate{i}\} = 0 \land \majoranaGate{i} \phaselessprod \majoranaOther = \majorana)$:
        \begin{equation}\label{eq:expected-c-c}
           \begin{split}
            & \exval_{1, \ldots, i-1,\majoranaGate{i},\theta_i}(c_{i+1, \majorana}^{(w_0)} c_{i+1, \majoranaP}^{(w_0)})\\
            =\ & \exval_{\majoranaGate{i}}(I([\majorana, \majoranaGate{i}] = 0)I([\majoranaP, \majoranaGate{i}] = 0)) \exval_{1, \ldots, i-1}(c_{i, \majorana}^{(w_0)} c_{i, \majoranaP}^{(w_0)})
            \\& + \exval_{\majoranaGate{i}}(I(\{\majorana, \majoranaGate{i}\} = 0)I(\{\majoranaP, \majoranaGate{i}\} = 0)) \exval_{1, \ldots, i-1}(c_{i, \majorana}^{(w_0)} c_{i, \majoranaP}^{(w_0)}) \exval_{\theta_i}(\cos^2(\theta_i)) 
            \\& + \exval_{\majoranaGate{i}}\left(\sum_{\majoranaOther, \majoranaOtherP \in \Gamma} I(\majoranaOther, \majoranaGate{i} \mapsto \majorana) I(\majoranaOtherP, \majoranaGate{i} \mapsto \majoranaP) \exval_{\theta_i}(\sin(\theta_i)^2) \exval_{1, \ldots, i-1}(c_{i\majoranaOther}^{(w_0)} c_{i\majoranaOtherP}^{(w_0)}) i^2 s_{\majoranaGate{i}\majoranaOtherP} s_{\majoranaGate{i}\majoranaOther}\right).
        \end{split} 
        \end{equation}

        Note that, in the last term, using the fact $\majorana \neq \majoranaP$:
        \[\begin{split}
            & I(\majoranaOther, \majoranaGate{i} \mapsto \majorana) I(\majoranaOtherP, \majoranaGate{i} \mapsto \majoranaP) = 1 \implies \majoranaOther \phaselessprod \majoranaGate{i} = \majorana \land \majoranaOtherP \phaselessprod \majoranaGate{i} = \majoranaP 
            \\ \iff & \majoranaOther \phaselessprod \majorana = \majoranaGate{i} = \majoranaOtherP \phaselessprod \majoranaP \implies \majoranaOther = \majoranaOtherP \phaselessprod \majorana \phaselessprod \majoranaP \neq \majoranaOtherP.
        \end{split}\]

        Hence, all non-zero terms are such that $\majoranaOther \neq \majoranaOtherP$. Now, we know by our inductive hypothesis that this necessarily means $\exval_{1, \ldots, i-1}(c_{i\majoranaOther}^{(w_0)} c_{i\majoranaOtherP}^{(w_0)}) = 0$. Similarly, since $\majorana \neq \majoranaP$, $\exval_{1, \ldots, i-1}(c_{i, \majorana}^{(w_0)} c_{i, \majoranaP}^{(w_0)}) = 0$. Overall, this gives us exactly that:
        \[\exval_{1, \ldots, i-1,\majoranaGate{i},\theta_i}(c_{i+1, \majorana}^{(w_0)} c_{i+1, \majoranaP}^{(w_0)}) = 0.\]
        
        \item \label{enum:proof-c-only-depends-on-w} We now want to show $\exval(c_{i\majorana}^{(w_0) 2}) = \exval(\partialNorm{O_i^{(w_0)}}{w}^2) / \binom{2\nModes}{w}$. To do that, we start by showing $\exval(c_{i\majorana}^{(w_0) 2})$ only depends on $|\majorana| = w$, but not on $\majorana$. We see the case $i= 0$ is easy, by definition of $O_0$. Moreover, the inductive case is very similar to the one we just did (point~\ref{enum:proof-c-uncorrelated}), using the facts $\exval_{\theta}(\cos(\theta)) = \exval_{\theta}(\sin(\theta)) = \exval_{\theta}(\sin(\theta)\cos(\theta)) = 0$, $I(E)^2 = I(E)$ for any condition $E$, and $\exval(c_{i\majoranaOther}^{(w_0)}c_{i\majoranaOtherP}^{(w_0)}) = 0$ for $\majoranaOther \neq \majoranaOtherP$ as we just proved, equation~\ref{eq:expected-c-c} becomes:
        \[\begin{split}
            \exval_{1, \ldots, i-1,\majoranaGate{i},\theta_i}(c_{i+1, \majorana}^{(w_0) 2} )
            =\ & \exval_{\majoranaGate{i}}(I([\majorana, \majoranaGate{i}] = 0)) \exval_{1, \ldots, i-1}(c_{i, \majorana}^{(w_0) 2})
            \\& + \exval_{\majoranaGate{i}}(I(\{\majorana, \majoranaGate{i}\} = 0)) \exval_{1, \ldots, i-1}(c_{i, \majorana}^{(w_0) 2}) \exval_{\theta_i}(\cos^2(\theta_i)) 
            \\& + \exval_{\majoranaGate{i}}\left(\sum_{\majoranaOther \in \Gamma} I(\majoranaOther, \majoranaGate{i} \mapsto \majorana) \exval_{\theta_i}(\sin(\theta_i)^2) \exval_{1, \ldots, i-1}(c_{i\majoranaOther}^{(w_0) 2}) i^2 s_{\majoranaGate{i}\majoranaOther}^2\right).
        \end{split}\]

        Now, $s_{\majoranaGate{i}, \majoranaOther} \in \{i, -i\}$ as mentioned in lemma~\ref{lma:c-propagation-rule}, so $s_{\majoranaGate{i}, \majoranaOther}^2 = -1$. This thus simplifies to:
        \begin{equation}\label{eq:expected-c-squared}
            \begin{split}
            \exval_{1, \ldots, i-1,\majoranaGate{i},\theta_i}(c_{i+1, \majorana}^{(w_0) 2} )
            =\ & \prob_{\majoranaGate{i}}([\majorana, \majoranaGate{i}] = 0) \exval_{1, \ldots, i-1}(c_{i, \majorana}^{(w_0) 2})
            \\& + \prob_{\majoranaGate{i}}(\{\majorana, \majoranaGate{i}\} = 0) \exval_{1, \ldots, i-1}(c_{i, \majorana}^{(w_0) 2}) \exval_{\theta_i}(\cos^2(\theta_i)) 
            \\& + \sum_{\majoranaOther \in \Gamma} \prob_{\majoranaGate{i}}(\majoranaOther, \majoranaGate{i} \mapsto \majorana) \exval_{\theta_i}(\sin(\theta_i)^2) \exval_{1, \ldots, i-1}(c_{i\majoranaOther}^{(w_0) 2}).
            \end{split}
        \end{equation}

        We know that, since $|\majoranaGate{i}| = 4$, then
        \[\prob_{\majoranaGate{i}}(\{\majorana, \majoranaGate{i}\} = 0)) = p_-(w) + p_+(w) = \frac{\binom{4}{1} \binom{2\nModes - 4}{w - 1} + \binom{4}{3} \binom{2\nModes - 4}{w - 3}}{\binom{2\nModes}{w}},\]
        which only depends on $w$ and not on $\majorana$. Similarly, by inductive hypothesis, $\exval_{1, \ldots, i-1}(c_{i\majoranaOther}^{(w_0) 2})$ only depends on $|\majoranaOther|$. The only question left is thus, for any given $u$, and $\majorana, \majoranaP$ such that $|\majorana| = |\majoranaP|$, whether the following holds:
        \[\sum_{\majoranaOther \in \Gamma: |\majoranaOther| = u} \prob_{\majoranaGateNoParam}(\majoranaOther, \majoranaGateNoParam \mapsto \majorana) = \sum_{\majoranaOther \in \Gamma: |\majoranaOther| = u} \prob_{\majoranaGateNoParam}(\majoranaOther, \majoranaGateNoParam \mapsto \majoranaP).\]

        To do that, we find a bijection between the two following sets:
        \[\begin{split}
            & \{(\majoranaOther, \majoranaGateNoParam) \in \Gamma_u \times \Gamma_4: \{\majoranaOther, \majoranaGateNoParam\} = 0 \land \majoranaGateNoParam \phaselessprod \majoranaOther = \majorana \} 
            \\ \longleftrightarrow\ &  \{(\majoranaOtherP, \majoranaGateNoParamP) \in \Gamma_u \times \Gamma_4: \{\majoranaOtherP, \majoranaGateNoParamP\} = 0 \land \majoranaGateNoParamP \phaselessprod \majoranaOtherP = \majoranaP\}.
        \end{split}\]

        Now, since $|\majorana| = |\majoranaP|$, we know that there is a permutation (of their bitstrings) such that $\phi(\majorana) = \majoranaP$. We thus consider the bijection $(\majoranaOtherP, \majoranaGateNoParamP) = (\phi(\majoranaOther), \phi(\majoranaGateNoParam))$. Note that this indeed a bijection, it does preserve the length of the monomials, it preserves anticommutation (since it only depends on the size of the intersection of the two monomials), and is such that $\majoranaGateNoParamP \phaselessprod \majoranaOtherP = \majoranaP$ since bitstrings are such that $\phi(x) + \phi(y) = \phi(x + y)$. This allows to conclude the inductive step: nothing depends specifically on $\majorana$, everything just depends on $|\majorana|$.
        
        Overall, we just proved that $\exval(c_{i\majorana}^{(w_0) 2})$ only depends on $|\majorana|$ by induction. Now, by definition of the partial 2-norm (definition~\ref{def:majorana-length-partial-norm}):
        \[\partialNorm{O_i^{(w_0)}}{w}^2 = \sum_{\majorana \in \Gamma_w} c_{i \majorana}^{(w_0)2}.\]

        This gives exactly our result, since there are $\binom{2\nModes}{w}$ terms in the sum, which are all equal on expectation.

        \item We finally prove $\exval(c_{i\majorana} c_{i\majoranaP}^{(w_0)}) = \exval(c_{i\majorana}^{(w_0)}c_{i\majoranaP}^{(w_0)})$. When $\majorana \neq \majoranaP$, both sides are zero, by an argument completely similar to point~\ref{enum:proof-c-uncorrelated}. Indeed, the case $|\majoranaP| > w_0$ is easy, so we can assume $|\majoranaP| \leq w_0$. The initial case is then easy: if $|\majorana| \leq w_0$ then $c_{0\majorana} = c_{0\majorana}^{(w_0)}$ and hence both sides are equal to $\exval(c_{i\majorana} c_{i\majoranaP}) = 0$ by the homogeneous distribution of $O_0$ (definition~\ref{def:homogeneous-distribution}); and if $|\majorana| > w_0$ then $c_{0\majorana}^{(w_0)} = 0$ and $\exval(c_{0\majorana} c_{0\majoranaP}^{(w_0)}) = \exval(c_{0\majorana} c_{0\majoranaP}) = 0$. The inductive case is the exact same as the one of point~\ref{enum:proof-c-uncorrelated}, since we have the same recursion rule.
        
        We thus suppose that $\majorana = \majoranaP$, and prove our result by induction on $i$. The initial case is again easy: if $|\majorana| > w_0$, then $c_{i\majorana}^{(w_0)} = 0$ with probability $1$, so both sides are equal to zero. Otherwise, $c_{i\majorana}^{(w_0)} = c_{i\majorana}$ with probability 1, so both sides are again equal.

        Let us now take a look at the inductive step. Again, if $|\majoranaP| > w_0$, then both sides are zero. We can thus suppose $|\majorana| = |\majoranaP| \leq w_0$. But then, in that case, just like before:
        \[\begin{split}
            \exval_{1, \ldots, i-1,\majoranaGate{i},\theta_i}(c_{i+1, \majorana} c_{i+1, \majorana}^{(w_0)} )
            =\ & \exval_{\majoranaGate{i}}(I([\majorana, \majoranaGate{i}] = 0)) \exval_{1, \ldots, i-1}(c_{i, \majorana} c_{i, \majorana}^{(w_0)})
            \\& + \exval_{\majoranaGate{i}}(I(\{\majorana, \majoranaGate{i}\} = 0)) \exval_{1, \ldots, i-1}(c_{i, \majorana}c_{i, \majorana}^{(w_0)}) \exval_{\theta_i}(\cos^2(\theta_i)) 
            \\& + \exval_{\majoranaGate{i}}\left(\sum_{\majoranaOther \in \Gamma} I(\majoranaOther, \majoranaGate{i} \mapsto \majorana) \exval_{\theta_i}(\sin(\theta_i)^2) \exval_{1, \ldots, i-1}(c_{i\majoranaOther} c_{i\majoranaOther}^{(w_0)}) i^2 s_{\majoranaGate{i}\majoranaOther}^2\right).
        \end{split}\]

        This gives directly our result by the inductive hypothesis.
        \end{enumerate}
\end{proof}

\subsubsection{Simplified expressions for the mean squared error}

We can now use the orthogonality properties found above to derive following simplified expressions for the moments of the function, the mean squared simulation error.

\begin{proposition}\label{Prop:f-moments}
    Writing $f_{\circuitDepth}(U, O_0) = f$ for the sake of simplicity and its approximation using Majorana propagation with length truncation $w_0$ as $f^{(w_0)}$. The following properties of their moments hold:
    \begin{enumerate}
        \item The truncated function is unbiased: $\exval(f^{(w_0)}) = 0$.
        \item  The correlations between $f^{(w_0)}$ and $f$ can be determined from $f^{(w_0)}$ alone: $\exval(f \cdot f^{(w_0)}) = \exval(f^{(w_0) 2})$.
        
        \item Cross terms in the mean squared error vanish: $\exval((f - f^{(w_0)})^2) = \exval(f^2) - \exval(f^{(w_0) 2})$.

        \item The mean squared error is equal to the difference in the variances: $\exval((f - f^{(w_0)})^2) = \Var(f) - \Var(f^{(w_0)})$.
    \end{enumerate}
\end{proposition}

The final two points are entirely analogous to the Pauli case with locally scrambling circuits \cite[equation 16]{angrisani2024classically}. By the exact same analysis as Ref.~\cite{angrisani2024classically}, this tells us that Majorana propagation with truncation is always better than guessing zero. The proof of Proposition~\ref{Prop:f-moments} stems straightforwardly from the orthogonality of the (averaged) coefficients.

\begin{proof}
     We consider each case separately.
    \begin{enumerate}
        \item The fact $\exval(f^{(w_0)}) = 0$ directly comes from the fact $\exval(c_{\circuitDepth\majorana}) = 0$ for all $\majorana$ by lemma~\ref{lma:c-moments}, since:
        \[\exval(f^{(w_0)}) = \sum_\majorana \exval(c_{\circuitDepth\majorana}^{(w_0)}) \bra{\fockState} \majorana \ket{\fockState}.\]

        \item The fact $\exval(f \cdot f^{(w_0)}) = \exval(f^{(w_0) 2})$ directly comes from the fact $\exval(c_{i\majorana} c_{i\majoranaP}^{(w_0)}) = \exval(c_{i\majorana}^{(w_0)}c_{i\majoranaP}^{(w_0)})$ by lemma~\ref{lma:c-moments}. 

        \item This equality is a direct consequence from the previous one:
        \[\exval((f - f^{(w_0)})^2) = \exval(f^2 - 2f f^{(w_0)} + f^{(w_0 )2}) = \exval(f^2 - 2f^{(w_0) 2} + f^{(w_0 )2}) = \exval(f^2 - f^{(w_0) 2}).\]

        \item  This is a direct consequence of 1 and 3, using the fact that $f^{(2N)} = f$ as mentioned in definition~\ref{def:truncated-observable}.
    \end{enumerate}
\end{proof}

We can further evaluate these expressions using the expression for the mean squared value of the coefficients found in point 2 in Lemma~\ref{lma:c-moments}. On doing so, we obtain the following proposition.  

\begin{proposition}\label{thm:poly-time-algo-for-squared-error-oldpart2}     Writing $f_{\circuitDepth}(U, O_0) = f$ for the sake of simplicity and its approximation using Majorana propagation with length truncation $w_0$ as $f^{(w_0)}$. Using $\partialNorm{P}{w}$ to be the partial 2-norm (definition \ref{def:majorana-length-partial-norm}), the following properties of their moments hold:
        \begin{enumerate}
        \item $\exval(f^{(w_0) 2}) = \sum_{w=0}^{2\nModes} \exval(\partialNorm{O_{\circuitDepth}^{(w_0)}}{w}^2) \probContribute{w}$.
        \item $\exval((f - f^{(w_0)})^2) = \sum_{w=0}^{2\nModes} (\exval(\partialNorm{O_{\circuitDepth}}{w}^2) - \exval(\partialNorm{O_{\circuitDepth}^{(w_0)}}{w}^2)) \probContribute{w}$.
    \end{enumerate}
    where $\probContribute{w}$ is the probability of a length $w$ monomial having a non-zero contribution to the expectation value as defined in Eq.~\eqref{eq:fullypaired}. 
\end{proposition}

\begin{proof}
    \begin{enumerate}
        \item We have, by lemma~\ref{lma:c-moments}:
        \[\begin{split}
            \exval(f^{(w_0) 2}) & = \sum_{\majorana, \majoranaP} \exval(c_{\circuitDepth\majorana}c_{\circuitDepth\majoranaP}) \bra{\fockState} \majorana\ket{\fockState}\bra{\fockState} \majoranaP\ket{\fockState} = \sum_\majorana \frac{\exval(\partialNorm{O_{\circuitDepth}^{(w_0)}}{|\majorana|}^2)}{\binom{2\nModes}{|\majorana|}} (\bra{\fockState} \majorana\ket{\fockState})^2 
            \\& = \sum_w \exval(\partialNorm{O_{\circuitDepth}^{(w_0)}}{w}^2) \frac{\sum_{\majorana \in \Gamma_w} (\bra{\fockState} \majorana\ket{\fockState})^2}{\binom{2\nModes}{w}}.
        \end{split}\]

        Now, $(\bra{\fockState} \majorana\ket{\fockState})^2$ is either 0 if $\majorana$ does not contribute, or $1$ if contributes. The sum over $\majorana \in \Gamma_w$ thus counts the number of terms that contribute, telling us that, by definition:
        \[\exval(f^{(w_0) 2}) = \sum_w \exval(\partialNorm{O_{\circuitDepth}^{(w_0)}}{w}^2) \probContribute{w}.\]

        \item This is a direct consequence of point 1 and lemma~\ref{Prop:f-moments}, using the fact that $f^{(2N)} = f$ as mentioned in definition~\ref{def:truncated-observable}.
    \end{enumerate}
\end{proof}

\subsubsection{A recursive expression for the averaged error}

We see in Proposition~\ref{thm:poly-time-algo-for-squared-error-oldpart2} point 2 that the mean squared error depends on the difference of partial 2-norms $\exval(\partialNorm{O_{\circuitDepth}}{w}^2) - \exval(\partialNorm{O_{\circuitDepth}^{(w_0)}}{w}^2)$, namely the difference between the averaged sum of the coefficients of the final propagated operator with and without truncation. Thus it remains to compute the contribution that stems from $\exval(\partialNorm{O_{\circuitDepth}^{(w_0)}}{w}^2)$ for arbitrary $w_0$ (including $w_0 = 2N$ for no truncation). This can be done recursively using the following proposition (before which we first have to sate a definition). 

\begin{definition}
    Similarly to proposition~\ref{thm:backflow} from the main text, we define:
    \begin{itemize}
        \item $p_-(w)$ is the probability that applying a random length-4 Majorana gate on a random length-$w$ Majorana decreases its length;
        \item $p_+(w)$ is defined analogously for length increase;
        \item $p_=(w)$ is defined analogously for no change in the length. 
    \end{itemize}
\end{definition}

\begin{proposition}\label{thm:poly-time-algo-for-squared-error-oldpart1}
For any $w \leq w_0$, the partial 2-norm of the observables (definition~\ref{def:majorana-length-partial-norm}) follows the following recursive formula:
        \[\begin{split}
             \exval(\partialNorm{O_{i+1}^{(w_0)}}{w}^2) =\ & \frac{p_+(w-2)}{2}  \exval(\partialNorm{O_i^{(w_0)}}{w-2}^2) + \frac{p_-(w+2)}{2}  \exval(\partialNorm{O_i^{(w_0)}}{w+2}^2)
            \\& + \left(1 - \frac{p_+(w)}{2} - \frac{p_-(w)}{2}\right)  \exval(\partialNorm{O_i^{(w_0)}}{w}^2).
        \end{split}\]
        Otherwise, when $w > w_0$, then $ \exval(\partialNorm{O_i^{(w_0)}}{w}^2) = 0$.
\end{proposition}

\begin{proof}
If $w > w_0$, then we trivially get that $ \exval(\partialNorm{O_i^{(w_0)}}{w}^2) = 0$. We thus suppose $w \geq w_0$. Equation~\ref{eq:expected-c-squared} told us that, using the fact $\exval(\cos^2 \theta) = \exval(\sin^2 \theta) = \frac{1}{2}$:
            \[\begin{split}
                \exval(c_{i+1, \majorana}^{(w_0) 2}) = &\ \prob([\majorana, \majoranaGate{i}] = 0) \exval(c_{i, \majorana}^{(w_0) 2}) + \frac{1}{2} \prob(\{\majorana, \majoranaGate{i}\} = 0) \exval(c_{i, \majorana}^{(w_0) 2})
                \\& + \frac{1}{2} \sum_u \sum_{\majoranaOther \in \Gamma_u} \prob(\majoranaOther, \majoranaGate{i} \mapsto \majorana) \exval(c_{i\majoranaOther}^{(w_0) 2}).
            \end{split}\]
            We sum both sides for all $\majorana \in \Gamma_w$. As mentioned in point~\ref{enum:proof-c-only-depends-on-w} of the proof of lemma~\ref{lma:c-moments}, $\prob([\majorana, \majoranaGate{i}] = 0)$ and $\prob(\{\majorana, \majoranaGate{i}\} = 0)$ only depend on $w$; and similarly, that $\exval(c_{i\majoranaOther}^{(w_0) 2})$ only depends on $u$. We are thus interested in the following value for any given $u$:
            \[\sum_{\majoranaOther \in \Gamma_u}\sum_{\majorana \in \Gamma_w} \prob(\majoranaOther, \majoranaGate{i} \mapsto \majorana) = \sum_{\majoranaOther \in \Gamma_u}\sum_{\majorana \in \Gamma_w} \prob(\{\majoranaOther, \majoranaGate{i}\} = 0 \land \majoranaOther \phaselessprod \majoranaGate{i} = \majorana) = \sum_{\majoranaOther \in \Gamma_u}\prob(\{\majoranaOther, \majoranaGate{i}\} = 0 \land |\majoranaOther \phaselessprod \majoranaGate{i}| = w).\]

            If $u = w-2$, then this is $p_+(w-2)$. If $u = w+2$, then this is $p_-(w+2)$. Otherwise, this is zero. Hence, this allows us to find that, overall:
            \[\begin{split}
            \exval(\partialNorm{O_{i+1}^{(w_0)}}{w}^2) =\ & \prob(\text{$w$ commutes}) \exval(\partialNorm{O_{i}^{(w_0)}}{w}^2) + \prob(\text{$w$ anticommutes}) \frac{1}{2} \exval(\partialNorm{O_{i}^{(w_0)}}{w}^2)
            \\& + \frac{1}{2} p_+(w-2) \exval(\partialNorm{O_{i}^{(w_0)}}{w-2}^2) + \frac{1}{2} p_-(w+2) \exval(\partialNorm{O_{i}^{(w_0)}}{w+2}^2).
            \end{split}\]
            We then directly get the result, by using the facts $\prob(\text{$w$ commutes}) = 1 - p_-(w) - p_+(w)$ and $\prob(\text{$w$ anticommutes}) = p_-(w) + p_+(w)$.
\end{proof}

We can combine Proposition~\ref{thm:poly-time-algo-for-squared-error-oldpart2} and Proposition~\ref{thm:poly-time-algo-for-squared-error-oldpart1} to provide a poly-time algorithm to evaluate and plot $\exval((f - f^{(w_0)})^2)$. This is captured by the following Theorem. 
\begin{theorem}\label{thm:poly-time-algo-for-squared-error-new}
    The mean squared error $\exval((f - f^{(w_0)})^2)$ for a random Majorana circuit as defined in Def.~\ref{def:circuitmodel} can be computed exactly classically in polynomial time. 
\end{theorem}

\begin{proof}
    One can use the iterative expression for the average contribution of each Majorana length at each time step in Proposition~\ref{thm:poly-time-algo-for-squared-error-oldpart1} for both the case of a cut off $w_0$ and the case of no cut off $w_0 = 2N$. These expressions can be then fed into the expression for the mean squared error given in point 2 of Proposition~\ref{thm:poly-time-algo-for-squared-error-oldpart2}. 
\end{proof}

This algorithm is validated in figure~\ref{fig:poly-time-algo-for-squared-error}. 

    \begin{figure}
        \centering
        \includegraphics[width=0.5\linewidth]{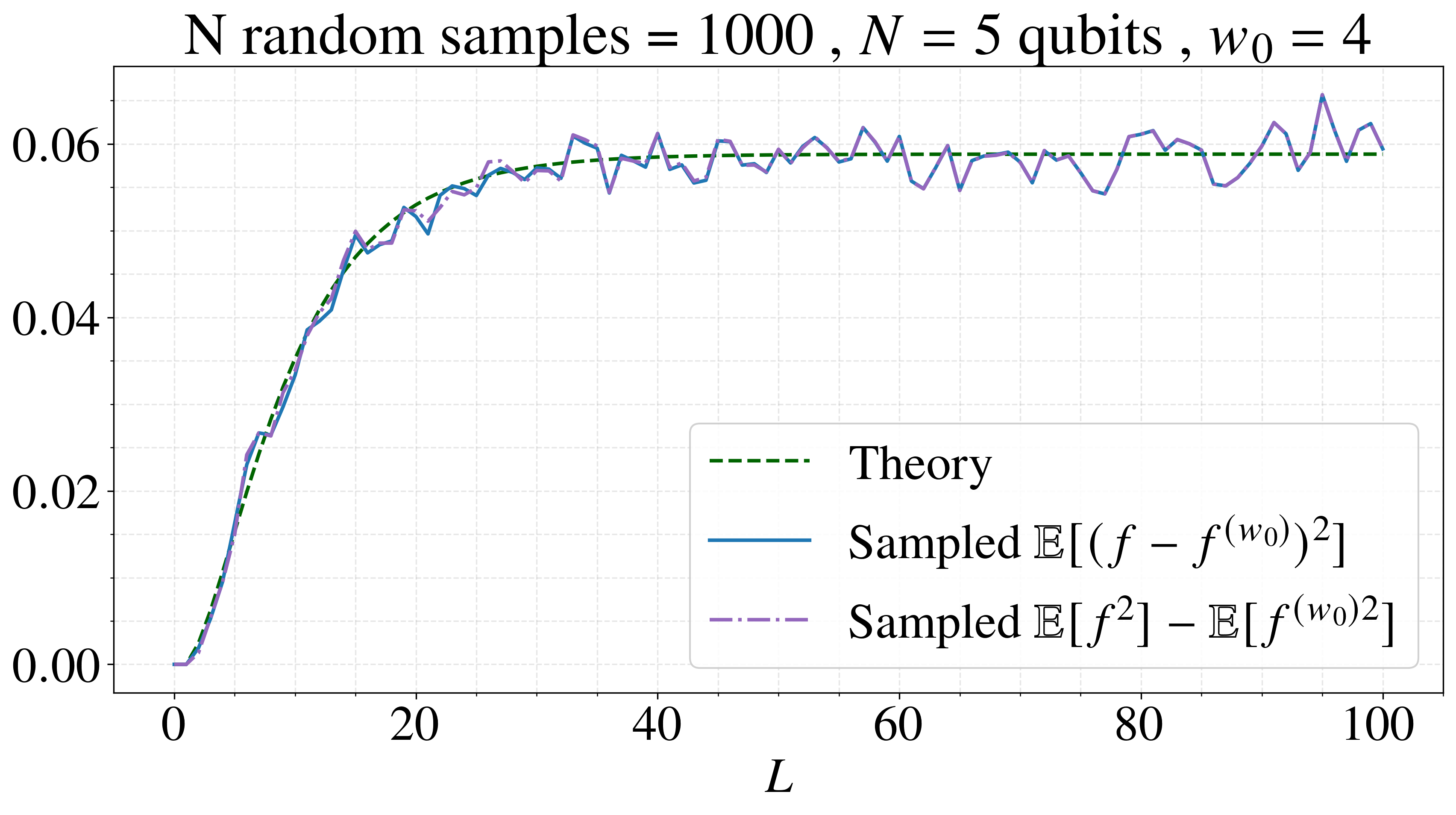}
        \caption{Numerical validation of the algorithm described in theorem~\ref{thm:poly-time-algo-for-squared-error-new}. The green dashed line is the result of this algorithm and the blue and purple lines are computed from performing full and truncated Majorana propagation with 1000 randomly sampled circuits. The initial observable is a random length-2 observable with a random $\pm$ sign.}
        \label{fig:poly-time-algo-for-squared-error}
    \end{figure}

\subsection{Bounding the mean squared error}\label{app:bounderror}

In this section, we proceed to find an upper bound on the error. As shown in Fig.~\ref{fig:errors} there are multiple regimes to the error - i) at short times the error initially increases before ii) starting to exponentially decay with depth to an exponential in $N$ saturation point. Our bound will aim to capture these two regimes separately using a Markov chain analysis. 

In particular we will consider the following Markov chain that capture the probability that the length of a Monomial increases, decreases or stays the same after each layer of the random circuit, i.e., after each random Majorana rotation, as specified in Eq.~\eqref{eq:singleupdate}. 

\begin{definition}
    Let $W_i$ be a random variable corresponding to the length of the Monomial at step $i$. Let $\prob(W_i = w' \mid W_{i-1} = w)$ denote the probability that the length of a monomial increases from $w$ to $w'$ from iteration $i-1$ to $i$. This procedure is captured by a Markov chain with the transition probabilities:
    \begin{enumerate}
        \item Probability that length decreases: $\prob(W_i = w - 2 \mid W_{i-1} = w) = \frac{p_-(w)}{2}$ .
        \item Probability that length increases: $\prob(W_i = w + 2 \mid W_{i-1} = w) = \frac{p_+(w)}{2}$.
        \item Probability that length stays the same: $\prob(W_i = w \mid W_{i-1} = w) = 1 - \frac{p_-(w)}{2} - \frac{p_+(w)}{2}.$
    \end{enumerate}
    
    Moreover, we let $\prob(W_0 = w) = \exval(\partialNorm{O_0}{w}^2) / \exval(\|O_0\|^2)$ where $\partialNorm{P}{w}$ is the partial 2-norm (definition~\ref{def:majorana-length-partial-norm}). Note that, since $\sum_{w=0}^{2\nModes} \exval(\partialNorm{O_0}{w}^2) = \exval(\|O_0\|^2)$, this is a valid probability distribution.
    
\end{definition}

This Markov chain intuitively captures that probability that the length of a Monomial increases, decreases or stays the same after each random rotation. In particular, we recall from Section~\ref{app:backflow_suppresion} that if the generator of rotation is length 4 then then the length can only increase or decrease by 2 with the probability of this occurring given by $p_{\pm}/2$ where the factor of $1/2$ comes the fact that $  \exval( \sin(\theta)^2 ) = \exval( \cos(\theta)^2 ) = 1/2$ when $\theta$ is randomly chosen in the range $0$ to $2\pi$. 
Alternatively, this Markov chain can be motivated by noting that the definition was already implicit in Proposition~\ref{thm:poly-time-algo-for-squared-error-oldpart1}, where it governed the recursive expression for how to update the length $w$ contributions to 2-norm of the propagated observable at each time step.

The Markov chain defined above does not take into account the fact that we truncate all monomials with length greater than $w_0$. To account for this, we introduce the variable $\tau$ that captures the first time that a path crosses outside the truncation zone. That is, the first time such that that the length $W_i$ of the monomial at time $i$ is greater than the truncation $w_0$:
\begin{equation}\label{eq:tau}
        \tau := \min_i \{W_i > w_0\} \, .
\end{equation}
The probability that a \textit{non-truncated} path has length $w$ at time $i$ is given by $\prob(W_i = w \land \tau > i)$: the probability that it has length $w$ (i.e., $W_i = w$) and it has not been truncated yet (i.e., $\tau > i$). The following proposition relates this probability to the average coefficient of the length $w$ terms of the back-propagated observable at iteration $i$ that previously appeared in our expressions for the mean squared error in Propositions~\ref{thm:poly-time-algo-for-squared-error-oldpart2}.

\begin{lemma}\label{lma:link-prob-2-norm}
    For all iterations $i$ and monomial lengths $w$, writing $\partialNorm{P}{w}$ to be the partial 2-norm (definition~\ref{def:majorana-length-partial-norm}): 
    \[\prob(W_i = w \land \tau > i) = \frac{\exval(\partialNorm{O_i^{(w_0)}}{w}^2)}{\exval(\|O_0\|^2)}.\]
\end{lemma}

\begin{proof}
    They have the same initial condition, i.e.~this is true for $i = 0$. Then, they have the same recurrence formula. Indeed, the case $w > w_0$ is easy. When $w \leq w_0$:
    \[p_{i+1} = \prob(W_{i+1} = w \land \tau > i+1)  = \prob(W_{i+1} = w \land \tau > i).\]

    Then, using the law of total probability:
    \[p_{i+1, w} = \sum_u \prob(W_{i+1} = w \land \tau > i \land W_i = u) = \sum_u \prob(W_{i+1} = w \mid \tau > i \land W_i = u) \underbrace{\prob(\tau > i \land W_i = u)}_{= p_{i, u}}.\]

    Now, we know that $\prob(W_{i+1} = w \mid \tau > i \land W_i = u) = \prob(W_{i+1} = w \mid W_i = u)$: $\tau > i$ only tells us information on $W_1, \ldots, W_i$, but $W_{i+1}$ only depends on $W_i$, not on the rest of its past. Hence:
    \[p_{i+1, w} = p_{i, w-2} \frac{p_+(w-2)}{2} + p_{i, w} \left(1 - \frac{p_+(w)}{2} - \frac{p_-(w)}{2}\right) + p_{i, w+2} \frac{p_-(w+2)}{2}.\]

    This matches the recurrence relation of $\exval(\partialNorm{O_i^{(w_0)}}{w}^2)$ as given in Proposition~\ref{thm:poly-time-algo-for-squared-error-oldpart1}, finishing the proof since $\exval(\|O_0\|^2)$ is just a constant.
\end{proof}

Next we can use $\tau$, the first time that a path crosses outside the truncation zone, to rewrite the expressions for the second moment of the truncated expectation value and the mean squared error. Intuitively, the mean squared error is given by probability that a path after $\circuitDepth$ layers would have contributed, $\probContribute{W_{\circuitDepth}}$, times the probability that that path was truncated because the first truncation time is that iteration, $ I(\tau \leq \circuitDepth))$. This is captured by the following lemma. 

\begin{lemma}\label{lma:link-result-expectation}
For all iterations $i$ and cutoffs $w_0$:
\begin{enumerate}
    \item The second moment of the truncated function: $\exval(f^{(w_0) 2}) = \exval(\|O_0\|^2) \exval(\probContribute{W_{\circuitDepth}} I(\tau > \circuitDepth))$ 
    
    \item The mean square error can be written as: $\exval((f - f^{(w_0)})^2) = \exval(\|O_0\|^2) \exval(\probContribute{W_{\circuitDepth}} I(\tau \leq \circuitDepth))$
\end{enumerate}
where $I(\tau > \circuitDepth)$ is the indicator function that equals 1 iff $\tau > \circuitDepth$ and otherwise is 0. 
\end{lemma}
\begin{proof}
    We consider both cases separately.
    \begin{itemize}
        \item By theorem~\ref{thm:poly-time-algo-for-squared-error-oldpart2} and lemma~\ref{lma:link-prob-2-norm}:
        \[\begin{split}
            \frac{\exval(f^{(w_0) 2})}{\exval(\|O_0\|^2)} & = \sum_{w=0}^{2\nModes} \frac{\exval(\partialNorm{O_{\circuitDepth}^{(w_0)}}{w}^2)}{\exval(\|O_0\|^2)} \probContribute{w} = \sum_{w=0}^{2\nModes} \prob(W_{\circuitDepth} = w \land \tau > \circuitDepth) \probContribute{w} 
            \\& = \sum_{w=0}^{2\nModes} \exval(I(W_{\circuitDepth} = w) I(\tau > \circuitDepth)) \probContribute{w} = \exval\left(\sum_{w=0}^{2\nModes} I(W_{\circuitDepth} = w)\probContribute{w} I(\tau > \circuitDepth)\right).
        \end{split}\]

        However, we notice that $\sum_{w=0}^{2\nModes} I(W_{\circuitDepth} = w)\probContribute{w} = \probContribute{W_{\circuitDepth}}$, so this simplifies exactly to our result.
        
        Note that another way of expressing this result is that $\exval(f^{(w_0) 2}) / \exval(\|O_0\|^2)= \exval(\probContribute{W_{\circuitDepth}} \mid \tau > i)\prob(\tau > i)$. This is direct from the expression above, since $\prob(W_{\circuitDepth} = w \land \tau > i) = \prob(W_{\circuitDepth} = w \mid \tau > \circuitDepth)\prob(\tau > \circuitDepth)$.

        \item By theorem~\ref{thm:poly-time-algo-for-squared-error-oldpart2} and what we just proved:
        \[\frac{\exval((f - f^{(w_0)})^2)}{\exval(\|O_0\|^2)} = \frac{\exval(f^2) - \exval(f^{(w_0)2})}{\exval(\|O_0\|^2)} = \exval(\probContribute{W_{\circuitDepth}}) - \exval(\probContribute{W_{\circuitDepth}} I(\tau > \circuitDepth)) = \exval(\probContribute{W_{\circuitDepth}} I(\tau \leq \circuitDepth)).\]
    \end{itemize}
\end{proof}

We now continue the game of rewriting the mean squared error using the notion of the minimal crossing time. This time we start from the fact that the average error is given by the probability that a length $w$ path contributes, times the probability that that term was truncated, summed over all possible lengths $w$ (see Eq.~\eqref{eq:firstrewrite} below). We then show that this can be rewritten as the probability that the first crossing time is $t$ multiplied by the probability that a path that crosses at this time eventually contributes, summed over all possible crossing times (see Eq.~\eqref{eq:secondrewrite}). Finally, this is upper bounded by looking at the crossing term where the probability of being truncated and then eventually contributing is largest (see Eq.~\eqref{eq:boundstep}). These steps are formally captured by the following lemma and proof. 

\begin{lemma}\label{lma:upperbound-by-u}
    For all truncation levels $w_0$ and circuit depths $\circuitDepth$ the mean squared error can be written as 
    \begin{equation}\label{eq:upperbound-by-u1}
        \frac{\exval((f - f^{(w_0)})^2)}{\exval(\|O_0\|^2)} = \sum_{t=0}^{\circuitDepth} \exval(\probContribute{W_{\circuitDepth-t}} \mid W_0 = w_0 + 2) \prob(\tau = t)
    \end{equation}
    where we recognize that $\exval((f - f^{(w_0)})^2) = (\exval(\probContribute{W_{t}} \mid W_0 = w_0 + 2) * \prob(\tau = t))_{\circuitDepth}$ is a convolution and intuitively, $\exval(\probContribute{W_k} \mid W_0 = w_0 + 2)$ captures the likelihood that a path which starts outside the truncation zone at time $t$ would have contributed at time $\circuitDepth$, i.e., after $\circuitDepth-t$ more iterations. 
    We can use Eq.~\eqref{eq:upperbound-by-u1} to further bound the mean squared error as 
    \begin{equation}
        \frac{\exval((f - f^{(w_0)})^2)}{\exval(\|O_0\|^2)} \leq \max_{0 \leq k \leq \circuitDepth} \exval(\probContribute{W_k} \mid W_0 = w_0 + 2) \, . 
    \end{equation}
\end{lemma}
\begin{proof}
    This is direct, by lemma~\ref{lma:link-result-expectation}:
    \begin{align}
        \frac{\exval((f - f^{(w_0)})^2)}{\exval(\|O_0\|^2)} & = \exval(\probContribute{W_{\circuitDepth}} I(\tau \leq \circuitDepth)) = \sum_w \probContribute{w} \prob(W_{\circuitDepth} = w \land \tau \leq \circuitDepth) \label{eq:firstrewrite}
        \\& = \sum_{t=0}^{\circuitDepth} \sum_w \probContribute{w} \prob(W_{\circuitDepth} = w \land \tau = t)
        \\& = \sum_{t=0}^{\circuitDepth} \sum_w \probContribute{w} \prob(W_{\circuitDepth} = w \mid \tau = t)\prob(\tau = t)
        \\& = \sum_{t=0}^{\circuitDepth} \sum_w \probContribute{w} \prob(W_{\circuitDepth} = w \mid W_t = w_0+2, |W_{t-1}| \leq w_0, \ldots, |W_{0}| \leq w_0)\prob(\tau = t)
        \\& = \sum_{t=0}^{\circuitDepth} \sum_w \probContribute{w} \prob(W_{\circuitDepth} = w \mid W_t = w_0+2)\prob(\tau = t) 
        \\& = \sum_{t=0}^{\circuitDepth} \exval(\probContribute{W_{\circuitDepth}} \mid W_t = w_0 + 2) \prob(\tau = t) \label{eq:secondrewrite}
        \\& = \sum_{t=0}^{\circuitDepth} \exval(\probContribute{W_{\circuitDepth-t}} \mid W_0 = w_0 + 2) \prob(\tau = t)
    \end{align}
    This finishes the first part of the proof. For the second part, we can just continue the chain of inequality:
    \begin{align}
                \sum_{t=0}^{\circuitDepth} \exval(\probContribute{W_{\circuitDepth-t}} \mid W_0 = w_0 + 2) \prob(\tau = t) & \leq \max_{0 \leq t \leq \circuitDepth} \{\exval(\probContribute{W_{\circuitDepth-t}} \mid W_0 = w_0 + 2)\} \sum_{t=0}^{\circuitDepth} \prob(\tau = t) \label{eq:boundstep}
        \\& = \max_{0 \leq k \leq \circuitDepth} \{\exval(\probContribute{W_{k}} \mid W_0 = w_0 + 2)\} \prob(\tau \leq t)
        \\& \leq \max_{0 \leq k \leq \circuitDepth} \{\exval(\probContribute{W_{k}} \mid W_0 = w_0 + 2)\}.
    \end{align}

\end{proof}

Thus we see that the mean squared error can be bounded simply by bounding $\prob(\tau = t) $ and/or the iterations $k$ for which the quantity $\exval(\probContribute{W_k} \mid W_0=w_0+2)$ is maximum.
Let us introduce a more compact notation for this key quantity by defining: 
\begin{equation}
    U_i(w) := \exval(\probContribute{W_i} \mid W_0=w) \,  
\end{equation}
and 
\begin{equation}
     \Delta U_i(w) := U_i(w+2) - U_i(w) \, .
\end{equation}
We plot these quantities in Fig.~\ref{fig:propsU}. As one can see, the following key properties hold for $U_i(w)$. For simplicity below we assume $\nModes$ is odd with the generalization to even $\nModes$ straightforward. 

    \begin{figure}
        \centering
        \includegraphics[width=0.95\linewidth]{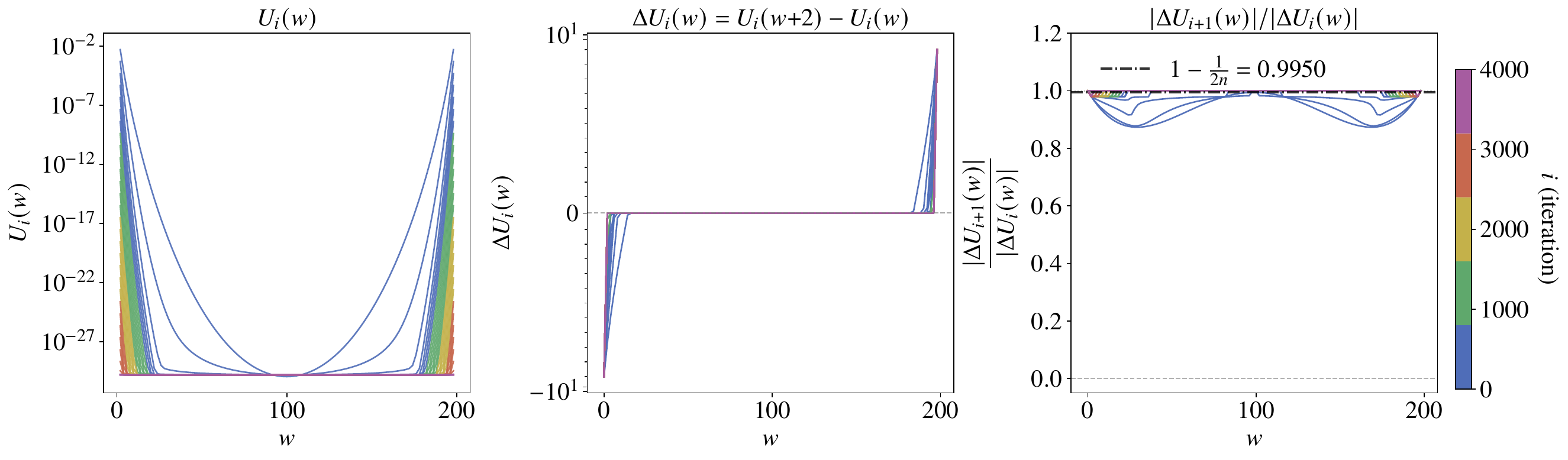}
        \caption{
Visualisation of the properties of \( U_i(w) \) claimed in Lemma~\ref{lma:SummaryPropertiesOfU}. 
\textbf{Left:} The probability distribution \( U_i(w) \) as a function of the Majorana length \( w \) for several iteration steps \( i \). 
\textbf{Centre:} The finite difference \( \Delta U_i(w) = U_i(w{+}2) - U_i(w) \), showing how \( U_i(w) \) changes between adjacent even values of \( w \). In particular we see that $\Delta U_i(w) < 0$ for $2 < w < \nModes$ (as claimed in part 1 of Lemma~\ref{lma:SummaryPropertiesOfU}).
\textbf{Right:} The ratio \( \lvert \Delta U_{i+1}(w) \rvert / \lvert \Delta U_i(w) \rvert \), quantifying the relative contraction of successive iterations. 
The horizontal dashed–dotted line corresponds to the theoretical bound of \( 1 - \tfrac{1}{2\nModes} \). The ratio is less than this bound in agreement with part 2 of Lemma~\ref{lma:SummaryPropertiesOfU}.
Colours encode the iteration index \( i \).}
        \label{fig:propsU}
    \end{figure}

\begin{lemma}\label{lma:SummaryPropertiesOfU}
    For any time step $i \geq 0$ and length $ 2 \leq w \leq 2\nModes-2$ we have that $ U_i(w)$ decreases exponentially with $i$ to a saturation point that is exponentially small in $\nModes$: 
     \begin{equation}\label{eq:upperbound-u}
         U_i(w) \leq \left(1 - \frac{1}{2\nModes}\right)^i \probContribute{w} + \frac{1}{2^{\nModes-1} + 1} \, .
     \end{equation}
\end{lemma}

\noindent\textit{Proof intuition}. To remain focused on the core of our argument we continue deriving our main error bound here and defer formally proving Eq.~\eqref{eq:upperbound-u} to Appendix~\ref{app:additionalprops}. However, the idea is that it relies on the following observations, which can be seen on figure~\ref{fig:propsU}.
\begin{enumerate}
    \item For any $i$, then $U_i(w)$ is symmetric around $N$, i.e.\ $U_i(2N - w)= U_i(w)$ for all $w$.
    
    This allows to prove properties of $U_i(w)$ by induction on $i$, even if they change behavior at the midpoint. The next observation is such an example. 
    \item If $w \leq N-1$, then $\Delta U_i(w) \leq 0$. If $w \geq N-1$, then $\Delta U_i(w) \geq 0$. In other words, $U_i(w)$ first decreases and then increases; its minimum is the middle point.
    \item $U_i(w)$ flattens out at an exponential speed. More mathematically, for all $i \geq 0$ and $w$, 
    \[|\Delta U_{i+1}(w)| \leq (1-  1/2\nModes) |\Delta U_i(w)|.\]
    \item As the curve $U_i(w)$ flattens out, it converges to $\frac{1}{2^{N - 1} + 1}$, which is exponentially small. More precisely, there exists a probability distribution $\pi(w)$ of support $\{2, 4, \ldots 2N-2\}$ such that for any step $i \geq 0$, we have the following invariant:
    \[\sum_{w \in \{2, 4, \ldots 2N-2\}} \pi(w) U_i(w) = \frac{1}{2^{\nModes-1} + 1}.\]
\end{enumerate}

In Appendix~\ref{app:additionalprops} we also derive the following bound on $\prob(\tau = t) $.

\begin{lemma}\label{lma:upperbound-tau-equal-t}
    We have, for $\nModes$ sufficiently large and all $t \geq 0$:
    \[\prob(\tau = t) \leq 4^{w_0} \left(1 -\frac{3}{4\nModes}\right)^{t-1}.\]
    Note that $1 - 3/4\nModes < 1-1/2\nModes$. This is thus decreasing faster than $U_i$ (point 4 of Lemma~\ref{lma:SummaryPropertiesOfU}). This will matter when we want to combine it with Lemma~\ref{lma:upperbound-by-u}.
\end{lemma}

We can now combine Lemma~\ref{lma:upperbound-tau-equal-t} and Lemma~\ref{lma:SummaryPropertiesOfU} with Lemma~\ref{lma:upperbound-by-u} to upper bound the mean squared error. 

\begin{theorem}[Bound on Mean Squared Error]\label{thm:upper-bound-mse}
    We have, for $\nModes$ sufficiently large:
    \[\frac{\exval((f - f^{(w_0)})^2)}{\exval(\|O_0\|^2)} \leq \frac{1}{2^{\nModes-1} + 1} + \probContribute{w_0+2}\min\left\{1, 5\nModes 4^{w_0} \left(1 - \frac{1}{2\nModes}\right)^{\circuitDepth}\right\}.\]
    Moreover, in particular:
    \begin{itemize}
        \item When $w_0 = \Theta(1)$:
        \[\frac{\exval((f - f^{(w_0)})^2)}{\exval(\|O_0\|^2)} \leq \frac{1}{2^{\nModes-1}} + \probContribute{w_0+2} = O(\nModes^{-w_0 + 2}).\]
        \item We can take $\circuitDepth_0 = O(\nModes^2)$ sufficiently large such that, for any $\circuitDepth \geq \circuitDepth_0$:
        \[\frac{\exval((f - f^{(w_0)})^2)}{\exval(\|O_0\|^2)} = \frac{\Theta(1)}{2^{\nModes-1} + 1} = \Theta(2^{-\nModes}).\]
    \end{itemize}
\end{theorem}

\begin{proof}[Proof of theorem~\ref{thm:upper-bound-mse}]
    We prove each inequality separately.
    \begin{itemize}
        \item We first want to show $\exval((f - f^{(w_0)})^2) / \exval(\|O_0\|^2) \leq \frac{1}{2^{\nModes-1} + 1} + \probContribute{w_0+2}$. By lemma~\ref{lma:upperbound-by-u}:
        \[\frac{\exval((f - f^{(w_0)})^2)}{\exval(\|O_0\|^2)} \leq \max_{0 \leq k \leq \circuitDepth} U_k(w_0+2).\]

        However, by lemma~\ref{lma:SummaryPropertiesOfU} (point 3):
        \[U_k(w_0 + 2) \leq \left(1 + \frac{1}{2\nModes}\right)^k \probContribute{w_0+ 2} + \frac{1}{2^{\nModes-1} + 1} \leq \probContribute{w_0 + 2} + \frac{1}{2^{\nModes-1} + 1}.\]

        Combining these two equations, we get exactly our result.

        \item We now want to show $\exval((f - f^{(w_0)})^2) / \exval(\|O_0\|^2)\leq \frac{1}{2^{\nModes-1} + 1} + 5\nModes \probContribute{w_0+2} 4^{w_0} \left(1 - \frac{1}{2\nModes}\right)^{\circuitDepth}$. By lemma~\ref{lma:upperbound-by-u}:
        \[\frac{\exval((f - f^{(w_0)})^2)}{\exval(\|O_0\|^2)} = \sum_{t=0}^{\circuitDepth} U_{\circuitDepth-t}(w_0+2) \prob(\tau = t).\]
    
        The idea is that $\prob(\tau = t)$ only keeps terms where $t$ is small. But then, this is is just the sum of terms that are close to $U_{\circuitDepth}(w_0 +2)$, which decrease towards a horizontal asymptote as $\circuitDepth$ increases. More formally, using lemmas~\ref{lma:SummaryPropertiesOfU} and~\ref{lma:upperbound-tau-equal-t}:
        \[\begin{split}
            \frac{\exval((f - f^{(w_0)})^2)}{\exval(\|O_0\|^2)} &\leq \frac{1}{2^{\nModes-1}+1}\sum_{t=0}^{\circuitDepth}  \prob(\tau = t) + \sum_{t=0}^{\circuitDepth} \prob(\tau = t) \left(1 - \frac{1}{2\nModes}\right)^{\circuitDepth-t} \probContribute{w_0+2}
            \\& \leq \frac{1}{2^{\nModes-1} + 1} \prob(\tau \leq t) + \probContribute{w_0 + 2} 4^{w_0} \sum_{t=0}^{\circuitDepth} \left(1 - \frac{3}{4\nModes}\right)^{t-1} \left(1 - \frac{1}{2\nModes}\right)^{\circuitDepth-t}.
        \end{split}\]
    
        Just like the proof of lemma~\ref{lma:upperbound-by-u}, we know we have $\prob(\tau \leq t) \leq 1$. We are thus only left with evaluating the sum, which is just a geometric series. Note that, very nicely:
        \[r = \frac{1 - 3/4\nModes}{1 - 1/2\nModes} = \frac{4\nModes - 3}{4\nModes -2} < 1.\]
    
        Hence, we have that:
        \[S = \sum_{t=0}^{\circuitDepth} \left(1 - \frac{3}{4\nModes}\right)^{t} \left(1 - \frac{1}{2\nModes}\right)^{-t} = \sum_{t=0}^{\circuitDepth} \left(\frac{1-3/4\nModes}{1-1/2\nModes}\right)^t \leq \sum_{t=0}^{\infty} r^t = \frac{1}{1-r} = \frac{1}{1 - \frac{4\nModes-3}{4\nModes-2}} = 4\nModes-2.\]
    
        This gives us our result, supposing $\nModes \geq 4$ (since Eq.~\eqref{lma:delta-u-increases} already forces this hypothesis):
        \[\begin{split}
            \frac{\exval((f - f^{(w_0)})^2)}{\exval(\|O_0\|^2)} &\leq \frac{1}{2^{\nModes-1} + 1}  + \probContribute{w_0 + 2} 4^{w_0} \left(1 - \frac{3}{4\nModes}\right)^{-1} \left(1 - \frac{1}{2\nModes}\right)^{\circuitDepth} S
            \\& = \frac{1}{2^{\nModes-1} + 1}  + \probContribute{w_0 + 2} 4^{w_0} \left(1 - \frac{1}{2\nModes}\right)^{\circuitDepth} \nModes\cdot \underbrace{4\frac{1-1/2\nModes}{1-3/4\nModes}}_{\leq 5.}.
        \end{split}\]
    \end{itemize}
\end{proof}

Finally, we can polish this bound to obtain the aesthetically simpler (but looser) Theorem~\ref{thm:main-text-upper-bound-mse} presented in the main text, which we repeat here for convenience. 

\begin{theorem}[Simplified Bound on Mean Squared Error (Theorem~\ref{thm:main-text-upper-bound-mse} in the main text)]\label{thm:upper-bound-mse-tidied}
    Consider a random Majorana circuit as defined in Def.~\ref{def:circuitmodel} and suppose that $w_0 \geq 2$. Then for circuits with $L \;\le\; 2 N \ln(5 e w_0 4^{w_0})$: 
    \[\frac{\exval((f - f^{(w_0)})^2)}{\exval(\|O_0\|^2)}  \leq \frac{1}{2^{N-1}} +  \left(\frac{e w_0 }{N}\right)^{w_0 / 2}.\]
    Otherwise, for deeper circuits with $L \geq 2 N \ln(5 e w_0 4^{w_0})$ we have: 
    \[ \frac{\exval((f - f^{(w_0)})^2)}{\exval(\|O_0\|^2)}  \leq \frac{1}{2^{N-1}} + 5e w_0\,\left(\frac{16 e w_0}{ N}\right)^{w_0/2} e^{-L/2N} .\]
\end{theorem}

\begin{proof}
To begin with, we can use equation~\ref{eq:upperboundonPw} together with $\frac{1}{2^{N-1}+1}\le \frac{1}{2^{N-1}}$ and $\left(1-\frac{1}{2N}\right)^{\!L}\;\le\; \exp\!\left(-\frac{L}{2N}\right)$  to simplify the result of Theorem~\ref{thm:upper-bound-mse}:
\begin{equation}
\label{eq:thm20-base}
\frac{\exval((f - f^{(w_0)})^2)}{\exval(\|O_0\|^2)}
\;\le\; \frac{1}{2^{N-1}}
\;+\; \left(\frac{e (w_0 + 2)}{2 N}\right)^{(w_0 + 2)/2}\,\min\!\left\{1,\;5N\,4^{w_0}\! \exp\left(\frac{-L}{2N}\right)\right\}.
\end{equation}

\noindent\textit{Shallow case (first line of Thm.~\ref{thm:upper-bound-mse-tidied}).} We consider the left branch ``$1$'' of the minimum in~\eqref{eq:thm20-base}, and use the fact $w_0 \geq 2$ and hence that $\frac{w_0}{2} \leq \frac{w_0 + 2}{2} \leq w_0$:
\[
\frac{\exval((f - f^{(w_0)})^2)}{\exval(\|O_0\|^2)} \leq \frac{1}{2^{N-1}}
\;+\; \left(\frac{e (w_0 + 2)}{2 N}\right)^{(w_0 + 2)/2} \;\leq\; \frac{1}{2^{N-1}}
\;+\;\left(\frac{e\,w_0}{N}\right)^{\!w_0/2}.
\]
This is exactly the first inequality of Theorem~\ref{thm:upper-bound-mse}.

\noindent\textit{Deep case (second line of Thm.~\ref{thm:upper-bound-mse-tidied}).} 
We use the second branch of the minimum in~\eqref{eq:thm20-base}, together with the fact $w_0 \geq 2$ and hence $\frac{w_0 + 2}{2} \leq w_0$: 
\[\begin{split}
    \frac{\exval((f - f^{(w_0)})^2)}{\exval(\|O_0\|^2)} & \leq \frac{1}{2^{n-1}} + \left(\frac{e (w_0 + 2)}{2 N}\right)^{(w_0 + 2)/2}5N 4^{w_0} e^{-L/2N} \leq \frac{1}{2^{n-1}} + \left(\frac{e w_0}{N}\right)^{w_0/2 + 1}5N 4^{w_0} e^{-L/2N}
    \\& = \frac{1}{2^{n-1}} +  \left(\frac{ e w_0}{N}\right)^{w_0/2}\frac{e w_0}{N} 5N 16^{w_0/2} e^{-L/2N} = \frac{1}{2^{n-1}} + 5e w_0 \left(\frac{16 e w_0}{N}\right)^{w_0/2} e^{-L/2N}.
\end{split}\]

This is exactly the second bound.

\noindent\textit{Branch selection}. Both inequalities above hold for any $L$. We however want to find which is best given some $L$. This is just a simple inequality to compare the two bounds:
\[\frac{1}{2^{N-1}} + \left(\frac{e w_0}{N}\right)^{w_0/2} \leq \frac{1}{2^{n-1}} + 5e w_0 \left(\frac{16 e w_0}{N}\right)^{w_0/2} e^{-L/2N} \iff e^{L/2N} \leq 5 e w_0 16^{w_0/2} \iff L \leq 2N \ln\left(5e w_0 4^{w_0}\right).\]

In other words, the shallow case gives a tighter bound if and only if $L \leq 2N \ln\left(5e w_0 4^{w_0}\right)$.

\end{proof}

\subsection{Time complexity of Majorana Propagation with Length Truncation}

Theorem~\ref{thm:upper-bound-mse} (and equivalently Theorem~\ref{thm:upper-bound-mse-tidied}) can be interpreted in different ways. First, it also gives us a probabilistic bound thanks to Markov's inequality:
 \[\prob(|f -f^{(w_0)}| \geq \delta) \leq \frac{\exval((f - f^{(w_0)})^2)}{\delta^2}.\]
Hence, for instance, when $\circuitDepth \gg 1$ is large, this tells us that $f$ and $f^{(w_0)}$ are exponentially close with high probability. Another way to state this theorem is that having a truncation length $w_0$ means using $O(\nModes^{w_0})$ space and hence $O(\circuitDepth\cdot \nModes^{w_0})$ time. In other words, this bound allows to link the resources used by Majorana propagation to the error of the algorithm. In this section, we formalize these claims. In particular, we provide a proof for Theorem 2 in the main text which we repeat here for convenience. 

\begin{theorem}[Time complexity]
\label{thm:resources-appendices}
    Let $U$ be a randomly sampled circuit from an $L$-layered random Majorana rotation circuit on $2 N$ modes (as defined in definition~\ref{def:circuitmodel}), and let $O$ be random homogeneous observable. Moreover, let $\epsilon, \delta > 0$ such that $\epsilon^{-1}\delta^{-1} \leq 2^{o(N)}$.
    There exists a truncation length $w_0$ such that Majorana propagation runs in time
    $L \cdot \nModes^{\mathcal{O}\left(\log(\epsilon^{-1} \delta^{-1})\right)},$
    and outputs a value $f_L^{(w_0)}(U, O)$, such that
    \begin{align}
        \abs{f_L^{(w_0)}(U,O) - f_L(U,O) } \leq \epsilon,
    \end{align}
    for at least $1 - \delta \exval(\|O\|_2^2)$ fraction of the circuits, if $\nModes$ is large enough.
\end{theorem}

\begin{proof}
    We aim to prove that $\prob\left(\abs{f_L(U,O) - f_L(U,O) } \leq \epsilon\right) \geq 1-\delta\exval(\|O\|^2)$ and hence that $\prob\left(\abs{f_L(U,O) - \Tr[U \rho U^\dagger O] } > \epsilon\right) \leq \delta\exval(\|O\|^2)$. Let $w_0$ be a value to be fixed later, and let us write $f = f_L(U,O)$ and $f^{(w_0)} = f_L^{(w_0)}(U, O)$. Then, as claimed above, by Markov's identity: 
    \begin{equation}\label{eq:markov-bound-time-complexity-bound}
        p = \prob\left(\abs{f^{(w_0)} - f } > \epsilon\right) \leq \prob\left(\abs{f^{(w_0)} - f } \geq \epsilon\right) \leq \frac{\exval((f - f^{(w_0)})^2)}{\epsilon^2}.
    \end{equation}

    Let $w_0$ be the smallest even integer larger than $2\log_2(2\delta^{-1}\epsilon^{-1})$. Note that we do have $w_0 \leq 2N$ for $N$ large enough since:
    \[w_0 < 2\log_2(2\delta^{-1}\epsilon^{-1}) + 2 \leq 2 \log_2(2^{o(\nModes)}) + 2 = o(N).\]

    Note moreover that this shows in particular that $\frac{e w_0}{N} = o(1)$ and hence $\frac{e w_0}{N} < \frac{1}{2}$ for $N$ large enough. Now, since $w_0 \leq 2N$, we can use the first guarantee of Theorem~\ref{thm:upper-bound-mse-tidied} to find that for $N$ large enough:
    \[\frac{\exval((f - f^{(w_0)})^2)}{\exval(\|O_0\|^2)} = \frac{1}{2^{N-1}} + \left(\frac{e w_0}{N}\right)^{w_0 / 2} \leq 2\left(\frac{e w_0}{N}\right)^{w_0 / 2} \overset{(\dagger)}{\leq} 2  \left(\frac{1}{2}\right)^{2 \log_2(2\delta^{-1} \epsilon^{-1})/2} = \delta \epsilon,\]
    where we use $\frac{e w_0}{N} \leq \frac{1}{2}$ as mentioned above and $w_0 \geq 2\log_2(2\delta^{-1}\epsilon^{-1})$ by definition for the inequality $(\dagger)$. Combining this inequality with Eq.~\eqref{eq:markov-bound-time-complexity-bound} gives exactly the required bound on the success probability. Majorana propagation with any $w_0$ finally runs in time $\mathcal{O}(L \cdot N^{w_0})$, giving our result since $w_0 = \mathcal{O}(\log(\delta^{-1} \epsilon^{-1}))$.
\end{proof}

\subsection{Additional properties of the Markov chain}\label{app:additionalprops}

Here we provide the proofs several lemmas in Appendix~\ref{app:bounderror} that we skipped at the time. We start with the proof of Lemma~\ref{lma:upperbound-tau-equal-t}.

\begin{proof}[Proof of Lemma~\ref{lma:upperbound-tau-equal-t}]
    First, we trivially notice that:
    \[p = \prob(\tau = t) \leq \prob(\tau \geq t).\]

    We will now consider a necessary condition for $\tau \geq t$, i.e.~for the Markov chain not to cross the truncation zone up to time $t' = t - 1$. We consider the Markov chain as a sequence of plus moves ($w \mapsto w+2$), of minus moves ($w \mapsto w-2$) and of stay moves ($w \mapsto w$). 

    We note $S_i \in \{-1, 0, 1\}$ to be a random variable that states the sign of the $i$th move. For instance, if $W_1 - W_0 = 2$, then $S_1 = 1$. We also note $C = \sum_{i=1}^{t'} S_i$ to be the signed distance travelled from the starting point. 

    Note that, if $C \geq w_0/2$, then the truncation zone has been crossed before time $t'$ and hence $\tau < t$. This states by the contrapositive that a necessary condition for $\tau \geq t$ is that $C < w_0/2$. In other words:
    \[\tau \geq t \implies C < w_0/2.\]

    This is already a nice necessary condition, but we want to simplify it further. To do so, we first consider coupling random variables $\chi_i$ sampled IID in $[0, 1[$. This allows us to describe the Markov chain:
    \[(W_i \mid W_{i-1} = w) = \begin{cases}
        w - 2, & \text{if $0 \leq \chi_i < q_-(w)$,} \\
        w + 2, & \text{if $1 - q_+(w) \leq \chi_i < 1$,} \\
        w, &\text{otherwise,}
    \end{cases}\]
    where we write $q_{\pm} = \frac{1}{2} p_{\pm}$. Note that this is the exact same Markov chain as before. The idea here is that we now have access to the random choice $\chi_i$. In particular, this allows us to also express the $S_i$ random variable:
    \[S_i = \begin{cases}
        -1, & \text{if $0 \leq \chi_i < q_-(W_{i-1})$,} \\
        1, & \text{if $1 - q_+(W_{i-1}) \leq \chi_i < 1$,} \\
        0, &\text{otherwise.}
    \end{cases}\]

    We can now consider a simpler model, where each plus moves happens with probability $q_+ = \min_{w \in\{2, 4,\ldots, w_0\}} q_+(w) = q_+(2)$, each minus moves happens with probability $q_- = \max_{w \in\{2, 4,\ldots, w_0\}} q_-(w) = q_-(w_0)$, and where each move is chosen independently. In other word, we define $\hat{S}_i \in \{+, -, =\}$ such that:
    \[\hat{S}_i = \begin{cases}
        -1, & \text{if $0 \leq \chi_i < q_-$,} \\
        1, & \text{if $1 - q_+ \leq \chi_i < 1$,} \\
        0, &\text{otherwise.}
    \end{cases}\]
    
    This is well defined since $q_+ + q_- \leq \frac{1}{2} + \frac{1}{2} = 1$.

    Thanks to the coupling $\chi_i$, we are able to compare the Markov model $S_i$ and the simpler model $\hat{S}_i$. We suppose that $\tau \geq t$, and hence that $W_{i-1} \leq w_0$ for all $i \leq t$, and we want to show that it necessarily implies that $\hat{S}_i \leq S_i$. For it to be wrong, we would need a $i$ such that $(\hat{S}_i, S_i) \in \{(1, 0), (1, -1), (0, -1)\}$. We consider two cases, that cover everything.
    \begin{enumerate}
        \item If $S_i = -1$, we must necessarily have $0 \leq \chi_i < q_-(W_{i-1}) \leq q_-$ (using the fact $W_{i-1} \leq w_0$), which implies that $\hat{S}_i = -1$.
        \item If $\hat{S}_i = 1$, we must necessarily have $\chi_i \geq 1 - q_+ \geq 1 - q_+(W_{i-1})$, which implies that $S_i = 1$. 
    \end{enumerate}

    Hence, when $\tau \geq t$, then $(\hat{S}_i, S_i) \in \{(1, 0), (1, -1), (0, -1)\}$ is impossible, and we must necessarily have $\hat{S}_i \leq S_i$. Leaving $\hat{C} = \sum_{i=1}^{t'} \hat{S}_i$ to be the analogue of $C$ in the simpler model, this necessarily means that $\hat{C} \leq C$ whenever $\tau \geq t$. Hence, we found that:
    \[\tau \geq t \iff \tau \geq t \land C < w_0/2 \implies \tau \geq t \land \hat{C} < w_0/2 \implies \hat{C} < w_0/2.\]
    
    Hence, a necessary condition for $\tau \geq t$ is that $\hat{C}< w_0/2$, giving:
    \[p \leq \prob(\tau \geq t) \leq \prob(\hat{C} < w_0/2).\]

    We are only left with evaluating this probability. To do so, we use a Chernoff bound. We do the whole derivation here for completeness, even though it is likely to be well-known. Let $\lambda > 0$ to be fixed later. We note that, by Markov's inequality:
    \[p \leq \prob(\hat{C} < w_0/2) \leq \prob(\hat{C} \leq w_0/2) = \prob(e^{-\lambda \hat{C}} \geq e^{-\lambda w_0 /2}) \leq \frac{\exval(\exp(-\lambda \hat{C}))}{e^{-\lambda w_0/2}}.\]

    Using the independence of the different move sign in the simpler model, we find:
    \[\begin{split}
        \exval(\exp(-\lambda \hat{C})) &= \left(\exval(\exp(-\lambda \hat{S}_i))\right)^{t'} = \left((1 - q_- - q_+) e^0 + q_- e^{\lambda} + q_+ e^{-\lambda}\right)^{t'}
        \\& = \left(1 - q_+\cdot(1 - e^{-\lambda}) + q_-\cdot(e^{\lambda} - 1)\right)^{t'}.
    \end{split}\]

    Now, we know that $q_- = \Theta(1/\nModes^3)$. Moreover, for $\nModes \geq 7/2$:
    \[q_+ = q_+(2) = \frac{1}{2} \cdot\frac{\binom{4}{1} \binom{2\nModes-4}{2-1}}{\binom{2\nModes}{2}} = \frac{4(2\nModes-4)}{2\nModes(2\nModes-1)} \geq \frac{4(\nModes - 1/2)}{2\nModes\cdot 2 (\nModes-1/2)} = \frac{1}{\nModes}.\]

    Overall, this gives us that, assuming $\lambda = \Theta(1)$ (which we will indeed fullfil later):
    \[p \leq e^{\lambda w_0/2}\left(1 - q_+ \cdot (1 - e^{-\lambda}) + q_-\cdot(e^{\lambda} - 1)\right)^{t'} \leq e^{\lambda w_0/2}\left(1 - \frac{1 - e^{-\lambda}}{\nModes} + O(1/\nModes^3)\right)^{t'}.\]

    Leaving $\lambda = \ln(8) = 3\ln(2) = \Theta(1)$, this reads, for $\nModes$ sufficiently large:
    \[p \leq \left(2^{3/2}\right)^{w_0} \left(1 - \frac{7}{8\nModes} + O(1/\nModes^3)\right)^{t'} \leq 4^{w_0} \left(1 - \frac{7}{8\nModes} + \frac{1}{8\nModes}\right)^{t'} \leq 4^{w_0} \left(1 - \frac{3}{4\nModes}\right)^{t'}.\]
\end{proof}

Next we proceed to proving Lemma~\ref{lma:SummaryPropertiesOfU}. The following lemmas correspond either to parts of the proof idea stated after Lemma~\ref{lma:SummaryPropertiesOfU} or are intermediary lemmas to prove these.

\begin{lemma}\label{lma:recurrence-relation-u}
    We let $q_{\pm}(w) = p_\pm(w)/2$ for the sake of notation. Then:
    \[U_{i+1}(w) = q_+(w) U_i(w+2) + q_-(w) U_i(w-2) + (1 - q_+(w) - q_-(w)) U_i(w).\]
\end{lemma}
\begin{proof}
    We condition on the first step:
    \[\begin{split}
        U_{i+1}(w) & = \exval(\probContribute{W_{i+1}} \mid W_0=w)
        \\& = \sum_u \exval(\probContribute{W_{i+1}} \mid W_1=u, W_0=w) \prob(W_1 = u \mid W_0 = w)
        \\& = \sum_u \exval(\probContribute{W_{i+1}} \mid W_1=u) \prob(W_1 = u \mid W_0 = w)
        \\& = \sum_u \exval(\probContribute{W_{i}} \mid W_0=u) \prob(W_1 = u \mid W_0 = w)
        \\& = \sum_u U_i(u) \prob(W_1 = u \mid W_0 = w)
        \\&  = q_+(w) U_i(w+2) + q_-(w) U_i(w-2) + (1 - q_+(w) - q_-(w)) U_i(w).
    \end{split}\]
\end{proof}

\begin{lemma}\label{lma:u-invariant}
    We let $\mathcal{W} = \{2, 4, \ldots, 2\nModes-2\}$ to be the set of interior points of the simulation. We moreover define:
    \[\overline{\pi}(w) = \binom{2\nModes}{w}, \mathspace \pi(w) = \frac{\overline{\pi}(w)}{\sum_{w \in \mathcal{W}} \overline{\pi}(w)}.\]

    When $\nModes \geq 2$, then we have the following invariant for all $i \geq 0$:
    \[\sum_{w \in \mathcal{W}} \pi(w) U_i(w) = \frac{1}{2^{\nModes-1} + 1}.\]

    In other words, if $U_i(w)$ does flatten out to a line when $i$ increases, then this line must be exponentially low.
\end{lemma}
\begin{proof}
    We do this proof by induction on $i\geq 0$.
    \begin{itemize}
        \item \textit{(Initial case)} We first compute:
        \[\sum_{w \in \mathcal{W}} \overline{\pi}(w) U_0(w) = \sum_{\substack{w=2 \\ \text{$w$ even}}}^{2\nModes -2} \binom{2\nModes}{w} \frac{\binom{\nModes}{w/2}}{\binom{2\nModes}{w}} = \sum_{\substack{w=2 \\ \text{$w$ even}}}^{2\nModes -2} \binom{\nModes}{w/2},\]
        where we used the fact $U_0(w) = \probContribute{w}$. Now, doing the change of variable $u = w/2$ and using Newton's binomial theorem:
        \[\sum_{w \in \mathcal{W}} \overline{\pi}(w) U_0(w) = \sum_{u=1}^{\nModes-1} \binom{\nModes}{u} = \sum_{u=0}^{\nModes} \binom{\nModes}{u} - 2 = (1 + 1)^\nModes - 2 = 2^\nModes - 2.\]

        We are only left with computing the normalisation $\sum_{w \in \mathcal{W}} \overline{\pi}(w)$. We notice that, using again Newton's binomial theorem:
        \[\sum_{w=1}^{2\nModes - 1} \binom{2\nModes}{w} = \sum_{w=0}^{2\nModes} \binom{2\nModes}{w} - 2 = (1 + 1)^{2\nModes} - 2 = 2^{2\nModes} - 2,\]
        \[\sum_{w=1}^{2\nModes - 1} \binom{2\nModes}{w} (-1)^w = \sum_{w=0}^{2\nModes} \binom{2\nModes}{w} (-1)^w - 2 = (1 - 1)^{2\nModes} - 2 = -2.\]

        We can average the two to only keep even terms in the sum:
        \[\sum_{\substack{w=2 \\ \text{$w$ even}}}^{2\nModes -2} \overline{\pi}(w) = \sum_{w=1}^{2\nModes-1} \frac{1 + (-1)^w}{2} \binom{2\nModes}{w} = \frac{(2^{2\nModes} - 2 ) + (-2)}{2} = 2^{2\nModes-1} - 2.\]

        This gives exactly our result, when $\nModes \geq 2$ (for $\mathcal{W} \neq \emptyset$):
        \[\sum_{w \in \mathcal{W}} \pi(w) U_i(w) = \frac{\sum_{w \in \mathcal{W}} \overline{\pi}(w) U_i(w)}{\sum_{u \in \mathcal{W}} \overline{\pi}(u)} = \frac{2^\nModes - 2}{2^{2\nModes - 1} - 2} = \frac{2^\nModes - 2}{(2^\nModes - 2)(2^{\nModes-1}+1)} = \frac{1}{2^{\nModes-1} + 1}.\]

        \item \textit{(Inductive case)} By lemma~\ref{lma:recurrence-relation-u}:
        \[\begin{split}
            &\ \sum_{w \in \mathcal{W}} \pi(w) U_{i+1}(w)
            \\=&\ \sum_{w \in \mathcal{W}} \pi(w) \left(q_+(w) U_i(w+2) + q_-(w) U_i(w-2) + (1 - q_+(w) - q_-(w)) U_i(w)\right)
            \\=&\ \sum_{\substack{w = 4 \\ \text{$w$ even}}}^{2\nModes} \pi(w-2) q_+(w-2) U_i(w) + \sum_{\substack{w = 0 \\ \text{$w$ even}}}^{2\nModes -4} \pi(w+2) q_-(w+2) U_i(w)
            \\&\ + \sum_{\substack{w = 2 \\ \text{$w$ even}}}^{2\nModes -2} \pi(w)(1 - q_+(w) - q_-(w)) U_i(w).
        \end{split}\]

        Now, we know that $q_+(0) = q_+(2\nModes-2) = 0$ and $q_-(2) = q_-(2\nModes) = 0$. Hence, this simplifies to:
        \[\begin{split}
            &\ \sum_{w \in \mathcal{W}} \pi(w) U_{i+1}(w)
            \\=&\ \sum_{\substack{w = 2 \\ \text{$w$ even}}}^{2\nModes -2} U_i(w)\left[\pi(w-2)q_+(w-2) + \pi(w+2)q_-(w+2) + \pi(w) (1 - q_+(w) - q_-(w)\right].
        \end{split}\]

        Hence, the result holds if and only if:
        \[\begin{split}
            &\ \pi(w-2)q_+(w-2) + \pi(w+2)q_-(w+2) + \pi(w) (1 - q_+(w) - q_-(w) = \pi(w)
            \\ \iff&\ \pi(w-2)q_+(w-2) + \pi(w+2)q_-(w+2) = \pi(w) q_+(w) + \pi(w) q_-(w).
        \end{split}\]

        We recognise the detailed balance equations, justifying the symbol $\pi$ as it is commonly used for the stationary distribution of a Markov chain. The normalisation cancels on both side, so we only need to check it for $\overline{\pi}$. The left hand side reads:
        \[\begin{split}
            \overline{\pi}(w-2)q_+(w-2) + \overline{\pi}(w+2)q_-(w+2) &= \binom{2\nModes}{w-2} \frac{\binom{4}{1}\binom{2\nModes-4}{(w-2)-1}}{2\binom{2\nModes}{w-2}} + \binom{2\nModes}{w+2} \frac{\binom{4}{3} \binom{2\nModes-4}{(w+2)-3}}{2\binom{2\nModes}{w+2}} 
            \\& = 2 \binom{2\nModes-4}{w-3} + 2\binom{2\nModes-4}{w-1}.
        \end{split}\]

        Similarly, the right hand side reads:
        \[\begin{split}
            \overline{\pi}(w)q_+(w) + \overline{\pi}(w)q_-(w) &= \binom{2\nModes}{w} \frac{\binom{4}{1}\binom{2\nModes-4}{w-1}}{2\binom{2\nModes}{w}} + \binom{2\nModes}{w} \frac{\binom{4}{3} \binom{2\nModes-4}{w-3}}{2\binom{2\nModes}{w+2}} 
            \\& = 2 \binom{2\nModes-4}{w-1} + 2\binom{2\nModes-4}{w-3}.
        \end{split}\]

        We thus see that the left hand-side is equal to the right hand-side, giving that the detailed balance equations do hold. Hence, our result simplifies exactly as required:
        \[\sum_{w \in \mathcal{W}} \pi(w) U_{i+1}(w) = \sum_{w \in \mathcal{W}} \pi(w) U_i(w) = \frac{1}{2^{\nModes-1} + 1}.\]
    \end{itemize}
\end{proof}

\begin{lemma}\label{lma:u-symmetry}
    For all $w \in \mathcal{W}$ and for all $i \geq 0$, we have the following symmetry:
    \[U_i(w) = U_i(2\nModes - w).\]

    In other words, for a given $i$, the curve $U_i(w)$ is symmetrical around $\nModes$.
\end{lemma}
\begin{proof}
    We make this proof by induction.
    \begin{itemize}
        \item \textit{(Initial case)} We have:
        \[U_0(2\nModes - w) = \exval(\probContribute{W_0} \mid W_0 = 2\nModes - w) = \probContribute{2\nModes - w} = \frac{\binom{\nModes}{(2\nModes - w)/2}}{\binom{2\nModes}{2\nModes-w}} = \frac{\binom{\nModes}{w/2}}{\binom{2\nModes}{w}} = \probContribute{w} = U_0(w).\]
        \item \textit{(Inductive case)} Note that:
        \[p_+(2\nModes - w) = \frac{\binom{4}{1} \binom{2\nModes-4}{(2\nModes-w)-1}}{\binom{2\nModes}{2\nModes-w}} = \frac{\binom{4}{3} \binom{2\nModes-4}{2\nModes - 4 - (2\nModes-w)+1}}{\binom{2\nModes}{2\nModes-w}} = \frac{\binom{4}{3} \binom{2\nModes-4}{w-3}}{\binom{2\nModes}{2\nModes-w}} = p_-(w),\]
        \[p_-(2\nModes - w) = \frac{\binom{4}{3} \binom{2\nModes-4}{(2\nModes-w)-3}}{\binom{2\nModes}{2\nModes-w}} = \frac{\binom{4}{1} \binom{2\nModes-4}{2\nModes - 4 - (2\nModes-w)+3}}{\binom{2\nModes}{2\nModes-w}} = \frac{\binom{4}{1} \binom{2\nModes-4}{w-1}}{\binom{2\nModes}{2\nModes-w}} = p_+(w).\]

        But then, $q_+(2\nModes - w) = q_-(w)$ and $q_-(2\nModes - w) = q_+(w)$.  This gives us, by lemma~\ref{lma:recurrence-relation-u}:
        \[\begin{split}
            U_{i+1}(2\nModes-w) =&\ q_+(2\nModes - w) U_i(2\nModes - w+2) + q_-(2\nModes -w) U_i(2\nModes - w-2) 
            \\&\ + (1 - q_+(2\nModes - w) - q_-(2\nModes - w)) U_i(2\nModes - w)
            \\=&\ q_-(w) U_i(w-2) + q_+(w) U_i(w+2) + (1 - q_-(w) - q_+(w)) U_i(w)
            \\=&\ U_{i+1}(w).
        \end{split}\]

        This finishes the proof.
    \end{itemize}
\end{proof}

\begin{definition}[Forward difference]
    We define the forward difference of $U_i(w)$ to be, for all $i \geq 0$ and $w \in \{2, 4, \ldots, 2\nModes-4\}$:
    \[\Delta U_i(w) = U_i(w+2) - U_i(w).\]

    Intuitively, this is the discrete equivalent of the first derivative.
\end{definition}

\begin{lemma}\label{lma:recurrence-relation-delta-u}
    For all $i \geq 0$ and $w$, leaving again $q_\pm(w) = p_\pm(w)/2$ for the sake of notation:
    \[\Delta U_{i+1}(w) = q_+(w+2) \Delta U_i(w+2) + q_-(w) \Delta U_i(w-2) + (1 - q_-(w+2) - q_+(w)) \Delta U_i(w).\]
\end{lemma}
\begin{proof}
    Using lemma~\ref{lma:recurrence-relation-u}, we know that:
    \[\begin{split}
        U_{i+1}(w) & = q_+(w) U_i(w+2) + q_-(w) U_i(w-2) + \left(1 - q_+(w) - q_-(w)\right) U_i(w)
        \\& = q_+(w) (U_i(w+2) - U_i(w)) + q_-(w) (U_i(w-2) - U_i(w)) + U_i(w)
        \\& = q_+(w) \Delta U_i(w) - q_-(w) \Delta U_i(w-2) + U_i(w).
    \end{split}\]

    Hence:
    \[\begin{split}
        \Delta U_{i+1}(w) = &\ U_{i+1}(w+2) - U_{i+1}(w) 
        \\ = &\ q_+(w + 2) \Delta U_i(w + 2) - q_-(w + 2) \Delta U_i(w) + U_i(w + 2) 
        \\ & - q_+(w) \Delta U_i(w) + q_-(w) \Delta U_i(w-2) - U_i(w)
        \\ =&\ q_+(w + 2) \Delta U_i(w + 2) + q_-(w) \Delta U_i(w-2) + (1 - q_-(w+2) - q_+(w)) \Delta U_i(w).
    \end{split}\]
\end{proof}

\begin{lemma}\label{lma:delta-u-non-positive}
    We let $\mathcal{W}_- = \{2 \leq w \leq \nModes-1 \mid \text{$w$ is even}\}$ to be the first half of the interior points. Then, for all $i \geq 0$ and $w \in \mathcal{W}_-$:
    \[\Delta U_i(w) \leq 0.\]

    Note that this implies that $\Delta U_i(w) \geq 0$ when $w \geq \nModes - 1$ by symmetry around $\nModes$ (lemma~\ref{lma:u-symmetry}).
\end{lemma}
\begin{proof}
    We do this proof by induction on $i \geq 0$.
    \begin{itemize}
        \item \textit{(Initial case)} Let $w \in \mathcal{W}_-$ be arbitrary. Then: 
        \[\Delta U_0(w) = U_0(w+2) - U_0(w) = \probContribute{w+2} - \probContribute{w} = \probContribute{w+2} \left(1-  \frac{\probContribute{w}}{\probContribute{w+2}}\right).\]

        Note that $\probContribute{w+2}$ is always non-negative. Moreover, we know that:
        \[\frac{\binom{\nModes}{k}}{\binom{\nModes}{k+1}} = \frac{\frac{\nModes!}{k!(\nModes-k)!}}{\frac{\nModes!}{(k+1)!(\nModes-k-1)!}} = \frac{(k+1)!}{k!} \cdot \frac{(\nModes-k-1)!}{(\nModes-k)!} = \frac{k+1}{\nModes-k},\]
        \[\frac{\binom{\nModes}{k + 2}}{\binom{\nModes}{k}} = \frac{\frac{\nModes!}{(k+2)!(\nModes-k-2)!}}{\frac{\nModes!}{k!(\nModes-k)!}} = \frac{k!}{(k+2)!} \cdot \frac{(\nModes-k)!}{(\nModes-k -2)!} = \frac{(\nModes-k)(\nModes-k-1)}{k(k-1)}.\]

        Hence, leaving $v = w/2$ for clarity:
        \[\begin{split}
            1-  \frac{\probContribute{w}}{\probContribute{w+2}} & = 1 - \frac{\binom{\nModes}{w/2} / \binom{2\nModes}{w}}{\binom{\nModes}{w/2 + 1} / \binom{2\nModes}{w + 2}} = 1 - \frac{\binom{\nModes}{v}}{\binom{\nModes}{v+1}}\cdot \frac{\binom{2\nModes}{2v+2}}{\binom{2\nModes}{2v}}
            = 1- \frac{v + 1}{\nModes - v} \cdot \frac{(2\nModes - 2v)(2\nModes - 2v -1)}{2v(2v-1)}
            \\& = \frac{2v(\nModes-v)(2v-1) - 2(v+1)(\nModes-v)(2\nModes-2v-1)}{2v(\nModes-v)(2v-1)}.
        \end{split}\]

        Notice that the denominator is always non-negative since $2 \leq w \leq 2\nModes -2 $ and hence $1 \leq v \leq \nModes - 1$. Moreover, the numerator is given by:
        \[\begin{split}
            2(\nModes-v)(v(2v-1) - (v+1)(2\nModes-2v-1)) &= 2(\nModes-v)(2v^2 - v - 2nv + 2v^2 + v - 2\nModes + 2v + 1)
            \\& = 2(\nModes-v)(4v^2 +v(-2\nModes+2) -2\nModes+1).
        \end{split}\]

        The term $\nModes-v$ is again always non-negative. We thus want to show that the quadratic polynomial is non-positive. To do so, we will show that $0 \leq w \leq \nModes$ (i.e.~$0 \leq v \leq \nModes/2$) implies necessarily that $v_- \leq v \leq v_+$, where $v_\pm$ are the two roots of this polynomial. Since it is a convex parabola, this will necessarily imply that it is indeed non-positive.

        Indeed, we do find:
        \[\begin{split}
            v_- & = \frac{-(-2\nModes+2) - \sqrt{(-2\nModes+2)^2 - 4(4)(-2\nModes+1)}}{2(4)} = \frac{2\nModes - 2 - 2\sqrt{\nModes^2 + 6\nModes - 3}}{8}
            \\& \leq \frac{2\nModes - 2 - 2\sqrt{(\nModes-1)^2}}{8} = 0 \leq v,
        \end{split}\]
        where we used the fact that $\nModes^2 + 6\nModes - 3 \geq (\nModes-1)^2 \iff \nModes \geq 1/2$, which holds since $\nModes \geq 1$. Completely similarly:
        \[\begin{split}
            v_+ & = \frac{-(-2\nModes+2) + \sqrt{(-2\nModes+2)^2 - 4(4)(-2\nModes+1)}}{2(4)} = \frac{2\nModes - 2 + 2\sqrt{\nModes^2 + 6\nModes - 3}}{8}
            \\& \geq \frac{2\nModes - 2 + 2\sqrt{(\nModes+1)^2}}{8} = \frac{\nModes}{2} \geq v,
        \end{split}\]
        where we used that $\nModes^2 + 6\nModes - 3 \geq (\nModes+1)^2 \iff \nModes \geq 1$, which again holds.

        Overall, we did find that $v_- \leq v \leq v_+$, which allows to conclude that the parabola is non-positive as mentioned before. Since all other terms are non-negative, this means that, overall, $\probContribute{w+2} - \probContribute{w}$ is non-positive, finishing the initial case.

        \item \textit{(Inductive case)} By lemma~\ref{lma:recurrence-relation-delta-u}:
        \begin{equation}\label{eq:recurrence-delta-u}
            \Delta U_{i+1}(w) = q_+(w+2) \Delta U_i(w+2) + q_-(w) \Delta U_i(w-2) + (1 - q_-(w+2) - q_+(w)) \Delta U_i(w).
        \end{equation}

        Now, $q_+(w+2) \geq 0$ and $q_-(w) \geq 0$. Similarly, $q_+(w) = p_+(w)/2 \leq 1/2$ and $q_-(w+2) \leq 1/2$, telling us $1 - q_-(w+2) - q_+(w)\geq 0$. Hence, all these coefficients are non-negative. This trivially gives our result thanks to the inductive hypothesis for $w \leq \nModes-3$. The only thing left to consider is the boundary $w = \nModes-2$ or $w = \nModes-1$ depending on the parity of $\nModes$. We consider each case separately.
        \begin{itemize}
            \item Suppose that $\nModes$ is odd, and hence $w = \nModes-1$. We directly get our result by the symmetry of $U_i(w)$ (lemma~\ref{lma:u-symmetry}):
            \begin{equation}\label{eq:delta-u-is-zero-at-middle}
                \Delta U_{i+1}(\nModes-1) = U_{i+1}(\nModes+1) - U_{i+1}(\nModes-1) = U_{i+1}(\nModes-1) - U_{i+1}(\nModes-1) = 0.
            \end{equation}

            \item Suppose now that $\nModes$ is even, and hence $w = \nModes-2$. Still by the symmetry (lemma~\ref{lma:u-symmetry}), we find:
            \begin{equation}\label{eq:delta-u-is-odd-at-middle}
                \Delta U_i(\nModes) = U_i(\nModes+2) - U_i(\nModes) = U_i(\nModes-2) - U_i(\nModes) = -\Delta U_i(\nModes-2).
            \end{equation}

            Hence, equation~\ref{eq:recurrence-delta-u} becomes:
            \[\begin{split}
                & \Delta U_{i+1}(\nModes-2)
                \\ = &\ q_+(w+2) \Delta U_i(\nModes) + q_-(w) \Delta U_i(\nModes-4) + (1 - q_-(w+2) - q_+(w)) \Delta U_i(\nModes-2) 
                \\ = &\ q_-(w) \Delta U_i(\nModes-4) + (1 - q_-(w+2) - q_+(w) - q_+(w+2))\Delta U_i(\nModes-2).
            \end{split}\]

            Hence, the exact same argument we used for $w \leq \nModes-3$ applies if we prove that $1 - q_-(w+2) - q_+(w) - q_+(w+2) \geq 0$. Now, we know that $p_+(w) + p_-(w) + p_=(w) = 1$ since it is a probability distribution and hence, in particular, $p_+(w) + p_-(w) \leq 1$. Overall, this gives us that:
            \[q_+(w+2) + q_-(w+2) \leq 1/2, \mathspace q_+(w) = p_+(w)/2 \leq 1/2.\]

            Adding the two, we get exactly that $q_-(w+2) + q_+(w) + q_+(w+2) \leq 1$, finishing the proof.
        \end{itemize}
    \end{itemize}
\end{proof}

\begin{lemma}\label{lma:delta-u-increases}
    If $\nModes \geq 4$, then, for all $i \geq 0$ and $w \in \mathcal{W}_-$:
    \[\Delta U_{i+1}(w) \geq \left(1-  \frac{1}{2\nModes}\right)\Delta U_i(w).\]
    A more intuitive formulation is that, thanks to lemma~\ref{lma:delta-u-non-positive}, this states that $|\Delta U_{i+1}(w)| \leq (1-  1/2\nModes) |\Delta U_i(w)|$. Combining it with lemma~\ref{lma:u-symmetry} means that $|\Delta U_{i+1}(w)| \leq (1-  1/2\nModes) |\Delta U_i(w)|$ in fact holds for all $w$. Hence, the curve $U_i(w)$ does flatten out as $i$ increases, at an exponential speed.
\end{lemma}
\begin{proof}
    This proof follows a similar idea to the one of lemma~\ref{lma:delta-u-non-positive}. We also do this proof by induction on $i \geq 0$. 
    \begin{itemize}
        \item \textit{(Initial case)} We will extensively use Sympy. We aim to show $\Delta U_{i+1}(w) \geq a\Delta U_i(w)$ for some $a > 0$ to be fixed later. By lemma~\ref{lma:recurrence-relation-delta-u}, we have:
        \[\Delta U_1(w) - a\Delta U_0(w) = q_+(w+2) \Delta U_0(w+2) + q_-(w) \Delta U_0(w-2) + (1 -a- q_-(w+2) - q_+(w))\Delta U_0(w).\]

        Now, we know that $\Delta U_0(w) = U_0(w+2) - U_0(w) = \probContribute{w+2} - \probContribute{w}$, using Sympy (or Mathematica), and leaving $w = 2u$, we find that this is equal to:
        \[\Delta U_1(w) - a\Delta U_0(w) =  \frac{-8 \left(- \nModes + 2 u + 1\right) P_\nModes(u) \left(\nModes - 2\right)!}{\left(2 \nModes - 3\right) \left(2 \nModes - 1\right) \left(- 2 \nModes + 2 u + 1\right) {\binom{2 \nModes}{2 u}} u! \left(\nModes - u\right)!}\]
        where:
        \[\begin{split}
            P_\nModes(u) =&\ - a \nModes^{4} + 3 a \nModes^{3} - \frac{11 a \nModes^{2}}{4} + \frac{3 a \nModes}{4} + \nModes^{4} - 2 \nModes^{3} u - 3 \nModes^{3} + 10 \nModes^{2} u^{2} + 4 \nModes^{2} u + \frac{15 \nModes^{2}}{4}
            \\& - 16 \nModes u^{3} - 18 \nModes u^{2} - 22 \nModes u - \frac{23 \nModes}{4} + 8 u^{4} + 16 u^{3} + 28 u^{2} + 20 u + 6.
        \end{split}\]

        We notice that, for $2 \leq w \leq \nModes-1 \iff 1 \leq u \leq \frac{\nModes-1}{2}$, the terms $(-8)$, $(-\nModes+2u+1)$ and $(-2\nModes+2u+1)$ are all non-positive, whereas all other terms (except for $P_\nModes(u)$) are non-negative. Hence, to show $\Delta U_1(w) - a\Delta U_0(w)$ is non-negative, we have to show $P_\nModes(u)$ is non-positive. 

        Using again Sympy, we find the following symmetry:
        \[P_\nModes(\nModes-1-u) = P_\nModes(u).\]

        Differentiating both sides according to $u$, this gives us that:
        \begin{equation}\label{eq:dpdu-is-odd}
            -P_\nModes'(\nModes-1-u) = P_\nModes'(u).
        \end{equation}

        In particular, for $u = \frac{\nModes-1}{2}$:
        \[-P_\nModes'\left(\frac{\nModes-1}{2}\right) = P_\nModes'\left(\frac{\nModes-1}{2}\right) \iff P_\nModes'\left(\frac{\nModes-1}{2}\right) = 0.\]

        Moreover, equation~\ref{eq:dpdu-is-odd} also tells us that there are as many roots of $P_\nModes'(u)$ in $[1, \frac{\nModes-1}{2}[$ as there are in $]\frac{\nModes-1}{2}, \nModes-2]$ (in other words, if there are five points in the first interval where the derivative is zero, then there must also be exactly five points where the derivative is zero in the second interval). However, we notice that $P_\nModes(u)$ is a quartic polynomial in $u$, and hence $P_\nModes'(u)$ is a cubic polynomial. By the fundamental theorem of algebra, this means that $P_\nModes'(u)$ has at most three root. Since we found there is one at $u = \frac{\nModes-1}{2}$, symmetry must mean that there is at most one in $[1, \frac{\nModes-1}{2}[$ (if there were two in this interval, there would also be two in $]\frac{\nModes-1}{2}, \nModes-2]$, which would be too many in total). This means that, on this interval, the derivative can change sign at most once by the intermediate value theorem. Moreover, again with Sympy, we find that:
        \[P_\nModes'(1) = - 2 \left(\nModes - 3\right) \left(\nModes^{2} - 9 \nModes + 26\right).\]

        The quadratic polynomial in $\nModes$ has a negative discriminant: $\Delta = (-9)^2 - 4\cdot1\cdot 26 < 0$. Hence, it has no real root. This tells us that, for $\nModes > 3$, we have $P_\nModes'(1) < 0$.

        So far, we found that $P_\nModes'(u)$ starts negative, and then changes sign at most once over $[1, \frac{\nModes-1}{2}[$. The supremum of $P_\nModes(u)$ over this interval is thus one of its endpoint, meaning that:
        \[\sup_{u\in [1, \frac{\nModes-1}{2}[} P_\nModes(u) = \max\left\{P_\nModes(1), P_\nModes\left(\frac{\nModes-1}{2}\right)\right\}.\]

        Using Sympy again, and taking $a = 1 - \frac{1}{2\nModes}$
        , this simplifies to:
        \[\sup_{u\in [1, \frac{\nModes-1}{2}[} P_\nModes(u) = \max\left\{- \frac{3 \left(4 \nModes^{3} - 36 \nModes^{2} + 159 \nModes - 207\right)}{8}, - 4 \nModes^{2} + \frac{39 \nModes}{8} + \frac{9}{8}\right\}.\]

        Numerically, we find that the first polynomial has a single real root, which is at $\nModes_0 < 3$. Similarly, we find all the roots $\nModes_0$ of the second polynomial are such that $\nModes_0 <2$. Since both polynomial have a negative dominant term we find that, for $\nModes \geq 3$:
        \[\sup_{u\in [1, \frac{\nModes-1}{2}[} P_\nModes(u) < 0.\]

        Overall, this means that $P_\nModes(u) < 0$ for all $u\in [1, \frac{\nModes-1}{2}[$, and hence $\Delta U_1(w) - a\Delta U_0(w) \geq 0$ for all even $1 \leq w \leq \nModes-1$. This finishes the initial case.
        
        \item \textit{(Inductive case)} Let $D_i(w) = \Delta U_{i+1}(w) - a\Delta U_i(w)$ for $a > 0$ fixed by the initial case. We know that $D_{i-1}(w) \geq 0$, and we aim to show $D_i(w) \geq 0$. By lemma~\ref{lma:recurrence-relation-delta-u}, we find:
        \[\begin{split}
            D_i(w) = &\ \Delta U_{i+1}(w) - a\Delta U_i(w)
            \\= &\ q_+(w+2) \Delta U_{i}(w+2) + q_-(w) \Delta U_{i}(w-2) + (1 - q_-(w+2) - q_+(w)) \Delta U_{i}(w)
            \\&\ - q_+(w+2) a\Delta U_{i-1}(w+2) - q_-(w) a\Delta U_{i-1}(w-2) - (1 - q_-(w+2) - q_+(w)) a\Delta U_{i-1}(w)
            \\=&\ q_+(w+2) D_{i-1}(w+2) + q_-(w) D_{i-1}(w-2) + (1 - q_-(w+2) - q_+(w)) D_{i-1}(w)
        \end{split}\]

        We get our result when $w \leq \nModes - 3$ by the exact same argument as the one of the inductive case of lemma~\ref{lma:delta-u-non-positive}, since this is the exact same recurrence relation. We then again need to consider the boundary.
        \begin{itemize}
            \item Suppose that $\nModes$ is odd, and hence that we are left with $w = \nModes-1$. Then, using equation~\ref{eq:delta-u-is-zero-at-middle} twice:
            \[D_i(\nModes-1) = \Delta U_{i+1}(\nModes-1) - a\Delta U_i(\nModes-1) = 0-0 = 0.\]
            \item Suppose now that $\nModes$ is even, and hence that we are left with $w = \nModes-2$. We notice that, using equation~\ref{eq:delta-u-is-odd-at-middle} twice:
            \[D_{i-1}(\nModes) = \Delta U_i(\nModes) - a\Delta U_{i-1}(\nModes) = -\Delta U_i(\nModes-2) + a\Delta U_{i-1}(\nModes-2) = -D_{i-1}(\nModes-2).\]

            We can thus use the exact same argument as in lemma~\ref{lma:delta-u-non-positive}.
        \end{itemize}
    \end{itemize}
\end{proof}

\begin{lemma}\label{lma:upperbound-u}
    For any $i \geq 0$ and $w \in \mathcal{W}_-$:
    \[U_i(w) \leq \left(1 - \frac{1}{2\nModes}\right)^i \probContribute{w} + \frac{1}{2^{\nModes-1} + 1}.\]

    Note that this bound also holds for any $w$ by symmetry (lemma~\ref{lma:u-symmetry}), proving that lemma~\ref{lma:SummaryPropertiesOfU} holds. Moreover, in particular, if $w = \Theta(1)$:
    \[U_i(w) \leq \left(\frac{ew}{2\nModes}\right)^{w/2} + \frac{1}{2^{\nModes-1}}= O\left(\nModes^{-w/2}\right).\]
\end{lemma}
\begin{proof}
    Let $w \in \mathcal{W}_-$ be arbitrary. We moreover let $\middlePoint \in \{\nModes-1, \nModes\}$ such that $\middlePoint$ is even. Intuitively, $\middlePoint$ is the middle point. Then, we notice that, using a telescoping series, we get the discrete equivalent to the FTC 2 (fundamental theorem of calculus part two):
    \[\begin{split}
        U_i(\middlePoint) - U_i(w) & = U_i(\middlePoint) - U_i(\middlePoint-2) + U_i(\middlePoint-2) - U_i(\middlePoint-4) + \ldots + U_i(w+2) - U_i(w)
        \\& = \Delta U_i(\middlePoint-2) + \Delta U_i(\middlePoint-4) + \ldots + \Delta U_i(w).
    \end{split}\]

    Hence, using lemma~\ref{lma:delta-u-increases}: 
    \[\begin{split}
        U_i(w) & = U_i(\middlePoint) - (\Delta U_i(\middlePoint-2) + \ldots + \Delta U_i(w)) 
        \\& \leq U_i(\middlePoint) - \left(1 - \frac{1}{2\nModes}\right)^i(\Delta U_0(\middlePoint-2) + \ldots + \Delta U_0(w))
        \\& = U_i(\middlePoint) + \left(1 - \frac{1}{2\nModes}\right)^i\left[-U_0(\middlePoint) + U_0(\middlePoint) - (\Delta U_0(\middlePoint-2) + \ldots + \Delta U_0(w))\right]
        \\& = U_i(\middlePoint) + \left(1 - \frac{1}{2\nModes}\right)^i \left[-U_0(\middlePoint) + U_0(w)\right]
        \\& = U_i(\middlePoint) - \left(1 - \frac{1}{2\nModes}\right)^i \probContribute{\middlePoint} + \left(1 - \frac{1}{2\nModes}\right)^i\probContribute{w}.
    \end{split}\]

    In other words, since the derivative keep increasing (while being non-positive), their signed distance $U_i(w) - U_i(\middlePoint)$ keeps decreasing. Hence, $U_i(w)$ is upper-bounded by the position of $U_i(\middlePoint)$ plus the signed distance between $U_0(\middlePoint)$ and $U_0(w)$ value shrunk by the flattening of the curve. 

    Now, by lemma~\ref{lma:delta-u-non-positive}, we know that, for all $\hat{w} \in \mathcal{W}_-$:
    \[U_i(\hat{w}) = U_i(\middlePoint) - (\Delta U_i(\middlePoint-2) + \ldots + \Delta U_i(\hat{w})) \geq U_i(\middlePoint).\]

    By symmetry (lemma~\ref{lma:u-symmetry}), this means that $U_i(\hat{w}) \geq U_i(\middlePoint)$ for all $\hat{w} \in \mathcal{W}$. In other words, $U_i(\middlePoint)$ is a minimum. Now, by lemma~\ref{lma:u-invariant}, we know that there exists some $\pi(\hat{w})$ such that 
    \[\sum_{\hat{w} \in \mathcal{W}} \pi(\hat{w}) = 1, \mathspace \sum_{\hat{w} \in \mathcal{W}} \pi(\hat{w}) U_i(\hat{w}) = \frac{1}{2^{\nModes-1} + 1}.\]

    The minimum cannot exceed the value of a weighted average (otherwise, all values exceed this weighted average and hence we get a direct contradiction). Hence, we must necessarily have:
    \[U_i(\middlePoint) \leq \frac{1}{2^{\nModes-1} + 1}.\]

    This finishes the proof:
    \[\begin{split}
        U_i(w) &\leq U_i(\middlePoint) - \left(1 - \frac{1}{2\nModes}\right)^i \probContribute{\middlePoint} + \left(1 - \frac{1}{2\nModes}\right)^i \probContribute{w} \leq \frac{1}{2^{\nModes-1} + 1} + \left(1 - \frac{1}{2\nModes}\right)^i \probContribute{w}.
    \end{split}\]
\end{proof}

\section{Numerical simulations}\label{app:numerics}

In this section, we provide the details about the numerical simulations presented in the main section. 

The simulations presented in Table \ref{tab:time-benchmarks} were run on the Vega supercomputer. MP and Qiskit simulations were run on a single CPU node (dual-socket AMD Rome 7H12, 256GB RAM) using 8 CPU threads.  CUDA-Q simulations were run on a single GPU (node architecture: dual-socket AMD Rome 7H12, 512 GB RAM, 4x Nvidia A100, 40GB HBM2). The CUDA-Q MPS simulations were carried out with default parameters and a maximum bond dimension $\chi=64$.

Each simulation is repeated 3 times and the results presented in the table are the average of the 3 runs. For MP, the circuits consist of exponentiations of Majorana monomials.  For all the other, the fermionic ansatz was transformed via Jordan-Wigner F2Q mapping, which transforms Majorana monomials exponentiations into Pauli exponentiations. For statevector simulators, the provided circuits were not transpiled. For CUDA-Q MPS, the transpiled input circuits yielded better timings than the untranspiled ones, and were consequently used for the table. The transpilation is performed for all-to-all connectivity, using Steiner-tree based implementation for Fermionic gates from \cite{miller2024treespilation} followed by Qiskit transpiler optimization level 1 with the basis gate set of generic single-qubit rotation gate and CNOT. The resulting circuits have 616, 3470 and 13,676 CNOTs, respectively, for 28-, 40- and 52-mode circuits.

The code snippets for obtaining the results are presented below.  For Qiskit, the quantum circuit is a \texttt{QuantumCircuit} and the Hamiltonian is a \texttt{SparsePauliOp} object. For CUDA-Q, the quantum circuit is a \texttt{kernel} and the Hamiltonian is a \texttt{SpinOperator} object.

\begin{center}
\begin{tcolorbox}[width=0.75\textwidth, 
  height=1.5cm,  colbacktitle=gray!20, boxrule=0.1mm,colframe=black, coltitle=black, colback=white, title=Qiskit SV]
\begin{lstlisting}[basicstyle=\small\ttfamily, aboveskip=-5pt, belowskip=0pt]
sv = Statevector.from_instruction(qc)
ev = sv.expectation_value(hamiltonian_qiskit)
\end{lstlisting}
\end{tcolorbox}
\vspace{0.1in}

\begin{tcolorbox}[width=0.75\textwidth, 
  height=1.5cm,  colbacktitle=gray!20, boxrule=0.1mm,colframe=black, coltitle=black, colback=white, title=CUDA-Q SV]
\begin{lstlisting}[basicstyle=\small\ttfamily, aboveskip=-5pt, belowskip=0pt]
cudaq.set_target("nvidia", option="fp64")
energy = estimate_energy_parametrized_kernel(kernel, operator)
\end{lstlisting}
\end{tcolorbox}
\vspace{0.1in}

\begin{tcolorbox}[width=0.75\textwidth, 
  height=1.5cm,  colbacktitle=gray!20, boxrule=0.1mm,colframe=black, coltitle=black, colback=white, title=CUDA-Q MPS]
\begin{lstlisting}[basicstyle=\small\ttfamily, aboveskip=-5pt, belowskip=0pt]
cudaq.set_target("tensornet-mps")
energy = estimate_energy_parametrized_kernel(kernel, operator)
\end{lstlisting}
\end{tcolorbox}
\end{center}

\section{Upper bound to state error}\label{app:state_error}

Consider the so-called spectral gap of the Hamiltonian, that is, the difference between its second lowest eigenvalue $E_1$ and its lowest eigenvalue $E_0$. 
Without loss of generality, let us further assume a given quantum state $\ket{\psi}$ to have a relative error $p= \frac{\bra \psi H \ket \psi - E_0}{E_1-E_0}\leq 1$.
Expressing $\ket \psi$ in the eigenbasis of the Hamiltonian $\ket{\psi} = \sum_k \alpha_k \ket{E_k}$, we can define the following bound, 
\begin{equation}
\begin{split}
(E_1-E_0)p & = \bra \psi H \ket \psi - E_0 = \sum_{k,m} \alpha_k^\ast \alpha_m E_m \bra {E_k} \ket {E_m} - E_0 = \sum_k |\alpha _k |^2E_k - E_0 \\
& \geq (1-|\alpha_0|^2) E_1 + |\alpha_0|^2 E_0 - E_0  = (E_1-E_0) (1 - |\alpha_0|^2).
\end{split}
\label{eq:state-err}
\end{equation}
Assuming a non-degenerate ground state ($E_1 > E_0$), comparing the left-hand side with the right-hand side of Eq.~\eqref{eq:state-err}\ yields a lower bound for the overlap, namely $|\alpha| \geq \sqrt{1-p}$ and, consequently, we can define the upper bound to the state error as  
\begin{equation}
\text{upper bound to state error} \equiv 1-|\alpha_0| \leq  1- \sqrt{1-p}\ .
\end{equation}
Hence, by knowing the relative to spectral gap error of a Hamiltonian, we can compute an upper bound to the overlap error between the obtained state $\ket{\psi}$ and the target state. For example, in the case of the ground state of TLD1433 with a 52-mode active space, optimizing the quantum circuit to an absolute error of 10 mHa error gives a relative error of $p\approx 5\%$, and thus state error $2.53\%$. We obtain $E_1$\ and $E_0$ from state-specific DMRG calculations within the various active orbital spaces.

In order to compute the upper bound, we used the spectral gap over the singlet subspace. For this to be valid it is necessary that the given quantum state $\ket{\psi}$ has a higher overlap with the singlet subspace. To show this, we computed the values of number, $S_z$ and $S^2$ symmetries for the final states from Fig.~\ref{fig:adapt_sims} . We estimated that the difference between computed and expected values for number operators are below $2 \cdot 10^{-4}$ and for $S_z$ operators are below $7\cdot 10^{-6}$. For the $S^2$ operators the values for 28, 40 and 52 qubits are respectively $1.66\cdot 10^{-2}$, $4.21 \cdot 10^{-3}$ and $9.44 \cdot 10^{-3}$, which is much smaller than the spectral gap 0.75 for $S^2$. This shows it is justified for our results to use the formula above based on singlet-subspace spectral gap.
\end{document}